\documentclass[11pt]{article}
\usepackage{fullpage,amsfonts,amsmath,amsthm,amssymb,mathtools,mathrsfs,graphicx,upgreek}
\usepackage[margin=1in,marginparwidth=0.75in]{geometry}

\usepackage{algorithm}
\usepackage[noend]{algpseudocode}

\usepackage{epsfig}
\usepackage{bm}
\usepackage{comment}

\usepackage{tablefootnote}

\usepackage{hyperref}
\usepackage{cleveref}
\crefname{subsection}{subsection}{subsections}

\usepackage{enumitem}
\newcommand{\subscript}[2]{$#1 _ #2$}
\numberwithin{equation}{section}

\newtheorem{theorem}{Theorem}[section]
\newtheorem{proposition}[theorem]{Proposition}
\newtheorem{corollary}[theorem]{Corollary}

\newtheorem{lemma}[theorem]{Lemma}
\newtheorem{example}[theorem]{Example}

\theoremstyle{definition}
\newtheorem{definition}{Definition}

\newtheorem{remark}[theorem]{Remark}

\usepackage{bbm}

\usepackage{algpseudocode}
\usepackage{datetime}

\usepackage[english]{babel}
\usepackage[autostyle, english = american]{csquotes}
\MakeOuterQuote{"}

\newcommand{\sta}{\text{st}}
\newcommand{\scr}{\mathcal}
\newcommand{\mb}{\mathbb}
\newcommand{\til}{\widetilde}

\newcommand{\eps}{\varepsilon}

\newcommand{\val}{\text{val}}

\newcommand{\OPT}{\text{OPT}}
\newcommand{\nOPT}{\text{OPT}_\text{n-adp}}
\newcommand{\rOPT}{\text{OPT}_{\text{rel}}}
\newcommand{\LPOPT}{\text{LPOPT}}
\newcommand{\poly}{\text{poly}}
\newcommand{\typ}{\text{typ}}
\newcommand{\qLPOPT}{\text{LPOPT}_{\text{QC}}}

\newcommand{\dLPOPT}{\LPOPT_{\text{DP}}}

\newcommand{\sLPOPT}{\LPOPT_{\text{std}}}
\newcommand{\Ber}{\textup{Ber}}

\newcommand{\citet}{\cite}
\newcommand{\citep}{\cite}

\makeatletter
\def\keywords{\xdef\@thefnmark{}\@footnotetext}
\makeatother

\usepackage[usenames,dvipsnames]{xcolor}

\newcommand{\CM}[1]{\textcolor{black}{#1}}

\title{Online Bipartite Matching in the Probe-Commit Model\footnote{Journal paper is based on the following two conference papers: ``Prophet Matching in the Probe-Commit Model'' \cite{bor_prophets} and ``Secretary Matching Meets Probing with Commitment'' \cite{borodin2021secretary}}
}

\author{
Allan Borodin
\thanks{Department of Computer Science, University of Toronto, Toronto, Canada.
\texttt{bor@cs.toronto.edu}}
\and
Calum MacRury \thanks{Graduate School of Business, Columbia University, New York, USA. 
	\texttt{cm4379@columbia.edu}}
}

\date{}

\begin{document}

\maketitle

\begin{abstract}
We consider the classical online bipartite matching problem in the probe-commit model. In this problem, when an online vertex arrives, its edges must be probed to determine if they exist, based on known edge probabilities. A probing algorithm must respect commitment, meaning that if a probed edge exists, it must be used in the matching. Additionally, each online vertex has a patience constraint which limits the number of probes that can be made to its adjacent edges. We introduce a new configuration linear program and use it to establish the following competitive ratios which depend on the model used to generate the instance graph, and the arrival order of its online vertices:
\begin{itemize}
	\item In the worst-case instance model, an optimal $1/e$ ratio when the vertices arrive in uniformly at random (u.a.r.) order.
	\item In the known independently distributed (i.d.) instance model, an optimal $1/2$ ratio when the vertices arrive in adversarial order, and a $1-1/e$ ratio when the vertices arrive in u.a.r. order.
\end{itemize}
The latter two results improve upon the previous best competitive ratio of $0.46$ due to Brubach et al. (Algorithmica 2020). Our $1-1/e$-competitive algorithm matches the best known result for the prophet secretary matching problem due to Ehsani et al. (SODA 2018). Our algorithm is efficient and implies a $1-1/e$ approximation ratio for the special case when the graph is known. This is the offline stochastic matching problem, and we improve upon the 0.42 approximation ratio for one-sided patience due to Pollner et al. (EC 2022), while also generalizing the $1-1/e$ approximation ratio for unbounded patience due to Gamlath et al. (SODA 2019).

\end{abstract}

\newpage
\tableofcontents
\newpage



\section{Introduction} \label{sec:intro}

Stochastic probing problems are part of the larger area of decision making under uncertainty and more specifically, stochastic optimization. Unlike more standard forms of stochastic optimization, it is not just that there is some stochastic uncertainty in the set of inputs, stochastic probing problems involve inputs which cannot be determined without probing (at some cost and/or within some constraint). Applications of stochastic probing occur naturally in many settings, such as in matching problems where compatibility cannot be determined without some trial or investigation (for example, in online dating, online advertising, and kidney exchange applications).
There is by now an extensive literature for stochastic probing problems. 


While we are only considering ``one-sided online bipartite matching'', offline stochastic matching was first considered in the context of a general graph by Chen et al. \cite{Chen}. In this problem, \textit{the probing algorithm} is presented a {\it stochastic graph} $G = (V, E)$ as input, which has a probability $p_e$ associated with each edge $e$ and a \textit{patience} (or time-out) parameter $\ell_v$ associated with each vertex $v$. An algorithm \textit{probes} edges in $E$ in some adaptive order within the constraint that at most $\ell_v$ edges are probed incident to any particular vertex $v$. When an edge $e$ is probed, it is guaranteed to exist with probability exactly $p_e$. If an edge $(u,v)$ is found to exist, then the algorithm must \textit{commit} to the edge -- that is, it must be added to the current matching. The goal is to maximize the expected size of a matching constructed in this way. Chen et al. showed that by greedily probing edges in  non-increasing order of edge probability,
one attains an approximation ratio of $1/4$ against an \textit{optimal offline probing algorithm} (we provide
a precise definition in \Cref{sec:benchmark}). The analysis was later improved by Adamczyk \cite{Adamczyk11}, who showed that this greedy algorithm in fact attains an approximation ratio of $1/2$. 

In addition to generalizing the results of Chen et al. to edge weights, Bansal et al. \cite{BansalGLMNR12} introduced a bipartite online variant of the (offline) stochastic matching problem called the \textit{online stochastic matching problem with known i.i.d. arrivals}.
In this problem, a single seller wishes to match
their offline (indivisible) \textit{items} to (unit-demand) \textit{buyers} which arrive online one by one. The seller knows how many buyers will arrive,
and the possible \textit{type/profile} of each online buyer, which is specified by edge probabilities, edge weights and a patience parameter. Here the edge probability models the likelihood a buyer type will purchase an item if the seller presents it to them and an edge weight represents the revenue the seller will gain from making such a sale successfully.
The patience of a buyer type indicates the maximum number of items they are willing to be shown. Note that
there are no restrictions on how many buyers an item may be shown to. The online buyers are drawn \textit{i.i.d.} from a \textit{known distribution}, where the type of each online buyer is presented to the seller upon its arrival.  The (potential) sale of an item to an online buyer must be made before the next online buyer arrives, and the seller's goal is to maximize their expected revenue. As in the Chen et al. model, the seller must \textit{commit} to the first sale to which an online buyer agrees. Fata et al. observed that this problem is closely related to the
multi-customer \textit{assortment optimization problem}, which has numerous practical applications in revenue management (see \cite{Fata2019MultiStageAM} for details).

Mehta and Panigrahi \cite{MehtaP12} introduced the \textit{online matching problem with stochastic rewards},
which studies the original online bipartite matching problem of Karp, Vazirani and Vazirani \cite{KarpVV90} in
the probe-commit model. Here the algorithm initially only knows the offline vertices $U$ of the stochastic graph $G=(U,V,E)$, and its online vertices $V$ are determined by an adversary and presented to the algorithm one by one. Moreover, they consider when $G$ has \textit{unit patience values}, that is, each online vertex can probe at most one of its adjacent edges. They considered the special case of uniform edge probabilities (i.e, $p_{e} =p$ for all $e \in E$) and proved a competitive ratio $\frac{1}{2}(1 + (1-p)^{2/p})$ (this limits to $\frac{1}{2}(1+e^{-2}) \approx .567$ as $p \rightarrow 0$).
Mehta et al. \cite{Mehta2015} considered arbitrary edge probabilities, and
attained a competitive ratio of $0.534$, and recently, Huang and Zhang \cite{huang2020online} additionally handled the case of arbitrary offline vertex weights, while improving this ratio to $0.572$.
However, as in \cite{MehtaP12}, both \cite{Mehta2015} and  \cite{huang2020online} require edge probabilities which are \textit{vanishingly small}\footnote{Vanishingly small edge probabilities must satisfy $\max_{e \in E} p_{e} \rightarrow 0$,
	where the asymptotics are with respect to the size of $G$.}. Goyal and Udwani \cite{Goyal2020OnlineMW} improved on both of these works by proving a $0.596$ competitive ratio in the same setting. They also showed that $1-1/e$ is attainable for the separate special case of \textit{vertex-decomposable}\footnote{Vertex-decomposable means that there exists probabilities $(p_{u})_{u \in U}$ and $(p_{v})_{v \in V}$, such that $p_{(u,v)}=p_{u} \cdot p_{v}$ for each $(u,v) \in E$.} edge probabilities. It remains open whether $1-1/e$ is attainable for arbitrary edge probabilities. Goyal and Udwani discuss the difficulty of this problem in the context of the Adwords problem with an arbitrary budget to bid ratio.

Clearly, any classical online matching problem can be generalized to the probe-commit model. Given such a problem, we can ask if the optimal competitive ratio is the same when probing is not required. When the optimal competitive ratio is not known (even in the classical setting), we can still ask whether there exists an online probing algorithm whose competitive ratio is equal
to the best known competitive ratio when probing is not required. We provide a number of positive answers to questions of this form. In particular, we generalize the problem of Bansal et al.
to \textit{known i.d. arrivals}. Specifically, each online buyer is drawn from a (potentially) distinct distribution, and the draws are done independently. When online buyers arrive adversarially, we generalize the \textit{prophet inequality matching problem} of Alaei et al. \cite{alaei_2012}. When online buyers arrive in random order, we generalize the \textit{prophet secretary matching problem} of Ehsani et al. \cite{Ehsani2017}. We prove a $1/2$ competitive ratio for adversarial arrivals
and a $1-1/e$ competitive ratio for random order arrivals.  These competitive ratios match the best known results when probing is not required (see \cite{alaei_2012,Ehsani2017}), and the $1/2$ result is in fact tight. Note that the arrival order does not matter for the case of identical distributions, and so our $1-1/e$ result implies a $1-1/e$ competitive ratio for known i.i.d. arrivals.  Up until very recently, this result also matched the best known competitive ratio when probing is not required due to Manshadi et al. \cite{ManshadiGS12}.  Yan \cite{Yan2022} improved on \cite{ManshadiGS12} and showed that a competitive ratio
of $0.645 > 1-1/e$ is attainable for known i.i.d. arrivals when probing is not required. We are also the first to study the online stochastic matching problem with \textit{worst-case random order arrivals}. Here the stochastic graph is adversarially generated, and its online vertices arrive in random order. When the graph is edge-weighted,
we generalize the \textit{secretary matching problem} \cite{KRTV2013} to the probe-commit model,
and prove a competitive ratio of $1/e$, which is exactly the optimal competitive ratio
when probing is not required. 

All of the above competitive ratios in fact hold in a more general probing model,
where for each $v \in V$ the patience value $\ell_v$ is generalized to a downward-closed \textit{online probing constraint} $\scr{C}_v$,
which specifies which sequences of edges adjacent to $v$ may be probed. For instance, this includes when $v$ has a \textit{budget} and \textit{edge probing costs} (i.e., $\scr{C}_v$ is a \textit{knapsack constraint}). Probing constraints of a similar nature were originally considered by Gupta and Nagarajan \cite{GuptaN13} for a wide range of offline stochastic probing problems. For simplicity, we defer the precise statement of our  probing model to \Cref{sec:general_probing}, and first introduce the online stochastic matching problem restricted to patience values.

\section{Preliminaries}

An instance of the \textbf{online stochastic matching problem} is a \textbf{stochastic graph} specified in the following way. Let $G=(U,V,E)$ be a bipartite graph with edge weights $(w_{e})_{e \in E}$ and edge probabilities $(p_{e})_{e \in E}$, where $\partial(r)$ denotes the edges of $G$ which include $r$ for $r \in U \cup V$.
Each $e \in E$ is \textbf{active} independently with probability $p_e$,
where  the \textbf{edge state} $\sta(e) \sim \Ber(p_e)$ is the indicator random variable
for this event. In addition, each \textbf{online vertex} $v \in V$ has an integer \textbf{patience parameter} $\ell_v \ge 1$. 
We denote $n:= |V|$ to be the number of online vertices of $G$.

An \textbf{online probing algorithm} begins with limited information
regarding $G$. Specifically, it knows $U$, the \textbf{offline vertices} of $G$, and
in all but one of the settings we study, it also knows the value of $n$. An ordering on $V$ is then generated either by an \textbf{adversary} or \textbf{uniformly at random}, independent of all other randomization. We refer to the
former case as the \textbf{adversarial order model (AOM)} and the latter case as the \textbf{random order model (ROM)}.



In the \textbf{worst-case instance model}, the stochastic graph $G$ is generated by an adversary.
Based on whichever ordering is generated on $V$,
the online vertices are then presented to the online probing algorithm one by one.
When an online node $v \in V$ is presented (arrives), the online probing algorithm learns $(p_{e})_{e \in \partial(v)}$
and $(w_{e})_{e \in \partial(v)}$, however, the edge states $(\sta(e))_{e \in \partial(v)}$ initially
remain hidden to the algorithm. Instead, the algorithm also learns the patience value
$\ell_v$ of $v$, and it is allowed to \textit{adaptively} \textbf{probe} at most $\ell_v$ edges of $\partial(v)$. 
Here a probe to an edge $e \in \partial(v)$  reveals the \textit{instantiation} of $\sta(e)$
to the algorithm. The algorithm 
operates in the \textbf{probe-commit model}, in which there is a \textbf{commitment} requirement upon probing an edge:
Specifically, if an edge $e = (u,v)$ is probed and turns out to be active, then the online probing algorithm must make an irrevocable decision as to whether or not to include $e$ in its matching, prior to probing any subsequent edges. This definition of commitment is the one considered by Gupta et al. \cite{GuptaNS16}, and is slightly different but equivalent to the Chen et al. \cite{Chen} model in which an active edge must be immediately accepted into the matching. As in the classical bipartite matching problem, an online probing algorithm must decide on a possible match for an online node $v$ before seeing the next online node. The algorithm
returns a matching $\scr{M}$ of (probed) active edges, and its goal is to maximize $\mb{E}[w(\scr{M})]$,
where $w(\scr{M}):= \sum_{e \in \scr{M}} w_e$ is the \textbf{weight} of $\scr{M}$, and the expectation is over $(\sta(e))_{e \in E}$ and any random decisions of the algorithm (as well as the ordering on $V$ in the ROM). We will consider online algorithms which know the value of $n = |V|$, as well as those which do not. By setting edge probabilities to $0$, we hereby assume $E = U \times V$ without loss in generality.

In the \textbf{known i.d. instance model}, the algorithm again executes on an unknown stochastic graph $G=(U,V,E)$,
however the online vertices $V$ and edges $E$ of $G$ are instead generated through the following stochastic process: Let $H_{\typ} = (U,B,F)$ be a
\textbf{type graph}, which is a bipartite graph with edge weights $(w_{f})_{f \in F}$, edge probabilities $(p_{f})_{f \in F}$, and patience values $(\ell_b)_{b \in B}$. We refer to $B$ as the \textbf{type nodes} of $H_{\typ}$, as 
$H_{\typ}$ is known to the algorithm and these represent the possible online vertices that $G$ may have.  In addition, the algorithm knows that $n=|V|$ online vertices will arrive,
and it is presented distributions $(\scr{D}_{i})_{i=1}^{n}$ supported on $B$. For $i =1, \ldots, n$, online vertex $v_{i}$ is drawn independently from $\scr{D}_{i}$, and $V$ 
is now defined to be the multiset including $v_1, \ldots ,v_n$. The online vertices $V$ are once again presented to the algorithm either in adversarial or random order, and processed in the same way as in the
worst-case instance model, with the caveat that the online algorithm additionally learns the distribution
the online vertex was drawn from  upon its arrival. We denote $G \sim (H_{\typ}, (\scr{D}_i)_{i=1}^{n})$ to indicate
that $G$ is drawn from the \textbf{known i.d. input} $(H_{\typ}, (\scr{D}_{i})_{i=1}^{n})$. The algorithm returns
a matching $\scr{M}$ of (probed) active edges of $G$, and its goal is to maximize $\mb{E}[ w(\scr{M})]$,
where the expectation is now also over the additional randomness in $G \sim (H_{\typ}, (\scr{D}_i)_{i=1}^{n})$.
Note that in the AOM, we assume that the ordering is generated by an \textbf{oblivious adversary}.
This means that the ordering is a permutation $\pi$ of $[n] = \{1, \ldots ,n\}$ which depends solely on
$(H_{\typ}, (\scr{D}_i)_{i=1}^n)$. The vertices $v_1, \ldots ,v_n$ then arrive in order
$v_{\pi(1)}, \ldots ,v_{\pi(n)}$. We again hereby assume that $F= U \times B$.

\subsection{Benchmark} \label{sec:benchmark}

It is easy to see that even when the edges are unweighted and the algorithms initially knows the stochastic graph we cannot hope to obtain a non-trivial competitive ratio against the
expected size of an optimal matching of the stochastic graph. Consider a stochastic graph with a single online vertex with patience $1$, and $k \ge 1$ offline (unweighted) vertices where each edge $e$ has probability $\frac{1}{k}$ of being active. The expectation of an online probing algorithm will be at most $\frac{1}{k}$ while the expected size of an optimal matching will be $1 - (1-\frac{1}{k})^k \rightarrow 1 -\frac{1}{e}$ as $k \rightarrow \infty$.  

The standard approach in the literature is to instead consider the \textbf{offline stochastic matching problem}
and benchmark against an \textit{optimal offline probing algorithm} \cite{BansalGLMNR12,Adamczyk15,BrubachSSX16,BrubachSSX20}.
An \textbf{offline probing algorithm} is given the stochastic graph $G=(U,V,E)$, but initially the edge states $(\sta(e))_{e \in E}$ are hidden. Its goal is
to construct a matching of (probed) active edges of $G$ with weight as large as possible in expectation. It can adaptively probe the edges of $E$ in any order, potentially interleaving edges between distinct vertices. For instance, it may probe $e_1 \in \partial(v)$ followed by $e_2 \in \partial(v')$ and then $e_3 \in \partial(v)$ for distinct $v, v' \in V$. However,  at most $\ell_v$ edges of $\partial(v)$ may be probed for each $v \in V$,
and it must operate in the same probe-commit model as an online probing algorithm. We define
the \textbf{(offline) adaptive benchmark}, denoted $\OPT$, to be an optimal offline probing algorithm,
and define $\OPT(G)$ to be the expected weight of the matching returned by $\OPT$ when it executes on $G$. An alternative weaker benchmark used by Brubach et al. \cite{brubach2021follow,brubach2021conf} is the \textbf{online adaptive benchmark}. This is defined as
an optimal offline probing algorithm which executes on $G$ and whose edge probes respect
some adaptively chosen vertex ordering on $V$. Equivalently, the edge probes
involving each $v \in V$ occur contiguously: If $e_2=(u,v') \in E$ is probed after $e_1 = (u,v)$ for $v' \neq v$, then
no edge of $\partial(v)$ is probed following $e_2$. 

In this work, we focus exclusively on the offline adaptive benchmark.
In the worst-case instance model, we benchmark against $\OPT(G)$ for a worst-case stochastic graph $G$. In the known i.d. instance model, we benchmark against $\OPT(H_{\typ}, (\scr{D}_i)_{i=1}^{n}):=\mb{E}[\OPT(G)]$, where the expectation is over the randomness in drawing $G$ from a worst-case known i.d. input $(H_{\typ}, (\scr{D}_i)_{i=1}^{n})$. 
In either instance model, we can benchmark against a restricted sub-class of instances. This is relevant
to us in the worst-case instance model, where in one case we benchmark
against \textbf{(offline) vertex-weighted} stochastic graphs (i.e., there exists $(w_u)_{u \in U}$
such that $w_{u,v} = w_u$ for all $(u,v) \in E$).

Even assuming the full generality of edge-weighted stochastic graphs,
we are still left with a wide range of problems, depending on if we work in the worst-case instance model or the known i.d. instance model, as well as whether we assume adversarial or random order arrivals. We refer to each problem as the \textbf{online stochastic matching
	problem} with a \textbf{worst-case (known i.d.) instance} and \textbf{adversarial (random order) arrivals}. If
we restrict to a sub-class of stochastic graphs, then we indicate this in the middle of the problem name. For instance, when $G$ is restricted to vertex weights, we refer to the problem as the online stochastic matching problem with a worst-case \textit{vertex-weighted} instance
and random order arrivals.

\section{Contributions and Techniques}

We first summarize the competitive ratios of our algorithms in \Cref{tab:my_label}.
We then discuss each result individually and explain its significance.
Afterwards, in \Cref{sec:general_probing}, we describe the general probing model from \Cref{sec:intro}.
All our algorithms are efficient in this model, as we prove in \Cref{sec:algorithm_efficiency}.
With the exception of our vertex-weighted algorithm, all our algorithms can be extended to this model
without a loss in competitiveness. 
In \Cref{sec:overview_of_techniques}, we give an overview of the techniques used
in the paper.

\begin{table}[H]
	\centering
	\begin{tabular}{c|c|c} \hline
		Competitive Ratios & AOM & ROM \\ 
		
		\hline
		Known I.D. Instance & $? \to  \ge \textbf{1/2}$ [\S\ref{sec:known_id_additions}] & $? \to \textbf{1-1/e}$ [\S\ref{sec:known_id_additions}] \\
		& $\le 1/2$ \cite{krengel_prophet} & $\le 0.703$ \cite{huang2022}\\
		\hline
		Known I.I.D. Instance & $\ge 0.46$ \cite{BrubachSSX20} $\to \ge\textbf{1-1/e}$ [\S\ref{sec:known_id_additions}] & $\ge0.46$ \cite{BrubachSSX20} $\to\ \ge\textbf{1-1/e}$ [\S\ref{sec:known_id_additions}] \\
		&  $\le 0.703$ \cite{huang2022} & $\le 0.703$ \cite{huang2022}\\
		\hline
		Worst-case & $-$ & $? \to\ \ge\textbf{1/e}$ [\S\ref{sec:edge_weights}]\\
		Edge-weighted Instance &   & $\le 1/e$ \cite{lindley61} \\
		\hline
		Worst-case & $\ge 1/2$ \ \cite{Brubach2019} & $1/2 \to\ \ge\textbf{1 - 1/e}$ \tablefootnote{This competitive ratio holds when $G$ is unweighted, and for certain
			cases when $G$ is vertex-weighted. We defer the precise statement of the result to \Cref{thm:ROM_rankable} of \Cref{sec:vertex_weights}.} [\S\ref{sec:vertex_weights}]\\
		Vertex-weighted Instance 	&  $\le 1 -1/e$ \cite{KarpVV90} & $\le 0.826 $ \cite{ManshadiGS12}\\
		\hline
	\end{tabular}
	\caption{New competitive ratios are \textbf{bolded}.  "$\ge$" refers to lower bounds on the \textit{optimal} competitive ratio (algorithmic results), "$\le$" refers to upper bounds (impossibility/hardness results), and arrows indicate improvement from the state of the art. ``$-$" indicates that no constant competitive ratio is attainable, and ``$?$" means that no previous competitive ratio was known.}
	\label{tab:my_label}
\end{table}

\textbf{Known I.D. Instance.} 
We first consider the setting when the algorithm is given a type graph $H_{\typ} = (U,B,F)$
and distributions $(\scr{D}_i)_{i=1}^n$, and executes on $G \sim (H_{\typ}, (\scr{D}_i)_{i=1}^n)$. Observe that if $p_{e} \in \{0,1\}$ for each $e \in F$ of $H_{\typ}=(U,B,F)$, then probing is unnecessary, and $\mb{E}[\OPT(G)]$ corresponds to the expected weight of the maximum matching of $G$. In this special case, the online algorithm also does not need to probe edges, and so the problem generalizes either the \textbf{prophet inequality matching problem} or the \textbf{prophet secretary matching problem}, depending on if we work with adversarial arrivals or
random order arrivals, respectively.

\textbf{Adversarial Order Arrivals.} In the AOM, we attain a competitive ratio of $1/2$. This is a tight bound since the problem generalizes the classical single item prophet inequality for which $1/2$ is an optimal competitive ratio \cite{krengel_prophet}. 
Brubach et al. \cite{brubach2021conf} independently proved the same competitive ratio against the weaker \textit{online} adaptive benchmark (see the definition in \Cref{sec:benchmark}). Our results are incomparable, as their result can be applied to an \textit{unknown} patience framework (at a loss in competitive ratio), whereas our result holds against a stronger benchmark.
Note that our result also holds assuming \textit{known} downward-closed online probing constraints, as we discuss in \Cref{sec:general_probing}.

\textbf{Random Order Arrivals.}
In the ROM, we attain a competitive ratio of $1-1/e$. Interestingly, $1-1/e$ remains the best known competitive ratio in the prophet secretary matching problem due to Ehsani et al. \cite{Ehsani2017}, despite significant progress in the case of a single
offline node (see \cite{AzarCK18,CorreaSZ19}). Huang et al. \cite{huang2022} very recently proved a $0.703$ hardness result for multiple offline nodes and known i.i.d. arrivals. 

\textbf{Known I.I.D. Instance.} 
The special case of identical distributions has been studied in multiple works \cite{BansalGLMNR12,Adamczyk15,BrubachSSX16,BrubachSSX20}, beginning with the $0.12$ competitive ratio of Bansal et al. \cite{BansalGLMNR12}, and which the previous best competitive ratio of $0.46$ is due to Brubach et al. \cite{BrubachSSX20}.
Fata et al. \cite{Fata2019MultiStageAM} improved this competitive ratio to $0.51$ for the special case of \textbf{unbounded patience}
(i.e., $\ell_b = \infty$ for all $b \in B$).  All of these previous competitive ratios (including ours) are proven against the offline adaptive benchmark. Our $1-1/e$ competitive ratio for a known i.d. instance and random order arrivals improves on both of these bounds, while simultaneously applying to non-identical distributions. Note that Brubach et al. \cite{brubach2021follow,brubach2021conf} independently achieved a $1-1/e$ competitive ratio for a known i.i.d. instance, however their ratio is again proven against the weaker online adaptive benchmark. Our $1-1/e$ result matches the previously best known competitive ratio when probing is not required due to Manshadi et al. \cite{ManshadiGS12}. Very recently, Yan \cite{Yan2022} improved on \cite{ManshadiGS12} and showed that $0.645 > 1-1/e$ is attainable when probing is not required.

\textbf{Known Stochastic Graph Instance.} An important case of the online stochastic matching problem with a known i.d. instance is the
case of a \textbf{known stochastic graph}.  In this setting, $H_{\typ}=(V,B,F)$ satisfies $n= |B|$,
and the distributions $(\scr{D}_i)_{i=1}^{n}$ are all point-mass on \textit{distinct} vertices of $B$. Thus, the online
vertices of $G$ are not randomly drawn, and $G$ is instead 
equal to $H_{\typ}$. The online probing algorithm thus knows the stochastic graph $G$ in advance, but remains unaware of the edge
states $(\sta(e))_{e \in E}$, and so it still must sequentially probe the edges to reveal their states. Again, it operates in the probe-commit model, and respects the patience values $(\ell_v)_{v \in V}$ as well as the ordering on $V$. 

As described in \Cref{sec:intro}, the focus of the original \textit{offline} stochastic matching problem is to design efficient offline probing algorithms which attain approximation ratios against the offline adaptive benchmark (see \cite{Chen, Adamczyk11,BansalGLMNR12,Adamczyk15,BavejaBCNSX18,Gamlath2019,brubach2021offline,PollnerRSW22} for a collection of such results). Since all our probing algorithms are efficient, one of the main benefits of working with the offline adaptive benchmark opposed to the \textit{online} adaptive benchmark, is that our competitive ratios imply approximation ratios. In particular,
we get a $1-1/e$ approximation ratio for stochastic graphs with one-sided patience values. For context, $0.426$ is the previously best known
approximation ratio for bipartite graphs with one-sided patience values due to Pollner et al. \cite{PollnerRSW22}. Note that their algorithm has the benefit of working for random order \textit{edge arrivals}, whereas ours requires one-sided random order vertex arrivals.

Gamlath et al. \cite{Gamlath2019} also consider an offline bipartite matching problem for the special case when $G =(U,V,E)$ has unbounded patience, which they refer to as the \textbf{query-commit problem}. Both of our algorithms process $V$ in random order, and attain performance guarantees of $1-1/e$ against very different non-standard LPs (linear programs) -- 
\ref{LP:new} and \ref{LP:full_patience}, respectively. Note that \ref{LP:full_patience} has exponentially many constraints and polynomially many variables, whereas \ref{LP:new} has polynomially many constraints and exponentially many variables (see \Cref{sec:LP_relations} for a statement of \ref{LP:full_patience}). To the best of our knowledge, \ref{LP:full_patience} does not seem to have an extension even to arbitrary patience values,
as it is unclear how to generalize its constraints while maintaining polynomial time solvability. Despite having such different forms, in the unbounded patience setting the LPs take on the same value, as we prove in Proposition \ref{prop:full_patience_equivalence} of \Cref{sec:LP_relations}. Thus, our $1-1/e$
competitive ratio can be viewed as a generalization of their work to arbitrary patience values and more general probing constraints, as well as to known i.d. instances with random order arrivals. For the special case
when $G$ is known and has unbounded patience, Derakhshan and Farhadi \cite{Derakhshan2022} recently proved an approximation
ratio of $1 - 1/e+ \delta$ for $\delta \ge 0.0014$.

Our $1-1/e$ approximation ratio is attained
by a \textbf{non-adaptive} probing algorithm. An offline probing algorithm is non-adaptive on $G$
if the edges it probes (and the order this is done in) can be described as a randomized function of $G$. Equivalently, the edges probed
by the algorithm are determined independently from the edge states $(\sta(e))_{e \in E}$ of $G$.  In \Cref{thm:adaptivity_gap_negative_restatement} of \Cref{sec:adaptivity_gap}, we prove a $1-1/e$ hardness result against $\OPT$ which applies to \textit{all} non-adaptive probing algorithms.
Thus, we show that the \textbf{adaptivity gap} of the offline stochastic matching problem with one-sided patience values is exactly $1-1/e$ (see \Cref{cor:adaptivity_gap} for details). 

\textbf{Edge-weighted Worst-case Instance.} We next consider the online stochastic matching problem 
in the most challenging setting for the algorithm. That is, the stochastic graph $G=(U,V,E)$
has edge weights and is generated by an adversary.
When the edge probabilities of $G$ are binary (i.e., $p_e \in \{0,1\}$ for all $e \in E$),
this is the edge weighted online bipartite matching problem. It is well known that no constant
competitive ratio is attainable for adversarial arrivals, and so we instead consider random order arrivals.
This is the \textbf{secretary matching problem} in which Kesselheim et al. \cite{KRTV2013} proved an
optimal asymptotic competitive ratio of $1/e$ (note that this is optimal even for a single offline vertex, i.e., $|U|=1$). 
We generalize their matching algorithm so as to apply to the stochastic probing setting,
and recover the same asymptotic $1/e$ competitive ratio.

\textbf{Vertex-weighted Worst-case Instance.} 
Purohit et al. \cite{Purohit2019} and Brubach et al. \cite{Brubach2019} both independently 
designed an \textit{efficient} probing strategy, referred to as $\textsc{DP-OPT}$, which matches an online node $v$ to an edge of maximum weight in expectation, no matter the patience value $\ell_v$. Using  $\textsc{DP-OPT}$, Brubach et al. considered a deterministic \textit{greedy} probing algorithm when $G$ has offline vertex weights $(w_u)_{u \in U}$ and is a worst-case instance. In this setting, they prove a competitive ratio of $1/2$ in the AOM,
and this is best possible amongst \textit{deterministic} online probing algorithms.

When $G$ is unweighted, this greedy probing algorithm is very simple as it probes the available edges of an arriving vertex $v \in V$ based on non-increasing edge probabilities. In this setting, we prove a competitive ratio of $1-1/e$ for random order arrivals. We then extend this competitive ratio to a number of important
settings when $G$ is vertex-weighted, yet \textbf{rankable} (see \Cref{sec:vertex_weights} for the precise definition). \textbf{Rankability} includes when $G$ is unweighted, as well as the well-studied case when $G$ has unit patience values \cite{MehtaP12,Mehta2015,huang2020online,Goyal2020OnlineMW}. In this latter setting,
each $v \in V$ \textit{ranks} its adjacent edges in non-increasing order of $(w_u \cdot p_{u,v})_{u \in U}$.
Assuming adversarial arrivals, this is the \textit{online matching with stochastic rewards problem}, as described in \Cref{sec:intro}. Mehta and Panigrahi \cite{MehtaP12} showed that $0.621 < 1-\frac{1}{e}$ is a randomized in-approximation for this problem with regard to  guarantees made against \ref{LP:standard_benchmark}. 
Here \ref{LP:standard_benchmark} is the LP used in \cite{MehtaP12} to upper bound/relax the offline adaptive benchmark in the unit patience setting. Our $1-1/e$ result in fact holds against \ref{LP:standard_benchmark} which implies that deterministic probing algorithms
in the ROM have strictly more power than randomized probing algorithms in
the AOM (see \Cref{cor:ROM_unit_patience} for details). This contrasts with the classical online bipartite matching setting
where such a separation is \textit{not} known. We also prove
that the same algorithm attains an \textit{asymptotic} (as $|G| \rightarrow \infty$)
competitive ratio of $1-1/e$ when  $p_v := \max_{e \in \partial(v)}p_e$
satisfies $p_v = o(\ell_v)$ for each $v \in V$. The vanishing probability
setting is similar in spirit to the small bid to budget assumption in the Adwords problem
(see Goyal and Udwani \cite{Goyal2020OnlineMW} for details).

Finally, we generalize $\textsc{DP-OPT}$ to downward-closed probing constraints and extend our competitive ratios
in a number of settings. Note that $\textsc{DP-OPT}$ is crucial for ensuring the efficiency of our other online probing algorithms, as we explain in \Cref{sec:algorithm_efficiency}.

\subsection{Generalizing to Downward Closed Online Probing Constraints} \label{sec:general_probing}
Consider an online vertex $v \in V$ of a stochastic graph $G=(U,V,E)$. 
The patience parameter $\ell_v$ can be viewed as a simple budgetary constraint, where each probe has unit cost and the patience parameter is the budget. A natural generalization is thus to instead consider a \textbf{knapsack} (or \textbf{budgetary}) constraint. That is, a non-negative \textbf{budget} $B_v \ge 0$ and \textbf{edge probing costs} $(c_{e})_{e \in \partial(v)}$, such that any subset $S \subseteq \partial(v)$ may be probed (in any order), provided $\sum_{e \in S} c_{e} \le B_v$. Alternatively,
$\ell_v$ can be viewed as an $\ell_v$-uniform-matroid constraint on $\partial(v)$, and so an incomparable generalization is to an \textit{arbitrary} matroid constraint on $\partial(v)$. Matroid probing constraints
of the latter form were considered in the offline setting in \cite{GuptaN13,adamczyk2016}. We introduce a general online probing framework that encompasses both budgetary and matroid probing constraints.

Define $\partial(v)^{(*)}$ as the collection of strings (tuples) formed from the edges of $\partial(v)$ whose characters (entries) are all distinct. Note that we use string/tuple notation and terminology interchangeably. Each $v \in V$ has an \textbf{(online) probing constraint} $\scr{C}_v \subseteq \partial(v)^{(*)}$ which we assume is \textbf{downward-closed}.
That is, $\scr{C}_v$ has the property that if $\bm{e} \in \scr{C}_v$,
then so is any substring or permutation of $\bm{e}$. 
Clearly our setting encodes the case when $v$ has a patience value $\ell_v$, and more generally, when $\scr{C}_v$ corresponds to a matroid or budgetary constraint.

In order to discuss the efficiency of our algorithms in the generality of our probing constraints, we
work in the \textbf{membership oracle model}. An online probing algorithm may make a
\textbf{membership query} to any string $\bm{e} \in \partial(v)^{(*)}$ for $v \in V$,
thus determining in a single operation whether or not $\bm{e} \in \partial(v)^{(*)}$ is in $\scr{C}_v$. Assuming access to such an oracle, all our algorithms are implementable in polynomial time, as we prove in \Cref{sec:algorithm_efficiency}. When a probing constraint can be stated explicitly, our algorithms are polynomial time in the usual (e.g., Turing machine time) sense. For example, for a budget constraint, the budget $B_v$ and edge probing costs $(c_{e})_{e \in \partial(v)}$ are revealed upon the arrival of the online vertex $v$.

For the remainder of the paper, when discussing a stochastic graph $G=(U,V,E)$, we work in the full generality of
online probing constraints $(\scr{C}_v)_{v \in V}$, unless indicated otherwise.

\subsection{An Overview of Our Techniques} \label{sec:overview_of_techniques}

We first describe our techniques in the known stochastic graph setting (i.e, when the stochastic graph $G=(U,V,E)$
is edge-weighted and known to the online probing algorithm). Afterwards, we explain how our techniques apply when $G$ is unknown and drawn from a known i.d. input or generated by an adversary.
We conclude by describing our techniques when $G$ is vertex-weighted.

The main challenge in designing \textit{efficient} probing algorithms (whether offline or online), is
that in general it is infeasible to efficiently compute the decisions made by the offline adaptive benchmark
(i.e., $\OPT$).
One of the standard approaches in the literature is to instead relax $\OPT$ via an LP. 
When $G$ has patience values $(\ell_v)_{v \in V}$,  Bansal et al. \cite{BansalGLMNR12} introduced an LP which assigns fractional values to the edges of $G$, say $(x_{e})_{e \in E}$, such that $x_e$ upper bounds the probability $e$ is probed by $\OPT$. Clearly, $\sum_{e \in \partial(v)} x_{e} \le \ell_{v}$ is a constraint for each $v \in V$, and 
so by applying a dependent rounding algorithm (such as the GKSP algorithm of Gandhi et al. \cite{GandhiGKSP06}), one can round the values $(x_{e})_{e \in \partial(v)}$ to determine $\ell_{v}$ edges of $\partial(v)$ to probe. By probing these edges in a carefully
chosen order, and matching $v$ to the first edge revealed to be active, one can guarantee that each 
$e \in \partial(v)$ is matched with probability reasonably close to $p_e  x_e$. This is the high-level approach used in many stochastic matching algorithms (for example \cite{BansalGLMNR12,Adamczyk15,BavejaBCNSX18,BrubachSSX20,brubach2021offline,PollnerRSW22}). However, even for a single online node, this LP overestimates the value of the offline adaptive benchmark, and so any algorithm designed in
this way will match certain edges with probability strictly less than $p_e  x_e$. This is problematic,
for the value of the match made to $v$ is ultimately compared to $\sum_{e \in \partial(v)} p_{e} w_{e} x_e$, the contribution
of the variables $(x_e)_{e \in \partial(v)}$ to the LP solution. In fact, Fata et al. \cite{Fata2019MultiStageAM} showed that the ratio between $\OPT(G)$ and an optimum solution to this LP can be as small as $0.51$, so our $1-1/e$ competitive ratio cannot be achieved via a comparison to this LP,
even for the special case when $G$ is known and has patience values.

\textbf{Defining \ref{LP:new}:}
Our approach is to work with a new configuration LP (\ref{LP:new}) which we can state even when $G$
has online probing constraints $(\scr{C}_v)_{v \in V}$. This LP has exponentially many variables which accounts for the many probing strategies available to an arriving vertex $v$ with probing constraint $\scr{C}_v$. For each $\bm{e} \in E^{(*)}$, define $q(\bm{e}) = \prod_{f \in \bm{e}} (1 - p_{f})$,
to be the probability that all the edges of $\bm{e}$ are inactive, where $q(\lambda):=1$ for the empty string $\lambda$. For $f \in \bm{e}$,
we denote $\bm{e}_{< f}$ to be the substring of $\bm{e}$ from its first
edge up to, but not including, $f$. Observe then that $\val(\bm{e}):=\sum_{f \in \bm{e}} w_{f} \cdot p_{f} \cdot q(\bm{e}_{< f})$
corresponds to the expected weight of the first active edge revealed if $\bm{e}$ 
is probed in order of its entries. For each $v \in V$ and $\bm{e} \in \scr{C}_v$, we introduce a decision variable $x_{v}(\bm{e})$
and state the following LP:
\begin{align}\label{LP:new}
	\tag{LP-config}
	&\text{maximize} &  \sum_{v \in V} \sum_{\bm{e} \in \scr{C}_v } \val(\bm{e}) \cdot x_{v}(\bm{e}) \\
	&\text{subject to} & \sum_{v \in V} \sum_{\substack{ \bm{e} \in \scr{C}_v: \\ (u,v) \in \bm{e}}} 
	p_{u,v} \cdot q(\bm{e}_{< (u,v)}) \cdot x_{v}( \bm{e})  \leq 1 && \forall u \in U  \label{eqn:relaxation_efficiency_matching}\\
	&& \sum_{\bm{e} \in \scr{C}_v} x_{v}(\bm{e}) = 1 && \forall v \in V,  \label{eqn:relaxation_efficiency_distribution} \\
	&&x_{v}( \bm{e}) \ge 0 && \forall v \in V, \bm{e} \in \scr{C}_v
\end{align}
We first prove that \ref{LP:new} is a \textbf{relaxation} of the offline adaptive benchmark. Unlike previous LPs used in the literature, we are not aware of an easy proof of this
fact, and so we consider our proof to be a technical contribution.
\begin{theorem}\label{thm:new_LP_relaxation}
	$\OPT(G) \le \LPOPT(G)$.
\end{theorem}
When each $\scr{C}_v$ is downward-closed, \ref{LP:new} can be solved efficiently by using a deterministic separation oracle for the dual of \ref{LP:new}, in conjunction with the ellipsoid algorithm \cite{Groetschel,GartnerM}.
The separation oracle we require is precisely the greedy probing strategy we analyze in the vertex-weighted
setting of \Cref{sec:vertex_weights}. We provide the details of this reduction in \Cref{sec:algorithm_efficiency}.

For the case of patience values, a closely related LP was independently introduced by Brubach et al. \cite{brubach2021follow,brubach2021conf} to design probing algorithms for known i.i.d. instances and known i.d. instances with adversarial arrivals. Their competitive ratios are proven against an optimal solution to this LP.
They then argue that the LP solution relaxes the \textit{online} adaptive benchmark, which
is a weaker statement than \Cref{thm:new_LP_relaxation}. We explain below why this weaker
statement is easier to prove.

\textbf{Proving Theorem \ref{thm:new_LP_relaxation}:} In order to prove \Cref{thm:new_LP_relaxation},
one natural approach is to view $x_{v}(\bm{e})$ as the probability that the offline adaptive benchmark
probes the edges of $\bm{e}$ in order, where $v \in V$ and $\bm{e} \in \scr{C}_v$. 
\CM{Let us first \textit{hypothetically} assume that the following restrictive assumptions
on the offline adaptive benchmark hold without loss of generality (w.l.o.g.):}
\begin{enumerate}[label=(\subscript{P}{{\arabic*}})]
	\item \CM{Suppose $v \in V$ is currently unmatched. If $e=(u,v)$ is probed and $\sta(e)=1$, then $e$ is included in the matching.} \label{eqn:single_vertex_committal}
	\item For each $v \in V$, the edge probes involving $\partial(v)$ are determined independently of the edge states $(\sta(e))_{e \in \partial(v)}$. \label{eqn:single_vertex_non_adaptivity}
\end{enumerate}
\CM{Observe then that \ref{eqn:single_vertex_committal} and \ref{eqn:single_vertex_non_adaptivity} would imply
that the expected weight of the edge assigned to $v$ is
$\sum_{\bm{e} \in \scr{C}_v } \val(\bm{e}) \cdot x_{v}(\bm{e})$.
Moreover, the left-hand side of \eqref{eqn:relaxation_efficiency_matching} 
would correspond to the probability $u \in U$ is matched,
so $(x_{v}(\bm{e}))_{v \in V, \bm{e} \in \scr{C}_v}$
would be a feasible solution to \ref{LP:new}, and so
we could upper bound $\OPT(G)$ by $\LPOPT(G)$. Unfortunately,
while we can assume \ref{eqn:single_vertex_committal} holds w.l.o.g., we cannot simultaneously assume \ref{eqn:single_vertex_non_adaptivity}. This is because the probes involving $v \in V$ will in general depend on $(\sta(e))_{e \in \partial(v)}$. For instance, if $e \in \partial(v)$ is probed and inactive, then perhaps the offline adaptive benchmark next probes $e' =(u,v') \in \partial(v')$ for some $v' \neq v$. If $e'$ is active and thus added to the matching by \ref{eqn:single_vertex_committal}, then the offline adaptive benchmark can never subsequently probe $(u,v)$ without violating \ref{eqn:single_vertex_committal}, as $u$ is now unavailable to be matched to $v$. 
An alternative approach would be to define $x_{v}(\bm{e})$ as the probability that the offline adaptive benchmark probes $\bm{e}$ in order,
\textit{conditioned} on its first $|\bm{e}| -1$ edges being inactive. In this case,
$\OPT(G) = \sum_{v \in V}\sum_{\bm{e} \in \scr{C}_v } \val(\bm{e}) \cdot x_{v}(\bm{e})$; however then $(x_{v}(\bm{e}))_{v \in V, \bm{e} \in \scr{C}_v}$ need not satisfy constraint \eqref{eqn:relaxation_efficiency_distribution}. This is because
for a fixed $v \in V$, $\sum_{\bm{e} \in \scr{C}_v} x_{v}(\bm{e})$ is a sum over
$|\scr{C}_v|$ \textit{conditional} probabilities -- each of which conditions on a \textit{different} event -- and so its value could exceed $1$. Thus, neither of these decision variable interpretations
seems to lead to an easy proof of Theorem \ref{thm:new_LP_relaxation}.}

\CM{Before continuing, we note that when working with the \textit{online} adaptive benchmark of Brubach et al. \cite{brubach2021conf,brubach2021follow}, \ref{eqn:single_vertex_committal} and \ref{eqn:single_vertex_non_adaptivity} \textit{can} be assumed to hold w.l.o.g. As a result of the discussion following \ref{eqn:single_vertex_committal} and \ref{eqn:single_vertex_non_adaptivity}, this explains why Brubach et al. are able to easily argue that $\LPOPT(G)$ upper bounds this weaker benchmark on $G$.}

Returning to the proof of \Cref{thm:new_LP_relaxation}, our solution is to introduce an alternative interpretation of $(x_{v}(\bm{e}))_{v \in V, \bm{e} \in \scr{C}_v}$
based upon a  \textit{relaxation} of the offline stochastic matching problem
we refer to as the \textbf{relaxed stochastic matching problem}. A solution to this problem is a \textbf{relaxed probing algorithm}. A relaxed probing algorithm
operates in the same framework as an offline probing algorithm, yet it returns a one-sided matching of the online vertices which matches each offline node at most once \textit{in expectation}. Observe that if $\rOPT(G)$ is the performance of an \textit{optimal} relaxed probing algorithm, then by definition $\OPT(G) \le \rOPT(G)$. 
Crucially, there exists an optimal relaxed probing algorithm which 
satisfies \ref{eqn:single_vertex_committal} and which is \textbf{non-adaptive}. 
A relaxed probing algorithm is said to be non-adaptive, provided its edge probes are
determined independently of the edge states $(\sta(e))_{e \in E}$.  Non-adaptivity is a much stronger property than \ref{eqn:single_vertex_non_adaptivity}, and so by the above discussion we
are able to conclude that $\rOPT(G) \le \LPOPT(G)$. Since $\OPT(G) \le \rOPT(G)$,
this implies Theorem \ref{thm:new_LP_relaxation}. Proving the existence of an optimal relaxed probing algorithm which is non-adaptive is one of
the most technically challenging parts of the paper, and is the main content of Lemma \ref{lem:non_adaptive_optimum} of
Section \ref{sec:relaxation_adaptive_benchmark}. Note that there may be a simpler proof of Theorem \ref{thm:new_LP_relaxation}, however our relaxed stochastic matching
problem exactly characterizes \ref{LP:new} (i.e., $\rOPT(G)= \LPOPT(G)$), and
so it helps us understand \ref{LP:new}. For instance, in \Cref{sec:LP_relations}, we show that in the unbounded patience setting, \ref{LP:full_patience} of \cite{Gamlath2019} is also characterized by our relaxed matching problem. This implies that the LPs take on the same value, despite having very different formulations in this special setting.

\textbf{Defining the probing algorithms:}
After proving that \ref{LP:new} is a relaxation of the offline adaptive benchmark, we use it to design online probing algorithms. 
Suppose that we are presented a feasible solution, say $(x_{v}(\bm{e}))_{v \in V, \bm{e} \in \scr{C}_v}$, to \ref{LP:new} for $G$. 
For each $e \in E$, define
\begin{equation}
	\til{x}_{e}:= \sum_{\substack{\bm{e}' \in \scr{C}_v: \\ e \in \bm{e}'}} q(\bm{e}_{ < e}') \cdot x_{v}( \bm{e}'). 
\end{equation}
We refer to the values $(\til{x}_{e})_{e \in E}$
as the \textbf{edge variables}
of the solution $(x_{v}(\bm{e}))_{v \in V,\bm{e} \in \scr{C}_v}$.
If we now fix $s \in V$, then we can easily
leverage constraint \eqref{eqn:relaxation_efficiency_distribution} to
design a simple \textit{fixed vertex} probing algorithm which matches
each edge of $e \in \partial(s)$ with probability \textit{exactly} equal to $p_e \til{x}_e$. Specifically, draw $\bm{e}' \in \scr{C}_s$ with probability $x_{s}(\bm{e}')$. If $\bm{e}' = \lambda$, then return the empty set. Otherwise,
set $\bm{e}' = (e_{1}', \ldots ,e_{k}')$ for $k := |\bm{e}'| \ge 1$, and probe the
edges of $\bm{e}'$ in order. Return the first edge which is revealed to be active, if such an
edge exists. Otherwise, return the empty set. We refer to this algorithm as $\mathtt{VertexProbe}$,
and denote its output on the input $(s, \partial(s), (x_{s}(\bm{e}))_{\bm{e} \in \scr{C}_s})$ by
$\mathtt{VertexProbe}(s, \partial(s), (x_{s}(\bm{e}))_{ \bm{e} \in \scr{C}_{s}})$.
\begin{lemma}\label{lem:fixed_vertex_probe}
	For each $e \in \partial(s)$,
	$\mb{P}[ \mathtt{VertexProbe}(s, \partial(s), (x_{s}(\bm{e}))_{ \bm{e} \in \scr{C}_{s}}) =e]  = p_e \til{x}_e$.
\end{lemma}
\begin{remark}
	We can view Lemma \ref{lem:fixed_vertex_probe} as an \textbf{exact rounding guarantee}. The fact that such a guarantee exists, no matter the choice of $\scr{C}_s$, is one of the main benefits of working with \ref{LP:new}, opposed to
	\ref{LP:standard_definition_general} or \ref{LP:full_patience}. As discussed, a solution to \ref{LP:standard_definition_general} provably cannot be rounded exactly in this way. There \textit{does} exist an exact rounding guarantee for
	\ref{LP:full_patience}, however it only applies to the unbounded patience setting (i.e.,  $\scr{C}_v = \partial(s)^{(*)}$), and the procedure is much more complicated than ours (see Theorem \ref{thm:costello_svensson_guarantee} of \Cref{sec:LP_relations} for details).
\end{remark}

\begin{definition}[Propose - Known Stochastic Graph]\label{def:vertex_probe}
	We say that $\mathtt{VertexProbe}$ \textbf{proposes} $s$ to the vertex $u \in N(s)$ provided the algorithm outputs $(u,s)$ when executing on the fixed node $s \in V$. When it is clear that $\mathtt{VertexProbe}$ is being executed on $s$, we say that $s$ proposes to $u$.
\end{definition}
Consider now the following online probing algorithm, where
the ordering $\pi$ on $V$ is either adversarial or u.a.r.
\begin{algorithm}[H]
	\caption{Known Stochastic Graph} 
	\label{alg:known_stochastic_graph}
	\begin{algorithmic}[1]
		\Require a stochastic graph $G=(U,V,E)$.
		\Ensure a matching $\scr{M}$ of active edges of $G$.
		\State $\scr{M} \leftarrow \emptyset$.
		\State Compute an optimal solution of \ref{LP:new} for $G$, say $(x_{v}(\bm{e}))_{v \in V, \bm{e} \in \scr{C}_v}$
		\For{$s \in V$ in order based on $\pi$} 
		\State Set $e \leftarrow \mathtt{VertexProbe}(s, \partial(s), (x_{s}(\bm{e}))_{ \bm{e} \in \scr{C}_{s}})$.
		\If{$e=(u,s)$ for some $u \in U$, and $u$ is unmatched} \Comment{this line ensures $e \neq \emptyset$}
		\State Add $e$ to $\scr{M}$. \label{line:matched_edge}
		\EndIf
		\EndFor
		\State \Return $\scr{M}$.
	\end{algorithmic}
\end{algorithm}
\begin{remark}
	Technically, line \eqref{line:matched_edge}
	should occur within the $\mathtt{VertexProbe}$ subroutine to adhere
	to the probe-commit model, however we express our algorithms in this way for conciseness.
\end{remark}
\textbf{Improvement via online contention resolution:} 
Algorithm \ref{alg:known_stochastic_graph} does not attain a constant competitive ratio when $\pi$ is adversarial, and its competitive ratio is only $1/2$ when $\pi$ is u.a.r.
Thus, we must modify  the algorithm to prove $1/2$ and $1-1/e$ competitive ratios, even in the known stochastic graph setting. When $\pi$ is adversarial, we reduce the problem to designing a $1/2$-\textbf{selectable} \textbf{online contention resolution scheme (OCRS)} for a $1$-uniform matroid \footnote{The independent sets of a $1$-uniform matroid on ground set $S$ are the subsets of $S$ of size at most $1$.}. When $\pi$ is u.a.r., we reduce the problem to designing a $1-1/e$-selectable \textbf{random order contention resolution scheme (RCRS)} for a $1$-uniform matroid. (See \Cref{sec:known_id_additions} for formal definitions). 
In both cases, such contention resolution schemes are known to exist (see \cite{Ezra_2020} for the OCRS,
and \cite{Lee2018} for the RCRS). 
We now provide a high level overview of the reduction for the known stochastic graph case.
Each CRS we use is for a $1$-uniform matroid, so we omit this in the descriptions below.

First observe that $\LPOPT(G) = \sum_{e \in E} w_e p_{e} \til{x}_e$, and $\OPT(G) \le \LPOPT(G)$ by \Cref{thm:new_LP_relaxation}.
Thus, it suffices to design an online probing algorithm which matches each $e \in E$ with probability $1/2 \cdot z_e$ when $\pi$ is generated by an adversary, and $(1-1/e) \cdot z_e$ when $\pi$ is generated u.a.r,
where $z_{e}:= p_{e} \til{x}_e$. Observe that for a fixed $u \in U$,
as \Cref{alg:known_stochastic_graph} executes,
each $v \in V$ proposes to $u$ independently with probability $z_{u,v}$. 
On the other hand, $\sum_{v \in V} z_{u,v} \le 1$ by \eqref{eqn:relaxation_efficiency_matching} of \ref{LP:new}. Let us first assume $\pi$ is adversarial. Observe that if $u$ executes its own $\alpha$-selectable OCRS on ground set $\partial(u)$ as \Cref{alg:known_stochastic_graph}
executes, then each $e \in \partial(u)$ will be matched to $u$ with probability $\alpha \cdot z_e$. By having \textit{each} $u \in U$
concurrently execute its own $\alpha$-selectable OCRS, the resulting probing algorithm will be
$\alpha$-competitive. Thus, an $\alpha$-selectable OCRS can be used to design
an $\alpha$-competitive online probing algorithm in the AOM. Similarly, when $\pi$ is u.a.r., an $\alpha$-selectable RCRS can be used to design an $\alpha$-competitive online probing algorithm
in the ROM.

For random order arrivals, the RCRS based approach simplifies the pricing based approach Gamlath et al. \cite{Gamlath2019} used to attain a competitive ratio of $1-1/e$ in the special case of unbounded patience. This simplified approach was also observed by Fu et al. \cite{Fu_2021} in the context of the Gamlath et al. LP (\ref{LP:full_patience}). They assume unbounded patience in the probe-commit model,
and design a $8/15$-competitive algorithm for general graph random order vertex arrivals. It remains open whether their results can be extended to general patience values. For context, $0.395$ is the best known competitive ratio when allowing for arbitrary patience values and random order edge arrivals \cite{PollnerRSW22} (this was recently improved to $0.474$ in the unbounded patience setting \cite{macrury_contention_2023}).

\subsubsection{Extending to Known I.D. Instances:} Consider when $G$ is unknown and drawn from a known i.d. input $(H_{\typ}, (\scr{D}_{i})_{i=1}^{n})$. In this case, we generalize \ref{LP:new} to a new LP called \ref{LP:new_id}. This LP
departs from previous ones used in the probing literature, as it depends both on the type graph as well as the distributions. For each $i \in [n]$,
we introduce a collection of variables $(x_{i}(\bm{e} \, || \, b))_{\bm{e} \in \scr{C}_b, b \in B}$ associated with the distribution $\scr{D}_i$.  We again reduce to known contention resolution schemes, however the additional variables associated with the possible types
of $v_i \sim \scr{D}_i$ introduce correlated events which must be treated delicately in the context of CRS selectibility. Crucially, the schemes do \textit{not} make use of the type of vertex $v_i$, and so we are able to argue that analogous edge variable lower bounds hold as in the known stochastic graph setting.

\subsubsection{Worst-Case Instances and Random-Order Arrivals.}
Suppose that the adversary chooses a graph $G=(U,V,E)$ whose
online vertices arrive in random order $v_{1}, \ldots ,v_{n}$. Upon receiving the online vertices $V_{t}:=\{v_1, \ldots ,v_{t}\}$, in order
to generalize the matching algorithm of Kesselheim et al. \cite{KRTV2013}, we 
would ideally probe the edges of $\partial(v_t)$ suggested by $\OPT$ on $G_t$, where $G_{t}:=G[U \cup V_t]$
is the induced stochastic graph on $U \cup V_t$. However, since we wish for our algorithms to
be efficient in addition to attaining optimal competitive ratios, this strategy is not feasible. 
We instead solve \ref{LP:new} for $G_t$ to get a solution  $(x_{v}(\bm{e}))_{v \in V_t, \bm{e} \in \scr{C}_{v}}$,
and then execute $\mathtt{VertexProbe}$ on $(v_t, \partial(v_t),(x_{v}(\bm{e}))_{\bm{e} \in \scr{C}_{v_t}})$.
If edge $e_t=(u_t,v_t)$ is returned, and $u_t$ is unmatched, then we add $e_t$ to the current matching. In
\Cref{thm:ROM_edge_weights} of \Cref{sec:edge_weights}, we show
that this online algorithm attains an asymptotic competitive ratio of $1/e$. The analysis follows a similar
proof structure as presented in \cite{KRTV2013}.

\subsubsection{Solving \ref{LP:new} Efficiently and Vertex-weighted Worst-case Instances}
\CM{Let $v \in V$ be a fixed vertex of $G=(U,V,E)$. In \Cref{thm:efficient_star_dp}, we show how to probe $\partial(v)$ in such a way that $v$ is matched to a neighbouring edge $e=(u,v)$ whose weight $w_e$ is as large as possible in expectation. This is precisely the behaviour of $\OPT$ on the \textit{induced stochastic graph} $G[ \{v\} \cup U]$.
While computing such a strategy is computationally challenging for a general graph, in \Cref{thm:efficient_star_dp} we show that it can be done efficiently
for the \textit{star graph} $G[\{v\} \cup U]$. This was first observed by Purohit et al. \cite{Purohit2019} and Brubach et al. \cite{Brubach2019} for the special case when $v$ has patience $\ell_v$. Using a similar dynamic programming (DP) based approach, we generalize their algorithm $\textsc{DP-OPT}$ to apply to an arbitrary downward-closed constraint $\scr{C}_v$. In \Cref{thm:LP_solvability} of \Cref{sec:algorithm_efficiency}, we show that $\textsc{DP-OPT}$ provides a separation oracle for the dual of \ref{LP:DP}, thus leading to an efficient way to solve \ref{LP:new} (despite it having exponentially many variables).}  

\CM{Suppose now that $G=(U,V,E)$ has vertex weights $(w_u)_{u \in U}$. In \Cref{sec:vertex_weights}, we use $\textsc{DP-OPT}$ to design a \textit{greedy} online probing algorithm for $G$. Upon the arrival of $v$, if $R \subseteq U$ denotes the unmatched vertices available, we apply $\textsc{DP-OPT}$ to determine which edges of $G[\{v\} \cup R]$
to probe. We analyze our algorithm using an LP relaxation of $\OPT$ called \ref{LP:DP}, a generalization of an LP introduced by Brubach et al. \cite{Brubach2019} to arbitrary probing constraints $(\scr{C}_v)_{v \in V}$. For both adversarial and random order arrivals, we apply the primal-dual method of Devanur et al. \cite{DJK2013}. As the matching of our algorithm is constructed, we simultaneously construct a (random)  dual solution of \ref{LP:DP}. For adversarial arrivals, the dual solution is feasible in expectation, and twice as large as the expected weight of the algorithm's matching, thus leading to a competitive ratio of $1/2$. When the vertices arrive in random order, we assign dual variables using $g(Y_v)$, where $g(y) := e^{y-1}$ is the function used to analysis the \textsc{Ranking} algorithm in \cite{DJK2013}, and $Y_v \in [0,1]$ is the \textit{random arrival time} of $v$. Interestingly, the \textit{non-monotonic} behaviour of our greedy algorithm prevents the dual-solution from being feasible in expectation, and so we are not able to prove an unconditional competitive ratio of $1-1/e$. Instead, in \Cref{thm:ROM_rankable} we characterize a number of important inputs in which our algorithm is monotonic, thus allowing us to prove a conditional competitive ratio of $1-1/e$.}

\section{Relaxing the Offline Adaptive Benchmark via \ref{LP:new}} \label{sec:relaxation_adaptive_benchmark}

Given a stochastic graph $G=(U,V,E)$, we define the \textbf{relaxed stochastic matching problem}.
A solution to this problem is a \textbf{relaxed probing algorithm} $\scr{A}$,
which operates in the previously described framework of an (offline) probing algorithm. 
That is, $\scr{A}$ is firstly given access to a stochastic graph $G=(U,V,E)$. Initially, the edge states $(\sta(e))_{e \in E}$ are unknown to $\scr{A}$, and $\scr{A}$ must adaptivity
probe these edges to reveal their states, while respecting the downward-closed probing constraints $(\scr{C}_v)_{v \in V}$. As in the offline problem, $\scr{A}$ returns a subset $\scr{M}$ of its active edge probes,
and its goal is to maximize $\mb{E}[w(\scr{M})]$, where $w(\scr{M}):= \sum_{e \in \scr{M}} w_{e}$.
However, unlike before where $\scr{M}$ was required to be a matching of $G$, we relax the required properties of $\scr{M}$:
\begin{enumerate}
	\item Each $v \in V$ appears in at most one edge of $\scr{M}$.
	\item If $N_{u}$ counts the number of edges of $\partial(u)$ which
	are included in $\scr{M}$, then $\mb{E}[N_{u}] \le 1$ for each $u \in U$.
\end{enumerate}
We refer to $\scr{M}$ as a \textbf{one-sided matching} of the online nodes,
and abuse terminology slightly and say that $e \in E$ is matched by $\scr{A}$ if $e \in \scr{M}$.  In constructing $\scr{M}$, $\scr{A}$ must operate in the previously
described probe-commit model. We define the \textbf{relaxed benchmark}
as an optimal relaxed probing algorithm,
and denote its expected value when executing on $G$
by $\rOPT(G)$. Observe that since any offline probing algorithm is a relaxed probing algorithm,
we have that
\begin{equation} \label{eqn:relaxed_benchmark_upper_bound}
	\OPT(G) \le \rOPT(G).
\end{equation}
We say that $\scr{A}$ is \textbf{non-adaptive},
provided its edge probes can be described as a (randomized) function of $G$.
Equivalently, $\scr{A}$ is non-adaptive if the edges probes of $\scr{A}$ are statistically
independent from $(\sta(e))_{e \in E}$.
Unlike for the offline stochastic matching problem,
there exists a relaxed probing algorithm which is both optimal \textit{and} non-adaptive:
\begin{lemma} \label{lem:non_adaptive_optimum}
	For any stochastic graph $G=(U,V,E)$ with downward-closed probing constraints $(\scr{C}_v)_{v \in V}$,
	there exists an optimum relaxed probing algorithm $\scr{B}$ which satisfies the following properties:
	\begin{enumerate}[label=(\subscript{Q}{{\arabic*}})]
		\item If $e=(u,v)$ is probed, $\sta(e)=1$, and $v$ was previously unmatched,
		then $\scr{B}$ matches $e$. \label{eqn:committal}
		\item $\scr{B}$ is non-adaptive on $G$. \label{eqn:non_adaptive}
	\end{enumerate}
\end{lemma}
\begin{remark}
	Note that \ref{eqn:non_adaptive} implies the hypothetical property \ref{eqn:single_vertex_non_adaptivity}, yet
	is much stronger.
\end{remark}
Let us assume that Lemma \ref{lem:non_adaptive_optimum} holds for now.
\begin{proof}[Proof of Theorem \ref{thm:new_LP_relaxation}]
	Consider algorithm $\scr{B}$ of Lemma \ref{lem:non_adaptive_optimum},
	and define $x_{v}(\bm{e})$ to be the probability that $\scr{B}$
	probes the edges of $\bm{e}$ in order for $v\in V$ and $\bm{e} \in \scr{C}_v$.
	Since $\scr{B}$ is a relaxed probing algorithm, we 
	can apply properties \ref{eqn:committal} and \ref{eqn:non_adaptive} 
	to show that $(x_{v}(\bm{e}))_{v \in V, \bm{e} \in \scr{C}_v}$
	is a feasible solution to \ref{LP:new}. Moreover, if $\scr{N}$ is returned when $\scr{B}$ executes on $G$, then
	$
	\mb{E}[w(\scr{N})] = \sum_{v \in V} \sum_{\bm{e} \in \scr{C}_v } \val(\bm{e}) \cdot x_{v}(\bm{e}).
	$
	Thus, the optimality of $\scr{B}$ implies
	that $\rOPT(G) \le \LPOPT(G)$, and so together with \eqref{eqn:relaxed_benchmark_upper_bound},
	Theorem \ref{thm:new_LP_relaxation} follows. 
\end{proof}
\begin{remark}
	As mentioned, \ref{LP:new} is an exact LP formulation of the relaxed stochastic matching problem, as we prove in Theorem \ref{thm:LP_relaxation_benchmark_equivalence} of Appendix \ref{sec:LP_relations}.
	
\end{remark}

\subsection{Proving Lemma \ref{lem:non_adaptive_optimum}}
Let us suppose that $G=(U,V,E)$ is a stochastic graph with downward-closed
probing constraints $(\scr{C}_v)_{v \in V}$. In order to prove
Lemma \ref{lem:non_adaptive_optimum}, we must show that there exists
an optimal relaxed probing algorithm which is non-adaptive and satisfies \ref{eqn:committal}.
Our high level approach is to consider an optimal relaxed probing algorithm $\scr{A}$
which satisfies \ref{eqn:committal}, and then to construct a new
non-adaptive algorithm $\scr{B}$ by \textit{stealing} the strategy
of $\scr{A}$, without any loss in performance. More
specifically, we construct $\scr{B}$ by writing down for each $v \in V$
and $\bm{e} \in \scr{C}_v$ the probability that
$\scr{A}$ probes the edges of $\bm{e}$ in order. These
probabilities necessarily satisfy certain inequalities which we make
use of in designing $\scr{B}$. In order to do so, we need a technical randomized rounding
procedure whose precise relevance will become clear in the proof of
Lemma \ref{lem:non_adaptive_optimum}.

Suppose that $\bm{e} \in E^{(*)}$, and recall that $\lambda$ is the empty string/character. 
Let us now assume that $(y_{v}(\bm{e}))_{\bm{e} \in \scr{C}_v}$ is a collection of non-negative values
which satisfy $y_{v}(\lambda)=1$, and
\begin{equation} \label{eqn:marginal_distribution}
	\sum_{\substack{e \in \partial(v): \\ (\bm{e}',e) \in \scr{C}_v}} y_{v}(\bm{e}', e) \le y_{v}(\bm{e}'),
\end{equation}
for each $\bm{e}' \in \scr{C}_v$. 

\begin{proposition} \label{prop:vertex_round}
	Given a collection of values $(y_{v}(\bm{e}))_{\bm{e} \in \scr{C}_v}$ which satisfy
	$y_{v}(\lambda)=1$ and \eqref{eqn:marginal_distribution}, there exists a distribution $\scr{D}^v$ supported on $\scr{C}_v$,
	such that if $\bm{Y} \sim \scr{D}^v$, then for each $\bm{e}=(e_1,\ldots,e_k) \in \scr{C}_v$ with $k:= |\bm{e}| \ge 1$,
	it holds that
	\begin{equation}\label{eqn:target_probability}
		\mb{P}[ (\bm{Y}_{1}, \ldots , \bm{Y}_{k}) = (e_1, \ldots , e_{k})] = y_{v}(\bm{e}),
	\end{equation}
	where $\bm{Y}_1,\ldots ,\bm{Y}_k$ are the first $k$ characters of $\bm{Y}$
	(where $\bm{Y}_i := \lambda$ if $\bm{Y}$ has no $i^{th}$ character).
	
\end{proposition}

\begin{proof}[Proof of Proposition \ref{prop:vertex_round}]
	First define $\scr{C}^{>}_v:= \{ \bm{e}' \in \scr{C}_v : y_{v}(\bm{e}') > 0\}$, which we observe
	is downward-closed since by assumption $\scr{C}_v$ is downward-closed and \eqref{eqn:marginal_distribution} holds. We
	prove the proposition for $\scr{C}^{>}_v$, which we then argue implies the proposition holds for $\scr{C}_v$.
	Observe now that for each $\bm{e}' \in \scr{C}^{>}_v$, we have that
	\begin{equation}\label{eqn:online_vertex_edge_distribution}
		\sum_{\substack{e \in \partial(v): \\ (\bm{e}',e) \in \scr{C}^{>}_v}} \frac{y_{v}(\bm{e}', e)}{y_{v}(\bm{e}')} \le 1
	\end{equation}
	as a result of \eqref{eqn:marginal_distribution} (recall that $y_{v}(\lambda):=1$). We thus define for
	each $\bm{e}' \in \scr{C}^{>}_v$,
	\begin{equation}\label{eqn:pass_probability}
		z_{v}(\bm{e}'):= 1 - \sum_{\substack{e \in \partial(v): \\ (\bm{e}',e) \in \scr{C}^{>}_v}} \frac{y_{v}(\bm{e}',e)}{y_{v}(\bm{e}')},
	\end{equation}
	which we observe has the property that $0  \le  z_{v}(\bm{e}') <  1$. The strict inequality follows from the definition
	of $\scr{C}^{>}_v$. This leads to the following randomized rounding algorithm,
	which we claim outputs a random string $\bm{Y}$ which satisfies the desired properties:
	\begin{algorithm}[H]
		\caption{VertexRound} \label{alg:general_vertex_rounding}
		\begin{algorithmic}[1]
			\Require a collection of values $(y_{v}(\bm{e}))_{\bm{e} \in \scr{C}^{>}_v}$ satisfying \eqref{eqn:marginal_distribution} and $y_{v}(\lambda)=1$.
			\Ensure a random string $\bm{Y}=(Y_{1},\ldots ,Y_{|\partial(v)|})$ supported on $\scr{C}^{>}_v$.
			\State Set $\bm{e}' \leftarrow \lambda$.
			\State Initialize $Y_{i}=\lambda$ for each $i=1, \ldots , |\partial(v)|$.
			\For{$i=1, \ldots , |\partial(v)|$}
			\State Exit the ``for loop'' with probability $z_{v}(\bm{e}')$.
			\Comment{pass with a certain probability -- see \eqref{eqn:pass_probability}}
			\State Draw $e \in \partial(v)$ satisfying $(\bm{e}',e) \in \scr{C}^{>}_v$ with probability $y_{v}(\bm{e}', e)/ (y_{v}(\bm{e}') \, (1-z_{v}(\bm{e}')))$. \label{line:edge_draw}
			\State Set $Y_{i}=e$.
			\State $\bm{e}' \leftarrow (\bm{e}',e)$.
			\EndFor
			\State \Return $\bm{Y}=(Y_{1},\ldots ,Y_{|\partial(v)|})$. 		\Comment{concatenate the edges in order and return the resulting string}
		\end{algorithmic}
	\end{algorithm}
	Clearly, the random string $\bm{Y}$ is supported on $\scr{C}^{>}_v$, thanks to line \ref{line:edge_draw} of Algorithm \ref{alg:general_vertex_rounding}.
	We now show that \eqref{eqn:target_probability} holds. As such,
	let us first assume $k=1$, and $e \in \partial(v)$ satisfies $(e) \in \scr{C}^{>}_v$. Observe
	that
	\[
	\mb{P}[Y_{1} = e] = (1 - z_{v}(\lambda)) \frac{ y_{v}(e)}{1 - z_{v}(\lambda)} = y_{v}(e),
	\]
	as the algorithm does not exit the ``for loop'' with probability $1-z_{v}(\lambda)$, in which case it draws $e$
	with probability $y_{v}(e)/(1-z_{v}(\lambda))$. In general, take $k \ge 2$, and assume that for each $\bm{e}' \in \scr{C}^{>}_v$ with $1 \le |\bm{e}'| < k$,
	it holds that
	$
	\mb{P}[ (Y_{1},\ldots ,Y_{k}) = \bm{e}'] = y_{v}(\bm{e}').
	$
	If we now fix $\bm{e} =(e_{1}, \ldots, e_{k}) \in \scr{C}^{>}_v$ with $|\bm{e}|=k$,
	observe that $\bm{e}_{<k}:= (e_{1}, \ldots ,e_{k-1}) \in \scr{C}^{>}_v$, as $\scr{C}^{>}_v$ is downward-closed.
	Moreover,
	\begin{align*}
		\mb{P}[(Y_{1}, \ldots , Y_{k}) = \bm{e}] &= \mb{P}[ Y_{k} =  e_k \, | \, (Y_{1}, \ldots , Y_{k-1}) = \bm{e}_{<k}] \cdot \mb{P}[ (Y_{1}, \ldots , Y_{k-1}) = \bm{e}_{<k}] \\
		&= \mb{P}[ Y_{k} = e_k \, | \, (Y_{1}, \ldots , Y_{k-1}) = \bm{e}_{<k})] \cdot y_{v}(\bm{e}_{<k}),
	\end{align*}
	where the last line follows by the induction hypothesis since $\bm{e}_{<k} \in \scr{C}^{>}_v$ is of length $k-1$. 
	We know however that
	\[
	\mb{P}[ Y_{k} = \bm{e}_k \, | \, (Y_{1}, \ldots , Y_{k-1}) = \bm{e}_{<k}]= (1-z_{v}(\bm{e}_{<k})) \, \frac{y_{v}(\bm{e}_{<k},e_k)}{y_{v}(\bm{e}_{<k}) \, (1-z_{v}(\bm{e}_{<k}))} = \frac{y_{v}(\bm{e}_{<k},e_k)}{y_{v}(\bm{e}_{<k})}.
	\]
	This is because once we condition on the event $(Y_{1}, \ldots , Y_{k-1}) = \bm{e}_{<k}$, we know that the algorithm
	does not exit the ``for loop'' with probability $1 - z_{v}(\bm{e}_{<k})$, in which case it selects $e_{k} \in \partial(v)$ with probability
	$y_{v}(\bm{e}_{<k}, e_k)/(y_{v}(\bm{e}_{<k}) \, (1-z_{v}(\bm{e}_{<k})))$, since $(\bm{e}_{<k}, e_k) \in \scr{C}^{>}_v$ by assumption. As such, we have that
	$
	\mb{P}[(Y_{1}, \ldots , Y_{k}) = \bm{e}] = y_{v}(\bm{e}),
	$
	and so the proposition holds for $\scr{C}^{>}_v$. 
	To complete the argument, observe
	that since $\bm{Y}$ is supported on $\scr{C}^{>}_v$, the substrings
	of $\bm{Y}$ are also supported on $\scr{C}^{>}_v$, as $\scr{C}^{>}_v$ is downward-closed. Thus, $\bm{Y}$
	satisfies \eqref{eqn:target_probability} for the non-empty strings of $\scr{C}_v \setminus \scr{C}^{>}_v$, in addition
	to the non-empty strings of $\scr{C}^{>}_v$.
\end{proof}

\begin{proof}[Proof of Lemma \ref{lem:non_adaptive_optimum}]
	
	Suppose that $\scr{A}$ is an optimal relaxed probing algorithm which returns
	the one-sided matching $\scr{M}$ after executing on the stochastic graph
	$G=(U,V,E)$. In a slight abuse of terminology, we say that $e$
	is matched by $\scr{A}$, provided $e$ is included in $\scr{M}$.
	We shall also make the simplifying assumption that $p_{e} < 1$ for each $e \in E$,
	as the proof can be clearly adapted to handle the case when certain edges have
	$p_{e}=1$ by restricting which strings of each $\scr{C}_v$ are considered.

	Observe that since $\scr{A}$ is optimal, it is clear
	that we may assume the following properties hold w.l.o.g. for each $e \in E$:
	\begin{enumerate}
		\item $e$ is probed only if $e$ can be added to the currently constructed one-sided
		matching. \label{eqn:probing_only_if}
		\item If $e$ is probed and $\sta(e)=1$, then $e$ is included in $\scr{M}$. \label{eqn:if_active_probe}
	\end{enumerate}
	Thus, in order to prove the lemma, we must find an alternative algorithm $\scr{B}$ which
	is non-adaptive, yet continues to be optimal.
	To this end, we shall first express $\mb{E}[w(\scr{M}(v))]$ in
	a convenient form for each $v \in V$, where  $w(\scr{M}(v))$ is the weight of the edge matched to $v$ (which is $0$ if no match occurs).
	
	Given $v \in V$ and $1 \le i \le |U|$, we define $X_{i}^{v}$ to be the $i^{th}$ edge adjacent to $v$ that is probed by $\scr{A}$.
	This is set equal to $\lambda$ by convention, provided no such edge exists. We may
	then define $\bm{X}^{v}:=(X^{v}_{1}, \ldots , X^{v}_{|U|})$, and
	$\bm{X}_{\le k}^{v} := (X^{v}_1, \ldots , X^{v}_k)$ for each $1 \le k \le |U|$. Moreover, given $\bm{e} =(e_{1}, \ldots ,e_{k}) \in E^{(*)}$ with $k \ge 1$, define $S(\bm{e})$ to be
	the event in which $e_{k}$ is the only active edge amongst $e_{1}, \ldots ,e_{k}$.
	Observe then that
	\[
	\mb{E}[ w(\scr{M}(v))] = \sum_{\substack{\bm{e}=(e_{1}, \ldots ,e_{k}) \in \scr{C}_{v}: \\ k \ge 1}} w_{e_k} \mb{P}[S(\bm{e}) \cap \{\bm{X}^{v}_{\le k} = \bm{e}\}],
	\]
	as \eqref{eqn:probing_only_if} and \eqref{eqn:if_active_probe} ensure $v$ is matched to the first probed
	edge which is revealed to be active. Moreover,
	if $\bm{e}=(e_{1}, \ldots ,e_{k}) \in \scr{C}_v$ for $k \ge 2$, then 
	\begin{equation}
		\mb{P}[S(\bm{e}) \cap  \{\bm{X}_{\le k}^{v} = \bm{e} \}] = \mb{P}[\{\sta(e_k) =1 \}\cap \{\bm{X}_{\le k}^{v} = \bm{e}\}],
	\end{equation}
	as \eqref{eqn:probing_only_if} and \eqref{eqn:if_active_probe} ensure $\bm{X}_{\le k}^{v} = \bm{e}$ only if $e_{1}, \ldots , e_{k-1}$ are inactive. Thus, 
	\begin{align*}
		\mb{E}[ w(\scr{M}(v))] &= \sum_{\substack{\bm{e}=(e_{1}, \ldots ,e_{k}) \in \scr{C}_{v}: \\ k \ge 1}} w_{e_k} \mb{P}[S(\bm{e}) \cap \{\bm{X}_{\le k}^{v} = \bm{e} \}] \\
		&= \sum_{\substack{\bm{e}=(e_{1}, \ldots ,e_{k}) \in \scr{C}_{v}: \\ k \ge 1}} w_{e_k} \mb{P}[\{\sta(e_{k}) =1\} \cap \{\bm{X}_{\le k}^{v} = \bm{e} \}] \\
		&= \sum_{\substack{\bm{e}=(e_{1}, \ldots ,e_{k}) \in \scr{C}_{v}: \\ k \ge 1}} w_{e_k} p_{e_k} \mb{P}[\bm{X}_{\le k}^{v} = \bm{e}],
	\end{align*}
	where the final equality holds since $\scr{A}$ must decide on whether to probe $e_{k}$ prior to revealing $\sta(e_k)$.
	As a result, after summing over $v \in V$,
	\begin{equation}\label{eqn:target_value}
		\mb{E}[ w(\scr{M})] = \sum_{v \in V} \sum_{\substack{\bm{e}=(e_{1}, \ldots ,e_{k}) \in \scr{C}_{v}: \\ k \ge 1}} w_{e_k} p_{e_k} \mb{P}[\bm{X}_{\le k}^{v} = \bm{e} ].
	\end{equation}
	Our goal is to find a non-adaptive relaxed probing algorithm which matches the value of \eqref{eqn:target_value}.
	Thus, for each $v \in V$ and $\bm{e}=(e_{1}, \ldots ,e_{k}) \in \scr{C}_v$ with $k \ge 1$,
	define $x_{v}(\bm{e}):= \mb{P}[\bm{X}_{\le k}^{v} = \bm{e}]$,
	where $x_{v}(\lambda):=1$. Observe now that for each $\bm{e}'=(e_{1}', \ldots , e_{k}') \in \scr{C}_v$,
	\begin{equation}\label{eqn:probability_consistency_conditional}
		\sum_{\substack{e \in \partial(v): \\ (\bm{e}',e) \in \scr{C}_v}} \mb{P}[\bm{X}_{\le k+1}^{v} = (\bm{e}',e) \, | \, \bm{X}_{\le k}^{v} = \bm{e}'] \le 1 - p_{e_k'}.
	\end{equation}
	To see \eqref{eqn:probability_consistency_conditional}, observe that the the left-hand side corresponds to the probability $\scr{A}$
	probes some edge $e \in \partial(v)$, given it already probed $\bm{e}'$ in order. On the other hand, if a subsequent edge is probed,
	then \eqref{eqn:probing_only_if} and \eqref{eqn:if_active_probe} imply that $e'_k$ must have been inactive, which occurs
	independently of the event $\bm{X}_{\le k}^{v} = \bm{e}'$. This explains the right-hand side of \eqref{eqn:probability_consistency_conditional}. Using \eqref{eqn:probability_consistency_conditional}, the values
	$(x_{v}(\bm{e}))_{\bm{e} \in \scr{C}_{v}}$ satisfy
	\begin{equation}\label{eqn:probability_consistency}
		\sum_{\substack{e \in \partial(v): \\ (\bm{e}',e) \in \scr{C}_v}} x_{v}(\bm{e}', e) \le (1 - p_{e'_k}) \cdot x_{v}(\bm{e}'),
	\end{equation}
	for each $\bm{e}'=(e_{1}', \ldots , e_{k}') \in \scr{C}_v$ with $k \ge 1$. Moreover, clearly
	$\sum_{e \in \partial(v)} x_{v}(e) \le 1$.

	Given $\bm{e} =(e_1, \ldots ,e_k) \in \scr{C}_v$ for $k \ge 1$,
	recall that $\bm{e}_{<k}:=(e_1, \ldots ,e_{k-1})$ where $\bm{e}_{< 1}:= \lambda$ if $k=1$.
	Moreover, $q(\bm{e}_{<k}):= \prod_{i=1}^{k-1}(1- p_{e_i})$, where $q(\lambda):=1$.
	Using this notation, define for each $\bm{e} \in \scr{C}_v$
	\begin{equation}\label{eqn:y_value_definition}
		y_{v}(\bm{e}):=
		\begin{cases}
			x_{v}(\bm{e})/ q(\bm{e}_{<|\bm{e}|}) & \text{if $|\bm{e}| \ge 1$,} \\
			1 & \text{otherwise.}
		\end{cases}
	\end{equation}
	Observe that \eqref{eqn:probability_consistency} ensures that for each $\bm{e}' \in \scr{C}_v$,
	\begin{equation}
		\sum_{\substack{e \in \partial(v): \\ (\bm{e}',e) \in \scr{C}_v}} y_{v}(\bm{e}', e) \le y_{v}(\bm{e}'),
	\end{equation}
	and $y_{v}(\lambda):=1$.
	As a result, Proposition \ref{prop:vertex_round} implies
	that for each $v \in V$, there exists a distribution $\scr{D}^{v}$ such that if
	$\bm{Y}^{v} \sim \scr{D}^{v}$, then for each $\bm{e} \in \scr{C}_v$ with $|\bm{e}|=k \ge 1$,
	\begin{equation}\label{eqn:non_adaptive_target_probability}
		\mb{P}[\bm{Y}^{v}_{\le k} = \bm{e}] = y_{v}(\bm{e}).
	\end{equation}
	Moreover, $\bm{Y}^{v}$ is drawn independently from the edge states, $(\sta(e))_{e \in E}$. Consider now the following algorithm $\scr{B}$, which satisfies the desired properties
	\ref{eqn:committal} and \ref{eqn:non_adaptive} of Lemma \ref{lem:non_adaptive_optimum}:
	\begin{algorithm}[H]
		\caption{Algorithm $\scr{B}$} \label{alg:non_adaptive_relaxed}
		\begin{algorithmic}[1]
			\Require a stochastic graph $G=(U,V,E)$.
			\Ensure a one-sided matching $\scr{N}$ of $G$ of active edges.
			\State Set $\scr{N} \leftarrow \emptyset$.
			\State Draw $(\bm{Y}^{v})_{v \in V}$ according to the product distribution $\prod_{v \in V} \scr{D}^{v}$.
			\For{$v \in V$}
			\For{$i=1, \ldots , |\bm{Y}^{v}|$}
			\State Set $e \leftarrow \bm{Y}^{v}_{i}$.			\Comment{$\bm{Y}^{v}_{i}$ is the $i^{th}$ edge of $\bm{Y}^{v}$}
			\State Probe the edge $e$, revealing $\sta(e)$.
			\If{$\sta(e) =1$ and $v$ is unmatched by $\scr{N}$}
			\State Add $e$ to $\scr{N}$.
			\EndIf
			\EndFor
			\EndFor
			\State \Return $\scr{N}$.
		\end{algorithmic}
	\end{algorithm}
	Using \eqref{eqn:non_adaptive_target_probability} and the non-adaptivity of $\scr{B}$, it is clear that
	for each $v \in V$,
	\begin{align*}
		\mb{E}[w(\scr{N}(v))] &= \sum_{\substack{\bm{e}=(e_{1}, \ldots ,e_{k}) \in \scr{C}_{v}: \\ k \ge 1}} w_{e_k} \mb{P}[S(\bm{e})] \cdot \mb{P}[\bm{Y}_{\le k}^{v} = \bm{e}] \\
		&= \sum_{\substack{\bm{e}=(e_{1}, \ldots ,e_{k}) \in \scr{C}_{v}: \\ k \ge 1}} w_{e_k} p_{e_k} q(\bm{e}_{<k}) y_{v}(\bm{e}) \\
		&=  \sum_{\substack{\bm{e}=(e_{1}, \ldots ,e_{k}) \in \scr{C}_{v}: \\ k \ge 1}} w_{e_k} p_{e_k} x_{v}(\bm{e})
		= \mb{E}[w(\scr{M}(v))].
	\end{align*}
	Thus, after summing over $v \in V$, it holds that $\mb{E}[w(\scr{N})] = \mb{E}[w(\scr{M})] = \rOPT(G)$,
	and so in addition to satisfying \ref{eqn:committal} and \ref{eqn:non_adaptive}, $\scr{B}$ is optimal.
	Finally, it is easy to show that each $u \in U$ is matched by $\scr{N}$
	at most once in expectation since $\scr{M}$ has this property.
	Thus, $\scr{B}$ is a relaxed probing algorithm which is optimal and satisfies the required properties of Lemma \ref{lem:non_adaptive_optimum}.
\end{proof}

\section{Known I.D. Instance Model} \label{sec:known_id_additions}

Suppose that $(H_{\typ}, (\scr{D}_i)_{i=1}^{n})$ is a known i.d. input,
where $H_{\typ}=(U,B,F)$ has downward-closed online probing constraints $(\scr{C}_b)_{b \in B}$. 
If $G \sim (H_{\typ}, (\scr{D}_i)_{i=1}^{n})$, where $G=(U,V,E)$ has vertices $V=\{v_{1}, \ldots ,v_{n}\}$, then
define $r_{i}(b):=\mb{P}[v_i = b]$ for each $i \in [n]$ and $b \in B$, where we hereby assume
that $r_{i}(b) > 0$. We generalize \ref{LP:new} to account for the distributions $(\scr{D}_i)_{i=1}^{n}$.
For each $i \in [n], b \in B$ and $\bm{e} \in \scr{C}_b$,
we introduce a decision variable $x_{i}( \bm{e} \, || \, b)$
to encode the probability that $v_i$ has type $b$ \textit{and} $\bm{e}$
is the sequence of edges of $\partial(v_i)$ probed by the \textit{relaxed} benchmark.

\begin{align}\label{LP:new_id}
	\tag{LP-config-id}
	&\text{maximize} &  \sum_{i \in [n], b \in B, \bm{e} \in \scr{C}_b} \val(\bm{e}) \cdot x_{i}(\bm{e} \, ||  \, b)  \\
	&\text{subject to} & \sum_{i \in [n], b \in B} \sum_{\substack{ \bm{e} \in \scr{C}_b: \\ (u,b) \in \bm{e}}} 
	p_{u,b} \cdot q(\bm{e}_{< (u,b)}) \cdot x_{i}( \bm{e} \, || \, b)  \leq 1 && \forall u \in U  \label{eqn:relaxation_efficiency_matching_id}\\
	&& \sum_{\bm{e} \in \scr{C}_b} x_{i}(\bm{e} \, ||  \, b)= r_{i}(b)  && \forall b \in B, i \in [n]  \label{eqn:relaxation_efficiency_distribution_id} \\
	&&x_{i}(\bm{e} \, ||  \, b) \ge 0 && \forall b \in B, \bm{e} \in \scr{C}_b, i \in [n]
\end{align}
Let us denote $\LPOPT(H_{\typ}, (\scr{D}_i)_{i=1}^{n})$ as the value of an optimum solution
to \ref{LP:new_id}.
\begin{theorem} \label{thm:known_id_relaxation}
	$\OPT(H_{\typ}, (\scr{D}_i)_{i=1}^{n}) \le \LPOPT(H_{\typ}, (\scr{D}_i)_{i=1}^{n})$.
\end{theorem}
One way to prove Theorem \ref{thm:known_id_relaxation} is to use the properties
of the relaxed benchmark on $G$ guaranteed by Lemma \ref{lem:non_adaptive_optimum},
and the above interpretation of the decision variables to argue that
$
\mb{E}[\rOPT(G)] \le \LPOPT(H_{\typ}, (\scr{D}_i)_{i=1}^{n}),
$
where $\rOPT(G)$ is the value of the relaxed benchmark on $G$. Specifically, we can interpret \eqref{eqn:relaxation_efficiency_matching_id} as saying that the
relaxed benchmark matches each offline vertex at most once in expectation. Moreover,
\eqref{eqn:relaxation_efficiency_distribution_id} holds by observing that if
$v_i$ is of type $b$, then the relaxed benchmark selects some $\bm{e} \in \scr{C}_{b}$
to probe (note $\bm{e}$ could be the empty-string). 
We provide a morally equivalent proof below. Specifically, we consider an optimum solution of \ref{LP:new} with respect to $G$, and apply a conditioning argument in conjunction with Theorem \ref{thm:new_LP_relaxation}. 
\begin{proof}[Proof of Theorem \ref{thm:known_id_relaxation}]
	
	Suppose that $(H_{\typ}, (\scr{D}_t)_{t=1}^{n})$ is a known i.d. input,
	where $H_{\typ}=(U,B,F)$. Recall that $\scr{C}_b$ corresponds to the online probing
	constraint of each type node $b \in B$. For convenience, we denote $\scr{I}:= \sqcup_{b \in B} \scr{C}_b$.
	We can then define the following collection of random variables, 
	denoted $(X_{t}(\bm{e}))_{t \in [n], \bm{e} \in \scr{I}}$,
	based on the following randomized procedure:
	
	\begin{itemize}
		\item Draw the stochastic graph $G \sim (H_{\typ}, (\scr{D}_t)_{t=1}^{n})$,
		whose vertex arrivals we denote by $v_{1}, \ldots , v_{n}$.
		\item Compute an optimum solution of \ref{LP:new} for $G$,
		which we denote by $(x_{v_t}(\bm{e}))_{t \in [n], \bm{e} \in \scr{C}_{v_t}}$.
		\item For each $t=1, \ldots ,n$ and $\bm{e} \in \scr{I}$, 
		set $X_{t}(\bm{e}) = x_{v_t}(\bm{e})$ if $\bm{e} \in \scr{C}_{v_t}$,
		otherwise set $X_{t}(\bm{e}) = 0$.
	\end{itemize}
	Observe then that since by definition $(X_{v_t}(\bm{e}))_{t \in [n], \bm{e} \in \scr{C}_{v_t}}$
	is a feasible solution to \ref{LP:new} for $G$, it holds that
	for each $t=1, \ldots ,n$
	\begin{equation}\label{eqn:online_distribution_iid}
		\sum_{\bm{e} \in \scr{I}} X_{t}(\bm{e}) = 1,
	\end{equation}
	and
	\begin{equation} \label{eqn:offline_matching_iid}
		\sum_{t \in [n], b \in B} \sum_{\substack{ \bm{e} \in \scr{I}: \\ (u,b) \in \bm{e}}} 
		p_{u,b} \cdot q(\bm{e}_{< (u,b)}) \cdot X_{t}( \bm{e}) \le 1,
	\end{equation}
	for each $u \in U$. Moreover, $(X_{t}( \bm{e}))_{t \in [n], \bm{e} \in \scr{C}_{v_t}}$ is
	a optimum solution to \ref{LP:new} for $G$, so Theorem \ref{thm:new_LP_relaxation}
	implies that
	\begin{equation}\label{eqn:known_iid_benchmark_relaxtion}
		\OPT(G) \le \LPOPT(G) = \sum_{t=1}^{n} \sum_{\bm{e} \in \scr{I}} \val(\bm{e}) \cdot X_{t}(\bm{e}). 
	\end{equation}
	In order to make use of these inequalities in the context of the type graph $H_{\typ}$,
	let us first fix a type node $b \in B$ and a string $\bm{e} \in \scr{C}_b$. For each $t \in [n]$, we can then define 
	\begin{equation}
		x_{t}(\bm{e} \, || \, b):=\mb{E}[ X_{t}(\bm{e}) \cdot \bm{1}_{[v_{t} = b]}],
	\end{equation}
	where the randomness is over the generation of $G$. Observe
	that by definition of the $(X_{t}(\bm{e}))_{t \in [n], \bm{e} \in \scr{I}}$ values, $x_{t}(\bm{e} \, || \, b) =0$,
	provided $\bm{e} \notin \scr{C}_b$. We claim that $(x_{t}(\bm{e} \, || \, b) )_{b \in B, t \in [n], \bm{e} \in \scr{C}_b}$ is a feasible
	solution to \ref{LP:new_id}. To see this, first observe that if we multiply \eqref{eqn:online_distribution_iid} by the indicator random
	variable $\bm{1}_{[v_{t} = b]}$, then we get that 
	$
	\sum_{\bm{e} \in \scr{I}} X_{t}(\bm{e}) \cdot \bm{1}_{[v_{t} = b]} = \bm{1}_{[v_{t} = b]}.
	$
	As a result, if we take expectations over this equality,
	\begin{align*}
		\sum_{\bm{e} \in \scr{I}} x_{t}(\bm{e} \, || \, b) = \sum_{\bm{e} \in \scr{I}} \mb{E}\left[ X_{t}(\bm{e}) \cdot \bm{1}_{[v_{t} = b]}\right] 
		&= \mb{P}[v_t =b] =: r_{t}(b),
	\end{align*}
	for each $b \in B$ and $t \in [n]$. Let us now fix $u \in U$. Observe that since $X_{t}( \bm{e}) \cdot \bm{1}_{[v_{t} = b]} = X_{t}( \bm{e})$ for
	each $\bm{e} \in \scr{C}_b$, \eqref{eqn:offline_matching_iid} ensures that
	\begin{equation} \label{eqn:rearrangment}
		\sum_{t \in [n], b \in B} \sum_{\substack{ \bm{e} \in \scr{C}_b: \\ (u,b) \in \bm{e}}} 
		p_{u,b} \cdot q(\bm{e}_{< (u,b)}) \cdot X_{t}( \bm{e}) \cdot \bm{1}_{[v_{t} = b]}= \sum_{t \in [n], b \in B} \sum_{\substack{ \bm{e} \in \scr{C}_b: \\ (u,b) \in \bm{e}}} 
		p_{u,b} \cdot q(\bm{e}_{< (u,b)}) \cdot X_{t}( \bm{e}) \le 1
	\end{equation}
	Thus, after taking expectations over \eqref{eqn:rearrangment},
	$
	\sum_{t \in [n], b \in B} \sum_{\substack{ \bm{e} \in \scr{C}_b: \\ (u,b) \in \bm{e}}} 
	p_{u,b} \cdot q(\bm{e}_{< (u,b)}) \cdot x_{t}( \bm{e} \, || \, b)  \le 1,
	$
	for each $u \in U$. Since $(x_{t}(\bm{e} \, || \, b))_{t \in [n], b \in B,\bm{e} \in \scr{C}_{b}}$ satisfies these inequalities,
	and the variables are clearly all non-negative, 
	it follows that $(x_{t}(\bm{e} \, || \, b))_{t \in [n], b \in B,\bm{e} \in \scr{C}_{b}}$ is a feasible solution to \ref{LP:new_id}.
	Let us now express the right-hand side of \eqref{eqn:known_iid_benchmark_relaxtion}
	as in \eqref{eqn:rearrangment} and take expectations. We then get that
	\begin{align*}
		\mb{E}[ \OPT(G)] &\le \sum_{b \in B, t \in [n]} \sum_{\bm{e} \in \scr{I}} \val(\bm{e}) \cdot x_{t}(\bm{e} \, || \, b).
	\end{align*}
	Now, $\OPT(H_{\typ}, (\scr{D}_i)_{i=1}^{n}) = \mb{E}[ \OPT(G)]$ by definition, so since $(x_{t}(\bm{e} \, || \, b))_{b \in B, t \in [n], \bm{e} \in \scr{C}_b}$
	is feasible, it holds that
	$	\OPT(H_{\typ}, (\scr{D}_{i})_{i=1}^{n}) \le \LPOPT(H_{\typ}, (\scr{D}_{i})_{i=1}^{n}),$
	thus completing the proof.
\end{proof}

Given a feasible solution to \ref{LP:new_id}, 
say $(x_{i}(\bm{e} \, || \, b))_{i \in [n], b \in B,  \bm{e} \in \scr{C}_b}$, for each $u \in U, i \in [n]$ and $b \in B$ define
\begin{equation}\label{eqn:induced_edge_variables_id}
	\til{x}_{u,i}(b):= \sum_{\substack{ \bm{e} \in \scr{C}_b: \\ (u,b) \in \bm{e}}} 
	q(\bm{e}_{< (u,b)}) \cdot x_{i}( \bm{e} \, || \, b). 
\end{equation}
We refer to $\til{x}_{u,i}(b)$ as an \textbf{edge variable}, thus extending the definition from the known stochastic
graph setting. Suppose now that we fix $i \in [n]$ and $b \in B$, and consider the variables, $(x_{i}(\bm{e} \, || \, b))_{\bm{e} \in \scr{C}_b}$. Observe that \eqref{eqn:relaxation_efficiency_distribution_id} ensures that 
$
\frac{\sum_{\bm{e} \in \scr{C}_b} x_{i}(\bm{e} \, || \, b)}{ r_{i}(b)} = 1.
$
Hence, regardless of which type node $v_{i}$ is drawn as,
$
\frac{\sum_{\bm{e} \in \scr{C}_{v_i}} x_{i}(\bm{e} \, || \, v_i)}{ r_{i}(v_i)} = 1.
$

Using the above observation, we can generalize $\mathtt{VertexProbe}$ as follows. Given vertex $v_i$,
draw $\bm{e}' \in \scr{C}_{v_i}$ with probability $x_{i}(\bm{e}' \, || \, v_i)/r_{i}(v_i)$. If $\bm{e}' = \lambda$, then return the empty-set. Otherwise, set $\bm{e}' = (e_{1}', \ldots ,e_{k}')$ for $k := |\bm{e}'| \ge 1$, and probe the
edges of $\bm{e}'$ in order. Return the first edge which is revealed to be active, if such an
edge exists. Otherwise, return the empty-set. We denote the output of $\mathtt{VertexProbe}$ on the input 
$(v_i, \partial(v_i), (x_{i}(\bm{e} \, || \, v_i)/ r_{i}(v_i))_{\bm{e} \in \scr{C}_{v_i}})$ by
$\mathtt{VertexProbe}(v_i, \partial(v_i), (x_{i}(\bm{e} \, || \, v_i)/ r_{i}(v_i))_{\bm{e} \in \scr{C}_{v_i}})$.
Observe then the following extension of Lemma \ref{lem:fixed_vertex_probe}:
\begin{lemma}\label{lem:fixed_vertex_probe_id}
	If $\mathtt{VertexProbe}$ is passed $\left(v_{i}, \partial(v_i), (x_{i}(\bm{e} \, || \, v_i) / r_{i}(v_i))_{\bm{e} \in \scr{C}_{v_i}}\right)$, then for any $b \in B$ and $u \in U$,
	\[
	\mb{P}[\mathtt{VertexProbe}(v_i, \partial(v_i), (x_{i}(\bm{e} \, || \, v_i)/ r_{i}(v_i))_{\bm{e} \in \scr{C}_{v_i}}) =(u,b) \, | \, v_{i} = b] = \frac{p_{u,b} \cdot \til{x}_{u,i}(b)}{r_{i}(b)}.
	\]
\end{lemma}
\begin{definition}[Propose - Known I.D Instance]\label{def:vertex_probe}
	We say that $\mathtt{VertexProbe}$ \textbf{proposes} $v_i$ to vertex $u \in \partial(v_i)$, provided the algorithm outputs $(u,v_i)$ when executing on online vertex $v_i$ for $i \in [n]$. When it is clear that $\mathtt{VertexProbe}$ is being executed on $v_i$, we say that $v_i$ proposes to $u$.
\end{definition}

Consider now the generalization of Algorithm \ref{alg:known_stochastic_graph} where
$\pi$ is generated either u.a.r. or adversarially.
\begin{algorithm}[H]
	\caption{Known I.D} 
	\label{alg:known_id}
	\begin{algorithmic}[1]
		\Require a known i.d. input $(H_{\typ}, (\scr{D}_i)_{i=1}^{n})$.
		\Ensure a matching $\scr{M}$ of active edges of $G \sim (H_{\typ}, (\scr{D}_i)_{i=1}^{n})$.
		\State $\scr{M} \leftarrow \emptyset$.
		\State Compute an optimum solution of \ref{LP:new_id} for $(H_{\typ}, (\scr{D}_i)_{i=1}^{n})$, say $(x_{i}(\bm{e} \, || \, b))_{i \in [n], b \in B,  \bm{e} \in \scr{C}_b}$.
		\For{$t=1, \ldots , n$} 
		\State Let $a \in B$ be the type of the current arrival $v_{\pi(t)}$.           \Comment{to simplify notation}
		\State Set $e \leftarrow \mathtt{VertexProbe}\left(v_{\pi(t)},\partial(v_{\pi(t)}), \left( x_{\pi(t)}(\bm{e} \, || \, a) \cdot r^{-1}_{\pi(t)}(a) \right)_{\bm{e} \in \scr{C}_{a}}\right)$.
		\If{$e=(u,v_{\pi(t)})$ for some $u \in U$, and $u$ is unmatched}
		\State Add $e$ to $\scr{M}$.
		\EndIf
		\EndFor
		\State \Return $\scr{M}$.
		
	\end{algorithmic}
\end{algorithm}
Similarly, to Algorithm \ref{alg:known_stochastic_graph},
one can show that Algorithm \ref{alg:known_id} attains a competitive ratio of $1/2$ for random order arrivals. Interestingly, if the distributions $(\scr{D}_{i})_{i=1}^{n}$ are identical -- that is, we work with a known i.i.d. instance -- then it is relatively  easy to show that
this algorithm's competitive ratio improves to $1-1/e$.
\begin{proposition} \label{prop:known_iid}
	If Algorithm \ref{alg:known_id} is presented a known $i.i.d.$ input, say the type graph $H_{\typ}$
	together with the distribution $\scr{D}$, then 
	$\mb{E}[w(\scr{M})] \ge \left(1 - 1/e \right) \OPT(H_{\typ} , \scr{D})$.
\end{proposition}
For the case of non-identical distributions, we require online contention resolution schemes. Given an arbitrary integer $k \ge 1$, define the ground set $[k]:=\{1, \ldots ,k\}$. Fix $\bm{z} \in [0,1]^{k}$,
and for each $i \in [k]$, let $Z_i$ be an indicator random variable
which is $1$ independently with probability $z_i$.
Let us denote $\scr{P}:=\{ \bm{z} \in [0,1]^{k}: \sum_{i=1}^{k} z_{i} \le 1\}$. Note that $\scr{P}$ is the convex relaxation of the constraint imposed by the $1$-uniform matroid on $[k]$. 

\begin{definition}[Contention Resolution Scheme -- $1$-Uniform Matroid -- \cite{Vondrak_2011}]\label{def:CRS}
	A \textbf{contention resolution scheme} (CRS) for the $1$-uniform matroid on $[k]$ is a (randomized) algorithm $\psi$, which given
	$\bm{z} \in \scr{P}$ and $S \subseteq [k]$ as inputs, returns at most one element $\psi(\bm{z},S)$ of $S$. Given $\alpha \in [0,1]$, $\psi$ is said to be $\alpha$-\textbf{selectable},
	provided for all $i \in [k]$ and $\bm{z} \in \scr{P}$,
	\begin{equation} \label{eqn:selectibility}
		\mb{P}[ i = \psi(\bm{z},R(\bm{z})) \mid i \in R(\bm{z})] \ge \alpha,
	\end{equation}
	where the probability is over the generation of $R(\bm{z}):=\{j \in [k]: Z_j =1\}$,
	and the potential randomness used by $\psi$. We say that $\psi$ is an \textbf{online} CRS, provided it is revealed
	$(Z_i)_{i=1}^k$ one-by-one in an (unknown) adversarially chosen order $\pi$. Upon learning $Z_{\pi(t)}$, it makes a irrevocable decision whether or not to return $\pi(t)$ as its output. A random CRS
	is an OCRS where the arrival order $\pi$ is instead drawn u.a.r. and independently of all other randomization.
\end{definition}
During the execution of Algorithm \ref{alg:known_id},
let $Z_{u,i}$ be the indicator random variable for the event
in which $v_i$ proposes to vertex $u$.
Then, $\mb{P}[ Z_{u,i} =1] = z_{u,i}$ where
\begin{equation}
	z_{u,i} :=  \sum_{b \in B} p_{u,b} \cdot \til{x}_{u,i}(b) = \sum_{b \in B} \sum_{\substack{ \bm{e} \in \scr{C}_b: \\ (u,b) \in \bm{e}}} 
	p_{u,b} \cdot q(\bm{e}_{< (u,b)}) \cdot x_{i}( \bm{e} \, || \, b).
\end{equation}
Moreover, for each fixed $u \in U$,
the random variables $(Z_{u,i})_{i=1}^{n}$ are independent. We now
prove that we can apply the same reduction to online (random order) contention resolution as in the known stochastic graph setting.

\begin{theorem} \label{thm:known_id_reduction}
	Given an $\alpha$-selectable OCRS (RCRS) for $1$-uniform matroids, there exists
	an $\alpha$-competitive online probing algorithm for known i.d. instances and adversarial
	(random order) arrivals.
\end{theorem}

\begin{proof}[Proof of Theorem \ref{thm:known_id_reduction}]
	Let us first consider the case when $\pi$ is adversarially generated, and we are given an $\alpha$-selectable OCRS $\psi$. For notational simplicity, suppose that $\pi(t)=t$ for each $t \in [n]$, so that the online vertices arrive in order $v_{1}, \ldots , v_{n}$.
	We first describe the modification of \Cref{alg:known_id} using $|U|$ concurrent executions of $\psi$.
	\begin{algorithm}[H]
		\caption{Known I.D. -- AOM -- Modified} 
		\label{alg:known_id_aom_modified}
		\begin{algorithmic}[1]
			\Require a known i.d. input $(H_{\typ}, (\scr{D}_i)_{i=1}^{n})$.
			\Ensure a matching $\scr{M}$ of active edges of $G \sim (H_{\typ}, (\scr{D}_t)_{t=1}^{n})$.
			\State $\scr{M} \leftarrow \emptyset$.
			\State Compute an optimum solution of \ref{LP:new_id} for $(H_{\typ}, (\scr{D}_i)_{i=1}^{n})$, say $(x_{i}(\bm{e} \, || \, b))_{i \in [n], b \in B,  \bm{e} \in \scr{C}_b}$.
			\For{$t=1,\ldots ,n$}
			\State Let $a \in B$ be the type of the current arrival $v_{t}$. 
			\State Set $e \leftarrow \mathtt{VertexProbe}\left(v_{t}, \partial(v_{t}), \left(x_{t}(\bm{e} \, || \, a) \cdot r^{-1}_{t}(a)  \right)_{\bm{e} \in \scr{C}_{a}} \right)$.
			\If{$e=(u,v_t)$ for some $u \in U$}			\Comment{$v_t$ proposes to $u$ (i.e., $Z_{u,t} =1$)}
			\State Execute $\psi$ on $(z_{u,v_i})_{i=1}^{n}$ and $(Z_{u,i})_{i=1}^{t}$, and add $e =(u,v_t)$ to $\scr{M}$ if $e$ is returned by $\psi$. \label{line:adversarial_contention}
			\EndIf
			\EndFor
			\State \Return $\scr{M}$.
		\end{algorithmic}
	\end{algorithm}
	We first verify that $\scr{M}$ is a matching of $G$. To see this, notice that $\psi$
	returns at most once edge for each $u \in U$. Moreover, given $i \in [n]$,
	a necessary condition for $(u,v_i) \in \scr{M}$ is that $\mathtt{VertexProbe}$ must propose $v_i$ to $u$. However, for a fixed $v_i$, there is at most one vertex of $U$ for which this occurs.

	To prove the algorithm is $\alpha$-competitive, first observe that the edge variables $(\til{x}_{u,t}(b))_{u \in U,t \in [n], b \in B}$
	satisfy
	$
	\LPOPT(H_{\typ}, (\scr{D}_i)_{i=1}^{n}) = \sum_{u \in U, t \in [n], b \in B} p_{u,b} w_{u,b} \til{x}_{u,t}(b).
	$
	Thus, by \Cref{thm:known_id_relaxation}, it suffices to show that
	\begin{equation} \label{eqn:adversarial_desired_selectibility}
		\mb{P}[ \text{$(u,v_t) \in \scr{M}$ and $v_{t} = b$}] \ge \alpha \cdot p_{u,b} \til{x}_{u,t}(b) 
	\end{equation}
	for each $u \in U, t \in [n]$ and $b \in B$. We may thus assume that $p_{u,b} \til{x}_{u,t}(b) >0$,
	where we note that $\mb{P}[v_{t} =b, Z_{u,t} =1] =  p_{u,b} \cdot \til{x}_{u,t}(b)$.

	Observe first that $(u,v_t) \in \scr{M}$ if and only if
	$(u,v_t)$ is returned by $\psi$ when it executes on $(z_{u,i})_{i=1}^n$ and $(Z_{u,i})_{i=1}^{t}$.
	Thus, since $\psi$ is an $\alpha$-selectable OCRS, by \Cref{def:CRS},
	\begin{equation} \label{eqn:adversarial_contention_id}
		\mb{P}[(u,v_t) \in \scr{M} \mid Z_{u,t} =1] \ge \alpha.
	\end{equation}
	We now show how this implies the stronger claim of \eqref{eqn:adversarial_desired_selectibility}. 
	Observe that whether or not $\psi$ returns $(u,v_t)$  is a
	randomized function of $(z_{u,v})_{u \in U}$ and $(Z_{u,i})_{i=1}^t$. Crucially,
	$\psi$ does \textit{not} make use of the type of $v_i$. As a result, conditional on $(Z_{u,i})_{i=1}^t$, 
	the events $\{(u,v_t) \in \scr{M}\}$ and $\{v_t = b\}$ are independent, and so
	\begin{equation} \label{eqn:conditional_contention_probability_first}
		\mb{P}[(u,v_t) \in \scr{M} \mid v_t = b, (Z_{u,i})_{i=1}^t] = \mb{P}[(u,v_t) \in \scr{M} \mid (Z_{u,i})_{i=1}^t].
	\end{equation}
	Now, taking expectations over $(Z_{u,i})_{i=1}^{t-1}$ in \eqref{eqn:conditional_contention_probability_first}, we get that
	\begin{equation}\label{eqn:conditional_contention_probability}
		\mb{P}[(u,v_t) \in \scr{M}   \mid v_{t} =b, Z_{u,t} =1] = \mb{P}[ (u,v_t) \in \scr{M} \mid  Z_{u,t} =1].
	\end{equation}
	However, $\mb{P}[ (u,v_t) \in \scr{M} \mid  Z_{u,t} =1] \ge \alpha$ by \eqref{eqn:adversarial_contention_id}. Thus,
	$
	\mb{P}[(u,v_t) \in \scr{M}   \mid v_{t} =b, Z_{u,t} =1] \ge \alpha,
	$
	and so \eqref{eqn:adversarial_desired_selectibility} holds by recalling that $\mb{P}[v_{t} =b, Z_{u,t} =1] =  p_{u,b} \cdot \til{x}_{u,t}(b)$, and that $Z_{u,t} =1$ is a necessary condition for $(u,v_t) \in \scr{M}$.
	
	The reduction for random order arrivals proceeds identically, 
	and so we omit the argument.
\end{proof}
For a $1$-uniform matroid, Ezra et al. \cite{Ezra_2020} prove the existence of a $1/2$-selectable OCRS,
and Lee and Singla prove the existence of a $1-1/e$-selectable RCRS. Combining these
results with \Cref{thm:known_id_reduction},
we get the claimed competitive ratios of \Cref{tab:my_label}.

\begin{corollary} \label{cor:known_id_aom}
	In the AOM, there exists a $1/2$-competitive online probing algorithm for the online
	stochastic matching problem with known i.d. instances.
\end{corollary}

\begin{corollary} \label{cor:known_id_rom}
	In the ROM, there exists a $1-1/e$-competitive online probing algorithm for the online
	stochastic matching problem with known i.d. instances.
\end{corollary}

\section{Edge-weighted Worst-case Instance Model}\label{sec:edge_weights}

Let us suppose that $G=(U,V,E)$ is an adversarially generated stochastic graph. Denote the online nodes of $V$ by
$v_{1}, \ldots ,v_{n}$, where the order is generated u.a.r,
and define $V_{t} = \{v_{1}, \ldots ,v_{t}\}$ to be the first $t$ arrivals of $V$. Moreover, set $G_{t}:= G[U \cup V_t]$, and $\LPOPT(G_t)$ as the value of an optimal solution to \ref{LP:new} (this is a random variable, as $V_{t}$ is a random subset of $V$). The following inequality then holds:
\begin{lemma} \label{lem:random_induced_subgraph}
	For each $t \ge 1$,  $\mb{E}[ \LPOPT(G_{t}) ] \ge \frac{t}{n} \, \LPOPT(G)$.
\end{lemma}
In light of this observation, we design an online probing algorithm which makes use of $V_{t}$, the currently known nodes, to derive an optimal LP solution with respect to $G_{t}$. As such, each time an online node
arrives, we must compute an optimal solution for the LP associated to $G_{t}$, distinct from the solution computed
for that of $G_{t-1}$. 
\begin{algorithm}[H]
	\caption{Edge-weighted Worst-case Instance ROM} 
	\label{alg:ROM_edge_weights}
	\begin{algorithmic}[1]
		\Require $U$ and $n:=|V|$.
		\Ensure a matching $\scr{M}$ from the (unknown) stochastic graph $G=(U,V,E)$ of active edges.
		\State Set $\scr{M} \leftarrow \emptyset$.
		\State Set $G_{0} = (U, \emptyset, \emptyset)$
		\For{$t=1, \ldots , n$}
		\State Input $v_{t}$, with $(w_{e})_{e \in \partial(v_t)}$, $(p_{e})_{e \in \partial(v_t)}$ and $\scr{C}_{v_t}$.  
		\State Compute $G_{t}$, by updating $G_{t-1}$ to contain $v_{t}$ (and its relevant information).
		\If{ $t < \lfloor n/e \rfloor$}
		\State Pass on $v_{t}$.
		\Else
		\State Solve \ref{LP:new} for $G_{t}$ and find an optimal solution $(x_{v}(\bm{e}))_{v \in V_{t}, \bm{e} \in \scr{C}_v}$. \label{line:cur_solution}
		\State Set $e_t \leftarrow \mathtt{VertexProbe}(v_t, \partial(v_t), (x_{v}(\bm{e}))_{\bm{e} \in \scr{C}_{v_t}})$.
		\If{$e_t=(u_t,v_t) \neq \emptyset$ and $u_t$ is unmatched}
		\State Add $e_t$ to $\scr{M}$.
		\EndIf
		\EndIf
		\EndFor
		\State \Return $\scr{M}$.
	\end{algorithmic}
\end{algorithm}
\begin{remark}
	Unlike the algorithm of Kesselheim et al., our algorithm is randomized, and
	the polytope \ref{LP:new} does not always have an optimal integral solution. We leave it as an
	interesting open question as to whether or not Algorithm \ref{alg:ROM_edge_weights} can be derandomized.
\end{remark}

\begin{theorem}\label{thm:ROM_edge_weights}
	If $\scr{M}$ is the matching returned by Algorithm \ref{alg:ROM_edge_weights} when executing on $G$,
	then 
	\[
	\mb{E}[w(\scr{M})] \ge \left(\frac{1}{e} - \frac{1}{|V|} \right) \cdot \OPT(G),
	\]
	provided the vertices of $V$ arrive u.a.r.
\end{theorem}

Let us consider the matching $\scr{M}$ returned by the algorithm,
as well as its weight, which we recall is denoted $w(\scr{M})$. Set $\alpha:=1/e$ for clarity, and take $t \ge \lceil \alpha n \rceil$. For each $\alpha n \le t \le n$, define $R_{t}$ as the \textit{unmatched vertices} of $U$ when vertex $v_{t}$ arrives. Note that proposing $v_t$ to $u_t$ is necessary, but not sufficient, for $v_t$ to match to $u_t$. Define $w(e_t):= w_{e_t} \cdot \bm{1}_{[e_t \neq \emptyset]}$.
With this notation, we have that $\mb{E}[ w(\scr{M}) ] = \sum_{t=\alpha n}^{n} \mb{E}[ w(u_t, v_t) \cdot \bm{1}_{[u_t \in R_t]} ]$.
Moreover, we claim the following:
\begin{lemma} \label{lem:edge_value_lower_bound}
	For each  $t \ge \lceil \alpha n \rceil$, $\mb{E}[ w(e_t)] \ge \LPOPT(G)/n$.
\end{lemma}
\begin{lemma} \label{lem:availability_lower_bound}
	For each $t \ge \lceil \alpha n \rceil$, define $f(t,n):= \lfloor \alpha n \rfloor /(t-1)$. In this case, $\mb{P}[ u_{t} \in R_t  \, | \, V_{t},v_t] \ge f(t,n)$,
	where $V_{t}=\{v_{1},\ldots ,v_{t}\}$ and $v_t$ is the $t^{th}$ arriving node of $V$ \footnote{Note that since
		$V_t$ is a set, conditioning on $V_t$ only reveals which vertices of $V$ encompass the first $t$ arrivals,
		\textit{not} the order they arrived in. Hence, conditioning on $v_t$ as well reveals strictly more information.}.
\end{lemma}
The proofs of Lemmas \ref{lem:edge_value_lower_bound} and \ref{lem:availability_lower_bound} mostly follow the analogous claims
as proven by Kesselheim et al. in the secretary matching problem. We present formal proofs in the \Cref{sec:deferred_proofs}. Assuming these lemmas, the proof of Theorem \ref{thm:ROM_edge_weights} follows easily:

\begin{proof}[Proof of Theorem \ref{thm:ROM_edge_weights}]	
Let us consider the matching $\scr{M}$ returned by the algorithm,
	as well as its weight, which we denote by $w(\scr{M})$. Set $\alpha:=1/e$ for clarity, and take $t \ge \lceil \alpha n \rceil$,
	where we define $R_{t}$ to be the \textit{unmatched vertices} of $U$ when vertex $v_{t}$ arrives. Moreover, 
	define $e_{t}:=(u_t,v_{t})$, where $u_{t}$ is the vertex of $U$ which $v_{t}$ proposes to, which is the empty-set by definition if no such proposal
	is made. Observe that
	\begin{equation}\label{eqn:value_of_matching}
		\mb{E}[ w(\scr{M}) ] = \sum_{t=\lceil \alpha n \rceil}^{n} \mb{E}[ w(u_t, v_t) \cdot \bm{1}_{[u_t \in R_t]} ].
	\end{equation}
	Fix $\lceil \alpha n \rceil \le t \le n$, and first observe that $w(u_t, v_t)$ and $\{ u_{t} \in R_t \}$ are conditionally independent given
	$(V_{t},v_t)$, as the probes involving $\partial(v_t)$ are independent from those of $v_{1}, \ldots ,v_{t-1}$. Thus,
	\[
	\mb{E}[ w(u_t,v_t) \cdot \bm{1}_{[u_t \in R_t]}  \, | \, V_{t}, v_{t}] = \mb{E}[ w(u_t,v_t)  \, | \, V_t, v_t] \cdot \mb{P}[ u_t \in R_t  \, | \, V_{t}, v_t].
	\]
	Moreover, \Cref{lem:availability_lower_bound} implies that
	$    
	\mb{E}[ w(u_t,v_t)  \, | \, V_t, v_t] \cdot \mb{P}[ u_t \in R_t  \, | \, V_{t}, v_t] \ge 
	\mb{E}[ w(u_t,v_t)  \, | \, V_{t}, v_t] f(t,n),
	$
	and so $\mb{E}[ w(u_t,v_t) \, \bm{1}_{[u_t \in R_t]}  \, | \, V_{t}, v_{t}] \ge \mb{E}[ w(u_t,v_t)  \, | \, V_{t}, v_t] \, f(t,n)$.
	Thus, by the law of iterated expectations\footnote{$\mb{E}[ w(u_t,v_t) \cdot \bm{1}_{[u_t \in R_t]}  \, | \, V_{t}, v_{t}]$
		is a random variable which depends on $V_t$ and $v_t$, and so the outer expectation is over the randomness in $V_t$ and $v_t$.}
	\begin{align*}
		\mb{E}[ w(u_t,v_t) \cdot \bm{1}_{[u_t \in R_t]} ] &= \mb{E}[ \, \mb{E}[ w(u_t,v_t) \cdot \bm{1}_{[u_t \in R_t]}  \, | \, V_{t}, v_{t}] \, ] \\
		&\ge \mb{E}[ \, \mb{E}[ w(u_t,v_t)  \, | \, V_{t}, v_t]  f(t,n) \, ] 
		= f(t,n) \mb{E}[ w(u_t,v_t)].
	\end{align*}
	As a result, using \eqref{eqn:value_of_matching}, 
	$
		\mb{E}[w(\scr{M})] = \sum_{t=\lceil \alpha n \rceil}^{n} \mb{E}[ w(u_t, v_t) \, \bm{1}_{[u_t \in R_t]} ]
		\ge \sum_{t=\lceil \alpha n \rceil}^{n} f(t,n) \, \mb{E}[ w(u_t,v_t)].
	$
	We may thus conclude that
	\[
	\mb{E}[ w(\scr{M})] \ge \LPOPT(G) \sum_{t=\lceil \alpha n \rceil}^{n} \frac{ f(t,n)}{n},
	\]
	after applying Lemma \ref{lem:edge_value_lower_bound}. 
	As $\sum_{t=\lceil \alpha n \rceil}^{n} f(t,n)/n \ge (1/e -1/n)$, the result holds.
\end{proof}

\section{Vertex-weighted Worst-case Instance Model} \label{sec:vertex_weights}
In this section,
we analyze a \textit{greedy} online probing algorithm when the stochastic graph $G=(U,V,E)$ has (offline) vertex weights $(w_u)_{u \in U}$.
Upon the arrival of $s$,
the probes to $\partial(s)$ are made in such a way that $s$ gains as large
a match as possible (in expectation), provided the currently unmatched nodes of $U$
are equal to $R \subseteq U$. As such, we must follow the probing strategy
of $\OPT$ when restricted to the \textbf{induced stochastic graph}\footnote{Given $R \subseteq U, V' \subseteq  V$, the induced stochastic graph 
	$G[R \cup V']$ is formed by restricting the edges weights and probabilities of $G$ to those edges within $R \times V'$.
	Similarly, each probing constraint $\scr{C}_s$ is restricted to those strings whose entries lie entirely
	in $R \times \{s\}$.} $G[\{s\} \cup R]$. We denote the performance of $\OPT$ on $G[\{s\} \cup R]$ by
$\OPT(R,s)$.

Assume that $R=U$, and that $w_{u,s}:=w_{u}$ for each $u \in U$
such that $(u,s) \in \partial(s)$.
Recall that for $\bm{e} \in \scr{C}_s$, 
\begin{equation}\label{eqn:expected_value_of_edge_probes}
	\val(\bm{e}):= \sum_{i=1}^{|\bm{e}|} p_{e_i} w_{e_i} \prod_{j=1}^{i-1} (1 - p_{e_i}).
\end{equation}
Observe that if one probes the edges of $\bm{e}$ in order,
then $\val(\bm{e})$ is the expected weight of the first active edge
of $\bm{e}$. This is a \textbf{non-adaptive} probing algorithm for
$G[ U \cup \{s\}]$. Since $\OPT$ operates in the probe-commit model,
there exists a non-adaptive probing algorithm with (optimal) performance $\OPT(s,U)$. Moreover, we show such a strategy can be found efficiently
assuming that $\scr{C}_s$ is downward-closed.
\begin{theorem}\label{thm:efficient_star_dp}
	There exists a dynamic programming (DP) based algorithm \textsc{DP-OPT}, which
	given access to $G[ \{s\} \cup U]$, computes a tuple $\bm{e}' \in \scr{C}_{s}$, such
	that $\OPT(s,U) = \val(\bm{e}')$. Moreover, \textsc{DP-OPT} executes in time $O(|U|^{2})$, assuming access to a membership oracle for 
	the downward-closed constraint $\scr{C}_s$.
\end{theorem}
\begin{proof}
	Our goal is to show that $\bm{e}'$ can be computed efficiently.
	Now, for any $\bm{e} \in \scr{C}_s$, let $\bm{e}^{r}$ be the rearrangement of $\bm{e}$, based on the non-increasing order
	of the weights $(w_{e})_{e \in \bm{e}}$. Since $\scr{C}_s$ is downward-closed,
	we know that $\bm{e}^{r}$ is also in $\scr{C}_s$. Moreover, $\val(\bm{e}^{r}) \ge \val(\bm{e})$
	(following observations in \cite{Purohit2019,Brubach2019}).
	Hence, let us order the edges of $\partial(s)$ as $e_{1}, \ldots ,e_{m}$,
	such that $w_{e_1} \ge \ldots \ge  w_{e_m}$, where $m:=|\partial(s)|$.
	Observe then that it suffices to maximize \eqref{eqn:expected_value_of_edge_probes} over
	those strings within $\scr{C}_s$ which respect this ordering on $\partial(s)$.
	Stated differently, let us denote $\scr{I}_{s}$ as the family of subsets of $\partial(s)$
	induced by $\scr{C}_s$, and define the set function $f: 2^{\partial(s)} \rightarrow [0, \infty)$,
	where $f(B):= \val(\bm{b})$ for $B=\{b_{1}, \ldots ,b_{|B|}\} \subseteq \partial(s)$,
	such that $\bm{b}=(b_{1}, \ldots ,b_{|B|})$ and $w_{b_1} \ge \ldots \ge w_{b_{|B|}}$.
	Our goal is then to efficiently maximize $f$ over the set-system $(\partial(s),\scr{I}_s)$.
	Observe that $\scr{I}_s$ is downward-closed and that we can simulate oracle access to
	$\scr{I}_s$, based on our oracle access to $\scr{C}_s$.

	For each $i=0, \ldots ,m-1$, denote $\partial(s)^{>i}:=\{e_{i+1}, \ldots ,e_{m}\}$,
	and $\partial(s)^{>m}:= \emptyset$. Moreover, define the family of subsets $\scr{I}_{s}^{>i}:= \{B \subseteq \partial(s)^{>i} : B \cup \{e_i\} \in \scr{I}_s\}$ for each $1 \le i \le m$, and $\scr{I}_{s}^{>0}:= \scr{I}_s$. Observe then that
	$(\partial(s)^{>i}, \scr{I}_{s}^{>i})$ is a downward-closed set system, as $\scr{I}_s$ is downward-closed.
	Moreover, we may simulate oracle access to $\scr{I}^{>i}_{s}$ based on our oracle access to $\scr{I}_s$.

	Denote $\OPT(\scr{I}_{s}^{>i})$ as the maximum value of $f$ over constraints $\scr{I}_{s}^{>i}$.
	Observe then that for each $0 \le i \le m-1$, the following recursion holds:
	\begin{equation} \label{eqn:dynamical_program}
		\OPT(\scr{I}^{>i}_{s}) :=  
		\max_{j \in \{i+1,\ldots,m\}} ( p_{e_j} \cdot w_{e_j} + (1 - p_{e_j}) \cdot \OPT(\scr{I}_{s}^{>j}) )
	\end{equation}
	Hence, given access to the values $\OPT(\scr{I}_{s}^{>i+1}), \ldots , \OPT(\scr{I}_{s}^{>m})$,
	we can compute $\OPT(\scr{I}^{>i}_s)$ efficiently. Moreover, $\OPT(\scr{I}_{s}^{>m})=0$ by definition. 
	Thus, it is clear that we can use \eqref{eqn:dynamical_program} to recover an optimal solution to $f$.
	We can define \textsc{DP-OPT} to be a memoization based implementation of \eqref{eqn:dynamical_program}.
	It is clear \textsc{DP-OPT} can be implemented in the claimed time complexity.
\end{proof}

Given $R \subseteq U$, denote
the output of executing \textsc{DP-OPT} on $G[\{s\} \cup R]$ by $\textsc{DP-OPT}(s,R)$. Consider now the following online probing
algorithm:
\begin{algorithm}
	\caption{Greedy-DP}\label{alg:dynamical_program}
	\begin{algorithmic}[1] 
		\Require offline vertices $U$ with vertex weights $(w_{u})_{u \in U}$.
		\Ensure a matching $\scr{M}$ of active edges of the unknown stochastic graph $G=(U,V,E)$.
		\State $\scr{M} \leftarrow \emptyset$.
		\State $R \leftarrow U$.
		\For{$t=1, \ldots , n$}
		\State Let $v_{t}$ be the current online arrival node, with constraint $\scr{C}_{v_t}$.
		\State Set $\bm{e}  \leftarrow \textsc{DP-OPT}(v_t,R)$
		\For{$i=1, \ldots , |\bm{e}|$}
		\State Probe $e_i$.
		\If{$\sta(e_i) =1$}
		\State Add $e_i$ to $\scr{M}$, and update $R \leftarrow R \setminus \{u_i\}$,
		where $e_{i}=(u_i,v_t)$.
		\EndIf
		\EndFor
		\EndFor
		\State \Return $\scr{M}$.
	\end{algorithmic}
\end{algorithm}

\begin{theorem}\label{thm:adversarial}
	If $\scr{M}$ is the matching returned by Algorithm \ref{alg:dynamical_program} when executing on vertex-weighted
	stochastic graph $G$,
	then 
	\[
	\mb{E}[w(\scr{M})] \ge \frac{1}{2} \cdot \OPT(G),
	\]
	provided the vertices of $V$ arrive in adversarial order.
\end{theorem}

In order to analyze \Cref{alg:dynamical_program}, we first upper bound $\OPT(G)$ using an LP relaxation which extends the LP from \cite{Brubach2019}
to online probing constraints. For each $u \in U$ and $v \in V$, let $x_{u,v}$ be a decision variable corresponding to the probability that $\OPT$ probes the edge $(u,v)$.
\begin{align}\label{LP:DP}
	\tag{LP-DP}
	& \text{maximize}  & \sum_{u \in U} \sum_{v \in V} w_u \cdot p_{u, v}  \cdot x_{u, v}\\
	&\text{subject to} &\, \sum_{v \in V}  p_{u, v} \cdot x_{u, v} &\leq 1 &&\forall u \in U       \label{eqn:dp_matching_constraint} \\
	&  &\sum_{u \in R} w_u \cdot p_{u,v} \cdot x_{u,v} & \le  \OPT(v,R)  && \forall v \in V, \, R \subseteq U	\label{eqn:OPT_constraint}\\ 
	&  &x_{u, v} & \ge 0 && \forall u \in U, v \in V		\label{eqn:edge_upper_bound}
\end{align}
Denote $\dLPOPT(G)$ as the optimal value of this LP. 
Constraint \eqref{eqn:dp_matching_constraint} can be viewed as ensuring that
the expected number of matches made to $u \in U$ is at most $1$. Similarly,
\eqref{eqn:OPT_constraint} can be interpreted as ensuring
that expected stochastic reward of $v$, suggested by the solution $(x_{u,v})_{u \in U,v \in V}$,
is actually attainable by the adaptive benchmark. More precisely, given $R \subseteq U$
and $v \in V$, if one restricts the match of $v$ made by $\OPT$ to $R \times \{v\}$, then its expected value cannot exceed $\OPT(v,R)$.
This implies the following (see \cite{Brubach2019} for a detailed proof specific to patience values):
\begin{lemma} \label{lem:DP_relaxation}
	$\OPT(G) \le \dLPOPT(G)$ 
\end{lemma}

In \cite{Brubach2019}, the performance of \Cref{alg:dynamical_program} is compared to a solution to \ref{LP:DP} to prove
the algorithm is $1/2$-competitive for patience values. It is easily verified that this argument extends to online probing constraints.
We instead prove \Cref{thm:adversarial} using a primal-dual charging argument based 
on the dual of \ref{LP:DP}, as this allows us to introduce the techniques needed for the analysis of
\Cref{alg:dynamical_program} in the ROM.

For each $u \in U$, define the variable $\alpha_u$. Moreover, for each $R \subseteq U$
and $v \in V$, define the variable $\phi_{v,R}$ (these latter variables correspond to constraint \eqref{eqn:OPT_constraint}):
\begin{align} \label{LP:DP_dual}
	\tag{LP-dual-DP}
	&\text{minimize}&   \sum_{u \in U} \alpha_u + \sum_{v \in V}\sum_{R \subseteq U}\OPT(v,R) \cdot \phi_{v,R}	\\
	&\text{subject to}&  \: p_{u, v} \cdot \alpha_u  + \sum_{\substack{R \subseteq U: \\ u \in R}} w_u \cdot p_{u,v} \cdot \phi_{v,R} &\geq  w_u \cdot p_{u, v} && \forall u \in U, v \in V\\
	& &   \alpha_u &\geq 0 && \forall u \in U\\
	& &     \phi_{v,R} & \ge 0 && \forall v \in V, R \subseteq U
\end{align}
Let $((\alpha_u)_{u \in U}, (\phi_{v,R})_{v \in V, R \subseteq U})$ be a dual solution which is initially identically $0$.
As \Cref{alg:dynamical_program} executes, we modify it in the following way. Fix $v \in V, u \in U$, and $R \subseteq U$, where $u \in R$.
If $R$ consists of the unmatched vertices when $v$ arrives,
then suppose that Algorithm \ref{alg:dynamical_program} matches $v$ to $u$ while making its probes to
a subset of the edges of $R \times \{v\}$. In this case, we \textbf{charge} $ w_{u}$ to $\alpha_{u}$ and
$w_{u}/ \OPT(v,R)$ to $\phi_{v,R}$.
Observe that each subset $R \subseteq U$ is charged at most once, as is each $u \in U$.
Moreover,
\begin{equation}\label{eqn:expected_competitive_ratio_dual_general_adv}
	\mb{E}[ w( \scr{M})] = \frac{1}{2} \cdot \left( \sum_{u \in U} \mb{E}[\alpha_{u}]  +  \sum_{v \in V} \sum_{R \subseteq U} \OPT(v,R) \cdot \mb{E}[ \phi_{v,R}] \right),
\end{equation}
where the expectation is over $(\sta(e))_{e \in E}$. Let us now set $\alpha^{*}_{u} := \mb{E}[\alpha_{u} ]$ and $\phi^{*}_{v,R} := \mb{E}[ \phi_{v,R}]$ for $u \in U, v \in V$ and $R \subseteq U$. We prove the following in \Cref{sec:deferred_proofs}:

\begin{lemma} \label{lem:adv_feasibility}
	$((\alpha^*_u)_{u \in U}, (\phi^*_{v,R})_{v \in V, R \subseteq U})$ 
	is a feasible solution  to \ref{LP:DP_dual}.
\end{lemma}
\Cref{thm:adversarial} then follows immediately from \eqref{eqn:expected_competitive_ratio_dual_general_adv}, \Cref{lem:DP_relaxation} and \Cref{lem:adv_feasibility}.

\subsection{Analyzing \Cref{alg:dynamical_program} in the ROM}
In general, the behaviour of $\OPT$ on $G[\{s\}\cup R]$ can change very much, even for minor changes to $R$.
For instance, if $R= U$, then $\textsc{OPT}$ may probe the 
edge $(u,s)$ first -- thus giving it highest priority -- whereas if
$u^* \in U$ is removed from $U$ (where $u^* \neq u$),
$\textsc{OPT}$ may not probe $(u,v)$ at all. 
See the below example for an explicit instance of this behaviour.
\begin{example} \label{example:bad_OPT_behaviour}
	Let $G = (U, V, E)$ be a bipartite graph with $U = \{u_1, u_2, u_3, u_4\}$, $V = \{v\}$ and $\ell_v =2$. 
	Set $p_{u_1, v} = 1/3$, $p_{u_2, v} = 1$, $p_{u_3, v} = 1/2$, $p_{u_4, v} = 2/3$.
	Fix $\eps > 0$, and let the weights of offline vertices be $w_{u_1} = 1 + \eps$, $w_{u_2} = 1 + \eps/2$, $w_{u_3} = w_{u_4} = 1$. 
	We assume that $\eps$ is sufficiently small -- concretely, $\eps \le 1/12$.
	If $R_{1}:= U$, then $\OPT$ probes $(u_{1},v)$ and then $(u_{2},v)$ in order. On the other hand,
	if $R_{2} = R_{1} \setminus \{u_{2}\}$, then $\OPT$ does \textit{not} probe $(u_{1},v)$. Specifically,
	$\OPT$ probes $(u_3,v)$ and then $(u_4,v)$.
\end{example}
Using \Cref{example:bad_OPT_behaviour}, it is easy to consider an
execution of \Cref{alg:dynamical_program} on $G$
where $v$ is matched to $u$, 
but if a new vertex $v^*$ is added to $G$ ahead of $v$,
$u$ is never matched. We thus refer to Algorithm \ref{alg:dynamical_program}
as being \textbf{non-monotonic}. This contrasts with the classical setting,
in which the deterministic greedy algorithm in the ROM setting does not exhibit this
behaviour, and thus is \textbf{monotonic}. The absence of monotonicity isn't problematic 
in the adversarial setting of Theorem \ref{thm:adversarial} because our charging argument
does not depend on the order of the online vertex arrivals. On the other hand, in the ROM we construct
a solution to \ref{LP:DP_dual} using a function $g:[0,1] \rightarrow [0,1]$ which depends on the arrival time of each online node.
In order for this solution to be feasible in expectation, we (provably) must restrict to collections of stochastic graphs in which executions of Algorithm \ref{alg:dynamical_program} are monotonic. This leads us to the definition of \textit{rankability}, which characterizes a large number of such graphs.

Given a vertex $v \in V$, and an ordering $\pi_v$ on $\partial(v)$,
if $R \subseteq U$, then define $\pi_{v}(R)$ to be the longest string constructible by iteratively appending the edges of $R \times \{v\}$ via $\pi_v$, subject to respecting constraint $\scr{C}^{R}_v$. More precisely, given $\bm{e}'$ after processing $e_1,\ldots ,e_i$
of $R \times \{v\}$ ordered according to $\pi_v$, if $(\bm{e}',e_{i+1}) \in \scr{C}_{v}^{R}$, then update $\bm{e}'$ by appending $e_{i+1}$ to its end, otherwise move to the next edge $e_{i+2}$ in the ordering $\pi_v$, assuming $i+2 \le |R|$. If $i+2 > |R|$,
return the current string $\bm{e}'$ as $\pi_{v}(R)$. We say that $v$ is \textbf{rankable}, provided there exists a choice of $\pi_{v}$ which depends \textit{solely} on $(p_e)_{e \in \partial(v)}$, $(w_{e})_{e \in \partial(v)}$ and $\scr{C}_v$, such that for \textit{every} $R \subseteq U$, the strings $\textsc{DP-OPT}(v,R)$ and $\pi_{v}(R)$
are equal. Crucially, if $v$ is rankable, then when vertex $v$ arrives while executing Algorithm \ref{alg:dynamical_program}, one can compute the ranking $\pi_{v}$ on $\partial(v)$ and probe the adjacent edges of $R \times \{v\}$  based on this order, subject to not violating the constraint $\scr{C}_{v}^{R}$. By following this probing strategy, the optimality of \textsc{DP-OPT} ensures that the expected weight of the match made to $v$ will be $\OPT(v,R)$.
We consider three (non-exhaustive) examples of rankability:
\begin{proposition} \label{prop:rankability}
	Let $G=(U,V,E)$ be a stochastic graph, and suppose that $v \in V$.
	If $v$ satisfies either of the following conditions, then $v$ is rankable:
	\begin{enumerate}
		\item $v$ has unit patience or unlimited patience; that is, $\ell_v \in \{1, |U|\}$. \label{eqn:rankability_examples_patience}
		\item $v$ has patience $\ell_v$, and for each $u_{1}, u_{2} \in U$, if $p_{u_{1},v} \le p_{u_{2},v}$ then $w_{u_{1}} \le w_{u_{2}}$.
		\item $G$ is unweighted, and $v$ has a budget $B_v$ with edge probing costs $(b_{u,v})_{u \in U}$, and for each $u_1, u_2 \in U$, if $p_{u_1,v} \le p_{u_{2},v}$ then $b_{u_1,v} \ge b_{u_{2},v}$.
	\end{enumerate}
\end{proposition}
The rankable assumption is similar to assumptions referred to as laminar, agreeable and compatible in other applications.  We refer to the stochastic graph $G$ as \textbf{rankable}, provided all of its
online vertices are themselves rankable. We emphasize that distinct vertices of $V$
may each use their own separate rankings of their adjacent edges. Our next result proves a $1-1/e$
competitive ratio $G$ is rankable, as well as an asymptotic competitive ratio of $1-1/e$
provided the edge probabilities of $G$ tend to $0$ sufficiently fast (as $|G| \rightarrow \infty$).

\begin{theorem}\label{thm:ROM_rankable}
	Suppose Algorithm \ref{alg:dynamical_program} returns the matching $\scr{M}$ when
	executed on the vertex-weighted stochastic graph $G=(U,V,E)$ with random order vertex
	arrivals.
	\begin{enumerate}
		\item If $G$ is rankable, then 
		\[
		\mb{E}[ w(\scr{M}) ] \ge  \left(1 - \frac{1}{e} \right) \cdot \OPT(G).
		\]
		\item If $c_v :=\max_{\bm{e} \in \scr{C}_v}|\bm{e}|$ and $p_v := \max_{e \in \partial(v)}p_e$, then
		\[
		\mb{E}[ w(\scr{M}) ] \ge  \min_{v \in V}(1-p_v)^{c_v}\cdot \left(1 - \frac{1}{e} \right)\cdot \OPT(G).
		\]
		Thus, if $\max_{v \in V} c_v \cdot p_v \rightarrow 0$,
		then $\mb{E}[ w(\scr{M}) ] \ge (1- o(1))\left(1 - 1/e \right) \cdot \OPT(G)$.
	\end{enumerate}
\end{theorem}
\begin{remark}
	The analysis of Algorithm \ref{alg:dynamical_program} is tight, as an execution of Algorithm \ref{alg:dynamical_program} corresponds to the seminal Karp et al. \cite{KarpVV90} \textsc{Ranking} algorithm for unweighted non-stochastic (i.e., $p_e  \in \{0,1\}$ for all $e \in E$) bipartite  matching. 
\end{remark}

\begin{corollary} \label{cor:ROM_unit_patience}
	Suppose $G=(U,V,E)$ is a vertex-weighted stochastic graph  
	with unit patience values. If $\scr{M}$
	is the matching returned by Algorithm \ref{alg:dynamical_program} when executing on $G$,
	then 
	\[
	\mb{E}[w(\scr{M})] \ge \left(1- \frac{1}{e}\right) \OPT(G),
	\]
	provided the vertices of $V$ arrive in random order.
\end{corollary}
\begin{remark}
	The guarantee of Theorem \ref{thm:ROM_rankable} is proven against $\dLPOPT(G)$.
	In the special case when $G$ has unit patience, $\sLPOPT(G) \le \dLPOPT(G)$, where $\sLPOPT(G)$
	is the value of an optimal solution to \ref{LP:standard_benchmark} on $G$. Thus, \Cref{cor:ROM_unit_patience} implies
	that the $0.621 <1 -1/e$ in-approximation of Mehta and Panigraphi \cite{MehtaP12} against \ref{LP:standard_benchmark} does not apply to the ROM setting. 
\end{remark}

\subsection{Proving \Cref{thm:ROM_rankable}} \label{sec:charging_scheme}
The dual-fitting argument used to prove Theorem \ref{thm:ROM_rankable} has
an initial set-up which based on the argument
in Devanur et al. \cite{DJK2013}. First define $g: [0,1] \rightarrow [0,1]$ 
where $g(z):= \exp(z-1)$ for $z \in [0,1]$. We shall use $g$ to perform our charging. Moreover,
recall that given $v \in V$, we defined $c_{v}:= \max_{\bm{e} \in \scr{C}_{v}}|\bm{e}|$ and $p_{v}:=\max_{e \in \partial(v)}p_e$.
Using these definitions, we define our target competitive ratio $F=F(G)$, where
\begin{equation}\label{eqn:target_competitive_ratio}
	F(G):=
	\begin{cases} 
		1-\frac{1}{e} &\text{if $G$ is rankable,}\\
		\left(1-\frac{1}{e}\right)\cdot \min_{v \in V}(1-p_v)^{c_v} &\text{otherwise.}\\
	\end{cases}
\end{equation}
In order to prove Theorem \ref{thm:ROM_rankable}, we shall prove that Algorithm \ref{alg:dynamical_program} 
returns a matching of expected weight at least $F(G) \cdot \dLPOPT(G)$ when executing on the stochastic graph $G$
in the ROM.  Clearly, we may assume $F(G)> 0$, as otherwise there is nothing to prove,
so we shall make this assumption for the rest of the section. Note that $F(G) \le 1-1/e$
no matter the stochastic graph $G$.

For each $v \in V$, draw $Y_{v} \in [0,1]$ independently and uniformly at random. We assume that the vertices of $V$ are presented to Algorithm \ref{alg:dynamical_program} in a non-decreasing order, based on the values of $(Y_{v})_{v \in V}$. 
We now describe how the charging assignments are made while Algorithm \ref{alg:dynamical_program} executes
on $G$.  First, we initialize a dual solution $((\alpha_u)_{u \in U}, (\phi_{v,R})_{v \in V, R \subseteq U})$ where all the variables are set equal to $0$. Next, we take $v \in V, u \in U$, and $R \subseteq U$, where $u \in R$.
If $R$ consists of the unmatched vertices of $v$ when it arrives at time $Y_{v}$,
then suppose that Algorithm \ref{alg:dynamical_program} matches $v$ to $u$ while making its probes to
a subset of the edges of $R \times \{v\}$. In this case, we \textbf{charge} $w_{u} \cdot (1 - g(Y_{v}))/F$ to $\alpha_{u}$ and
$w_{u} \cdot g(Y_{v})/ (F \cdot \OPT(v,R))$ to $\phi_{v,R}$.
Observe that each subset $R \subseteq U$ is charged at most once, as is each $u \in U$.
Thus,
\begin{equation}\label{eqn:expected_competitive_ratio_dual_general}
	\mb{E}[ w( \scr{M})] = F \cdot \left( \sum_{u \in U} \mb{E}[\alpha_{u}]  +  \sum_{v \in V} \sum_{R \subseteq U} \OPT(v,R) \cdot \mb{E}[ \phi_{v,R}] \right),
\end{equation}
where the expectation is over the random variables $(Y_{v})_{v \in V}$
and $(\sta(e))_{e \in E}$. If we now set $\alpha^{*}_{u} := \mb{E}[\alpha_{u} ]$ and $\phi^{*}_{v,R} := \mb{E}[ \phi_{v,R}]$ for $u \in U, v \in V$ and $R \subseteq U$, then \eqref{eqn:expected_competitive_ratio_dual_general} implies the following lemma:

\begin{lemma} \label{lem:dual_variables_expected_value}
	Suppose $G=(U,V,E)$ is a stochastic graph for which
	Algorithm \ref{alg:dynamical_program} returns the matching $\scr{M}$
	when presented $V$ based on $(Y_v)_{v \in V}$ generated $u.a.r.$ from $[0,1]$.
	In this case, if the variables $((\alpha^*_u)_{u \in U}, (\phi^*_{v,R})_{v \in V, R \subseteq U})$ 
	are defined through the above charging scheme, then
	\[
	\mb{E}[ w( \scr{M})] = F \cdot \left( \sum_{u \in U} \alpha^{*}_{u}  +  \sum_{v \in V} \sum_{R \subseteq U} \OPT(v,R) \cdot \phi_{v,R}^{*} \right).
	\]
\end{lemma}
We claim the following regarding $((\alpha^*_u)_{u \in U}, (\phi^*_{v,R})_{v \in V, R \subseteq U})$:

\begin{lemma} \label{lem:dual_feasibility_constraints_general}
	If the online nodes of $G=(U,V,E)$ are presented to Algorithm \ref{alg:dynamical_program}
	based on $(Y_v)_{v \in V}$ generated $u.a.r.$ from $[0,1]$,
	then the solution $((\alpha^*_u)_{u \in U}, (\phi^*_{v,R})_{v \in V, R \subseteq U})$ 
	is a feasible solution  to \ref{LP:DP_dual}.
	
\end{lemma}  
Since \ref{LP:DP} is a relaxation of the adaptive benchmark,
Theorem \ref{thm:ROM_rankable} follows from Lemmas \ref{lem:dual_variables_expected_value}
and \ref{lem:dual_feasibility_constraints_general} in conjunction with weak duality.

\subsubsection{Proving Dual Feasibility: Lemma \ref{lem:dual_feasibility_constraints_general}}

Let us suppose that the 
variables $((\alpha_u)_{u \in U}, (\phi_{v,R})_{v \in V, R \subseteq U})$ 
are defined as in the charging scheme of Section \ref{sec:charging_scheme}.
In order to prove Lemma \ref{lem:dual_feasibility_constraints_general},
we must show that for each fixed $u_0 \in U$ and $v_0 \in V$, we have that
\begin{equation}\label{eqn:dual_feasibility_fixed_vertices}
	\mb{E}[ p_{u_0,v_0} \cdot \alpha_{u_0} +   w_{u_0} \cdot p_{u_0,v_0} \, \sum_{\substack{R \subseteq U: \\ u_0 \in R}} \phi_{v_0,R}] \ge w_{u_0} \cdot p_{u_0,v_0}.
\end{equation}
Our strategy for proving \eqref{eqn:dual_feasibility_fixed_vertices}
first involves the same approach as used in Devanur et al. \cite{DJK2013}.
Specifically, we define the stochastic graph $\til{G}:=(U, \til{V}, \til{E})$,
where $\til{V}:= V \setminus \{v_0\}$ and $\til{G}:=G[U \cup \til{V}]$. We wish to compare the execution of the algorithm on the instance $\til{G}$ to its execution on the instance $G$. It will be convenient to couple the randomness
between these two executions by making the following assumptions:
\begin{enumerate} \label{eqn:greedy_algorithm_coupling_general}
	\item For each $e \in \til{E}$, $e$ is active in $\til{G}$ if and only if it is active in $G$.
	\item The same random variables, $(Y_{v})_{v \in \til{V}}$, are used in both executions.
\end{enumerate}
If we now focus on the execution of $\til{G}$, then define the random variable
$\til{Y}_c$ where $\til{Y}_c:=Y_{v_c}$ 
if $u_0$ is matched to some $v_c \in \til{V}$, 
and $\til{Y}_c:=1$ if $u_0$ remains unmatched after the execution on $\til{G}$.
We refer to the random variable $\til{Y}_c$ as the \textbf{critical time} of
vertex $u_0$ with respect to $v_0$. We claim the following lower bounds on $\alpha_{u_0}$ in terms of the critical time $\til{Y}_{c}$. 
\begin{proposition}~ \label{prop:offline_dual_variable_lower_bound_general}
	\begin{itemize}
		\item If $G$ is rankable, then $\alpha_{u_0} \ge \left(1- \frac{1}{e}\right)^{-1} w_{u_0}(1 - g(\til{Y}_c))$.
		\item Otherwise, $\mb{E}[\alpha_{u_0} \, | \, (Y_{v})_{v \in V}, (\sta(e))_{e \in \til{E}}] \ge \left(1- \frac{1}{e}\right)^{-1} w_{u_0}(1 - g(\til{Y}_c))$.
	\end{itemize}
	
\end{proposition}

\begin{remark}
	Note that the proof of Proposition \ref{prop:offline_dual_variable_lower_bound_general} is the only part of the proof of Theorem \ref{thm:ROM_rankable} which depends on whether or not $G$ is rankable.
\end{remark}
\begin{proof}[Proof of Proposition \ref{prop:offline_dual_variable_lower_bound_general}]
	For each $v \in V$, denote $R^{\text{af}}_{v}(G)$ as the unmatched (remaining) vertices of
	$U$ right after $v$ is processed (attempts its probes) in the execution on $G$. We emphasize that if a probe of $v$
	yields an active edge, thus matching $v$, then this match is excluded from $R^{\text{af}}_{v}(G)$.
	Similarly, define $R^{\text{af}}_{v}(\til{G})$ in the same way for the execution on $\til{G}$
	(where $v$ is now restricted to $\til{V}$).
	
	We first consider the case when $G$ is rankable, and so $F(G) =1-1/e$. Observe that since the constraints $(\scr{C}_v)_{v \in V}$ are substring-closed,
	we can use the coupling between the two executions to inductively prove that
	\begin{equation}\label{eqn:monotonicity_general}
		R^{\text{af}}_{v}(G) \subseteq R^{\text{af}}_{v}(\til{G}),
	\end{equation}
	for each $v \in \til{V}$ \footnote{Example \ref{example:bad_OPT_behaviour} shows why \eqref{eqn:monotonicity_general} will not hold if $G$ is not rankable.}. Now, since $g(1)=1$ (by assumption), there is nothing to prove if
	$\til{Y}_{c}=1$. Thus, we may assume that $\til{Y}_c < 1$, and as a consequence,
	that there exists some vertex $v_{c} \in V$ which matches to $u_0$ at time $\til{Y}_c$
	in the execution on $\til{G}$.
	
	On the other hand, by assumption we know that $u_0 \notin R^{\text{af}}_{v_c}(\til{G})$
	and thus by \eqref{eqn:monotonicity_general}, that $u_0 \notin R^{\text{af}}_{v_c}(G)$.
	As such, there exists some $v' \in V$ which probes $(u_0,v')$ and succeeds
	in matching to $u_0$ at time $Y_{v'} \le \til{Y}_c$. Thus,
	since $g$ is monotone,
	\[
	\alpha_{u_0}  \ge \left(1- \frac{1}{e}\right)^{-1} w_{u_0} \cdot (1 - g(Y_{v'})) \cdot \bm{1}_{[ \til{Y}_c < 1]} \ge 
	\left(1- \frac{1}{e}\right)^{-1} w_{u_0} \cdot (1 - g(\til{Y}_c)),
	\]
	and so the rankable case is complete.
	
	We now consider the case when $G$ is not rankable. Suppose that
	$\scr{M}(v_0)$ is the vertex matched to $v_0$ when the algorithm executes on 
	$G$, where $\scr{M}(v_0):=\emptyset$ provided no match is made. Observe
	then that if no match is made to $v_0$ in this execution, then
	the execution proceeds identically to the execution on $\til{G}$.
	As a result, we get the following relation:
	$
	\alpha_{u_0} \ge \frac{w_{u_0}}{F} (1 - g(\til{Y}_c)) \cdot \bm{1}_{[\scr{M}(v_0)=\emptyset]}.
	$
	Now, let us condition on $(\sta(e))_{e \in \til{E}}$ and $(Y_v)_{v \in V}$,
	and recall the definitions of $p_{v_0}:=\max_{e \in \partial(v_0)}p_e$ and $c_{v_0}:=\max_{\bm{e} \in \scr{C}_{v_0}} |\bm{e}|$. 
	Observe that if every probe involving an edge of $\partial(v_0)$ is inactive,
	then $\scr{M}(v_0)=\emptyset$. On the other hand, each probe independently fails with probability
	at least $(1-p_{v_0})$, and there are at most $c_{v_0}$ probes made to $\partial(v_0)$.
	Thus,
	\[
	\mb{P}[\scr{M}(v_0)=\emptyset \, | \,  (\sta(e))_{e \in \til{E}} , (Y_v)_{v \in V}] \ge (1-p_{v_0})^{c_{v_0}}
	\]
	Now, since $F(G)= (1-1/e) \cdot \min_{v \in V} (1-p_v)^{c_v}$, we get that
	\[
	\mb{E}[\alpha_{u_0} \, | \, (Y_v)_{v \in V} , (\sta(e))_{e \in \til{E}} ] \ge \left(1- \frac{1}{e}\right)^{-1} w_{u_0} (1 - g(\til{Y}_c)),
	\]
	and so the proof is complete.
\end{proof}

By taking the appropriate conditional expectation, we can also 
lower bound the random variables $(\phi_{v_{0},R})_{\substack{R \subseteq U: \\ u_{0} \in R} }$.
\begin{proposition} \label{prop:online_dual_variable_lower_bound_general}
	$\sum_{\substack{R \subseteq U: \\ u_{0} \in R}} \mb{E}[ \phi_{v_{0},R} \, | \, (Y_{v})_{v \in \til{V}}, (\sta(e))_{e \in \til{E}}] \ge   \frac{1}{F}\int_{0}^{\til{Y}_c} g(z) \, dz.$
\end{proposition}

\begin{proof}[Proof of Proposition \ref{prop:online_dual_variable_lower_bound_general}]
	We first define $R_{v_0}$ as
	the unmatched vertices of $U$ when $v_0$ arrives (this is a random subset of $U$).
	We also once again use $\scr{M}$ to denote the matching returned by Algorithm \ref{alg:dynamical_program}
	when executing on $G$. If we now take a \textit{fixed} subset $R \subseteq U$, then
	the charging assignment to $\phi_{v_0,R}$
	ensures that
	\[
	\phi_{v_0,R} = w( \scr{M}(v_0)) \cdot \frac{g(Y_{v_0})}{F \cdot \OPT(v_0,R)} \cdot \bm{1}_{[ R_{v_0} = R]},
	\]
	where $w( \scr{M}(v_0))$ corresponds to the weight of
	the vertex matched to $v_0$ (which is zero if $v_0$ remains unmatched after the execution on $G$).
	In order to make use of this relation,
	let us first condition on the values of $(Y_{v})_{v \in V}$, as well as the 
	states of the edges of $\til{E}$; that is, $(\sta(e))_{e \in \til{E}}$. Observe
	that once we condition on this information, we can determine $g(Y_{v_0})$,
	as well as $R_{v_0}$. As such,
	\[
	\mb{E}[ \phi_{v_0, R} \, | \, (Y_{v})_{v \in V}, (\sta(e))_{e \in \til{E}}] = \frac{g(Y_{v_0})}{F \cdot \OPT(v_0,R)} \, \mb{E}[w( \scr{M}(v_0)) \, | \, (Y_{v})_{v \in V}, (\sta(e))_{e \in \til{E}}] \cdot  \bm{1}_{[ R_{v_0} = R]}.
	\]
	On the other hand, the only randomness which remains in 
	the conditional expectation involving $w(\scr{M}(v_0))$ is 
	over the states of the edges adjacent to $v_0$. Observe now
	that since \textsc{DP-OPT} behaves optimally
	on $G[ \{v_{0}\} \cup R_{v_{0}}]$,
	we get that
	\begin{equation} \label{eqn:DP_optimality}
		\mb{E}[w( \scr{M}(v_0)) \, | \, (Y_{v})_{v \in V}, (\sta(e))_{e \in \til{E}}] = \OPT(v_0,R_{v_0}),
	\end{equation}
	and so for the \textit{fixed} subset $R \subseteq U$,
	$
	\mb{E}[w( \scr{M}(v_0)) \, | \, (Y_{v})_{v \in V}, (\sta(e))_{e \in \til{E}}] \cdot \bm{1}_{[ R_{v_0} = R]} = \OPT(v_0,R) \cdot\bm{1}_{[ R_{v_0} = R]},
	$
	after multiplying each side of \eqref{eqn:DP_optimality} by the indicator random variable $\bm{1}_{[R_{v_{0}}=R]}$.
	Thus, 
	\[
	\mb{E}[\phi_{v_0,R}\, | \, (Y_{v})_{v \in V}, (\sta(e))_{e \in \til{E}}] = \frac{g(Y_{v_0})}{F} \, \bm{1}_{[ R_{v_0} = R]},
	\]
	after cancellation. We therefore get that
	$
	\sum_{\substack{R \subseteq U: \\ u_{0} \in R}} \mb{E}[ \phi_{v_{0},R} \, | \, (Y_{v})_{v \in V}, (\sta(e))_{e \in \til{E}}] 
	= \frac{g(Y_{v_0})}{F} \sum_{\substack{R \subseteq U: \\ u_{0} \in R}} \bm{1}_{[ R_{v_0} = R]}.
	$
	Let us now focus on the case when $v_0$ arrives before the critical time; that is, $0 \le Y_{v_0} < \til{Y}_c$. Up until the arrival of $v_0$, the executions of the algorithm on $\til{G}$ and $G$ proceed identically, thanks to the coupling between the executions.
	As such, $u_0$ must be available when $v_0$ arrives. 
	We interpret this observation in the above notation as saying the following:
	$
	\bm{1}_{[ Y_{v_0} < \til{Y}_c]} \le \sum_{\substack{R \subseteq U: \\ u_{0} \in R}} \bm{1}_{[ R_{v_0} = R]}.
	$
	As a result, $\sum_{\substack{R \subseteq U: \\ u_{0} \in R}} \mb{E}[ \phi_{v_{0},R} \, | \, (Y_{v})_{v \in V}, (\sta(e))_{e \in \til{E}}]
	\ge \frac{g(Y_{v_0})}{F} \, \bm{1}_{[ Y_{v_0} < \til{Y}_c]}$.
	Now, if we take expectation over $Y_{v_0}$, while still conditioning on the random variables $(Y_{v})_{v \in \til{V}}$, then we get
	that
	\[
	\mb{E}[ g(Y_{v_0}) \cdot \bm{1}_{[ Y_{v_0} < \til{Y}_c]} \, | \, (Y_{v})_{v \in \til{V}}, (\sta(e))_{e \in \til{E}}]= \int_{0}^{\til{Y}_c} g(z) \, dz,
	\]
	as $Y_{v_0}$ is drawn uniformly from $[0,1]$, independently from $(Y_{v})_{v \in \til{V}}$ and $(\sta(e))_{e \in \til{E}}$.
	Thus, after applying the law of iterated expectations,
	\[
	\sum_{\substack{R \subseteq U: \\ u_{0} \in R}} \mb{E}[ \phi_{v_{0},R} \, | \, (Y_{v})_{v \in \til{V}}, (\sta(e))_{e \in \til{E}}] \ge  \frac{1}{F} \int_{0}^{\til{Y}_c} g(z) \, dz,
	\]
	and so the claim holds.
\end{proof}

With Propositions \ref{prop:offline_dual_variable_lower_bound_general} and \ref{prop:online_dual_variable_lower_bound_general},
the proof of Lemma \ref{lem:dual_feasibility_constraints_general} follows easily.
\begin{proof}[Proof of Lemma \ref{lem:dual_feasibility_constraints_general}]
	We first observe that by taking the appropriate conditional expectation, 
	\Cref{prop:offline_dual_variable_lower_bound_general} ensures that
	$
	\mb{E}[\alpha_{u_0} \, | \, (Y_{v})_{v \in \til{V}}, (\sta(e))_{e \in \til{E}}] \ge \left(1- \frac{1}{e}\right)^{-1}w_{u_0} \cdot (1 - g(\til{Y}_c)),
	$
	where the right-hand side follows since $\til{Y}_c$ is entirely determined from $(Y_{v})_{v \in \til{V}}$ and $(\sta(e))_{e \in \til{E}}$. Thus, combined with Proposition \ref{prop:online_dual_variable_lower_bound_general},
	\begin{equation}\label{eqn:variables_conditional_expectation}
		\mb{E}[ p_{u_0,v_0} \cdot \alpha_{u_0}+   w_{u_0} \cdot p_{u_0,v_0} \cdot \sum_{\substack{R \subseteq U: \\ u_0 \in R}} \phi_{v,R} \, | \, (Y_{v})_{v \in \til{V}}, (\sta(e))_{e \in \til{E}}],
	\end{equation}
	is lower bounded by
	\begin{equation} \label{eqn:integral_lower_bound}
		\left(1- \frac{1}{e}\right)^{-1} w_{u_0} \cdot p_{u_0,v_0} \cdot ( 1 - g( \til{Y}_c)) + \frac{w_{u_0} \, p_{u_0,v_0}}{F} \int_{0}^{\til{Y}_c} g(z) \, dz .
	\end{equation}
	However, $g(z):= \exp(z-1)$ for $z \in [0,1]$ by assumption, so
	$
	( 1 - g( \til{Y}_c)) + \int_{0}^{\til{Y}_c} g(z) \, dz = \left(1 - \frac{1}{e} \right),
	$
	no matter the value of the critical time $\til{Y}_c$. Thus,
	\begin{equation} \label{eqn:integral_inequality}
		\left(1- \frac{1}{e}\right)^{-1} \left(( 1 - g( \til{Y}_c)) + \frac{1-1/e}{F} \int_{0}^{\til{Y}_c} g(z) \, dz \right)
		\ge  1,
	\end{equation}
	as $F \le 1-1/e$ by definition (see \eqref{eqn:target_competitive_ratio}).
	If we now lower bound \eqref{eqn:integral_lower_bound} using \eqref{eqn:integral_inequality}
	and take expectations over \eqref{eqn:variables_conditional_expectation}, it follows that
 $
		\mb{E}[p_{u_0,v_0} \cdot \alpha_{u_0} + w_{u_0} \cdot p_{u_0,v_0} \cdot \sum_{\substack{R \subseteq U: \\ u_0 \in R}} \phi_{v,R}] \ge  w_{u_0} \cdot p_{u_0,v_0}.
	$
	As the vertices $u_{0} \in U$ and $v_{0} \in V$ were chosen arbitrarily, the
	proposed dual solution of Lemma \ref{lem:dual_feasibility_constraints_general} is feasible, 
	and so the proof is complete.
\end{proof}

\section{Efficiency of Our Algorithms} \label{sec:algorithm_efficiency}
In this section, we prove that all of our algorithms can be implemented efficiently in
the membership oracle model. Given a stochastic graph $G$, we denote $|G|$ to be the number
of bits needed to encode all of its parameters \textit{excluding} its downward-closed probing
constraints $(\scr{C}_v)_{v \in V}$.
Clearly, \Cref{alg:dynamical_program} has a runtime which is polynomial in $|G|$ by \Cref{thm:efficient_star_dp} (denoted
$\poly(|G|)$).
For our remaining algorithms, we first show that \ref{LP:new} can be solved efficiently.
\begin{theorem} \label{thm:LP_solvability}
	Suppose that $G=(U,V,E)$ in a stochastic graph with downward-closed probing constraints $(\scr{C}_v)_{v \in V}$.
	In the membership oracle model, \ref{LP:new} can be solved in time $\poly(|G|)$.
\end{theorem}
We prove Theorem \ref{thm:LP_solvability} by first considering
the dual of \ref{LP:new}. Note, that in the below LP formulation,
if $\bm{e}=(e_{1}, \ldots , e_{k}) \in \scr{C}_v$,
then we set $e_{i}=(u_{i},v)$ for $i=1, \ldots ,k$ for convenience.
\begin{align*}\label{LP:config_dual}
	\tag{LP-config-dual}
	&\text{minimize} &  \sum_{u \in U} \alpha_{u} + \sum_{v \in V} \beta_{v}  \\
	&\text{subject to} & \beta_{v} + \sum_{j=1}^{|\bm{e}|} p_{e_j} \cdot q(\bm{e}_{< j}) \cdot \alpha_{u_j} \ge \sum_{j=1}^{|\bm{e}|} p_{e_j} \cdot w_{e_j} \cdot q( \bm{e}_{ < j}) &&
	\forall v \in V, \bm{e} \in \scr{C}_v \\ 
	&&  \alpha_{u} \ge 0 && \forall u \in U\\
	&& \beta_{v} \in \mb{R} && \forall v \in V
\end{align*}

Observe that to prove Theorem \ref{thm:LP_solvability},
it suffices to show that \ref{LP:config_dual} has a (deterministic) polynomial time separation
oracle, as a consequence of how the ellipsoid algorithm \cite{Groetschel,GartnerM} executes (see \cite{Williamson,Vondrak_2011,Adamczyk2017,Lee2018} for more detail). 

Suppose that we are presented a particular selection of dual variables,   
say $(\alpha_{u})_{u \in U}$ and $(\beta_{v})_{v \in V}$, which may or may not
be a feasible solution to \ref{LP:config_dual}. Our
separation oracle must determine efficiently whether these variables
satisfy all the constraints of \ref{LP:config_dual}. In the case
in which the solution is \textit{infeasible}, the oracle must additionally
return a constraint which is violated. It is clear that we can accomplish this for the non-negativity constraints,
so let us fix a particular $v \in V$ in what follows. 
We wish to determine whether there exists some
$\bm{e}=(e_{1}, \ldots ,e_{k}) \in \scr{C}_v$,
such that if $e_{i}=(u_{i},v)$ for $i=1, \ldots ,k$, then
\begin{equation}\label{eqn:existence_of_string}
	\sum_{j=1}^{|\bm{e}|}(w_{e_j} - \alpha_{u_j}) \cdot p_{e_j}  \cdot q(\bm{e}_{< j}) > \beta_{v},
\end{equation}
where the left-hand side of \eqref{eqn:existence_of_string} is $0$ if $\bm{e}=\lambda$.
In order to make this determination, it suffices to solve the following maximization problem.
Given any selection of real values, $(\alpha_{u})_{u \in U}$,
\begin{align}
	\text{maximize} \quad &\sum_{i=1}^{|\bm{e}|} (w_{e_i} - \alpha_{u_i}) \cdot p_{e_i} \cdot \prod_{j=1}^{i-1} (1- p_{e_j}) \label{eqn:demand_oracle}\\
	\text{subject to} \quad & \bm{e} \in \scr{C}_v  \nonumber
\end{align}
Before we show how \eqref{eqn:demand_oracle} can be solved, we provide a buyer/seller interpretation of the optimization
problem. Assuming first that the edges exist with certainty
(i.e. $p_e \in \{0,1\}$ for all $e \in \partial(v)$), let us suppose a seller is trying to allocate the items of $U$
to a number of buyers. We view the vertex $v$ as a \textit{buyer} who wishes to purchase a subset
of items $S \subseteq U$, based on their valuation function $f(S)$. Assume that
$v$ has \textbf{unit demand}, that is $f(S):= \max_{s \in S} p_{s,v} w_{s,v}$. The values $(\alpha_{s})_{s \in U}$ are viewed as prices the buyer must pay, and $\max_{S \subseteq U} (f(S) - \sum_{s \in S} \alpha_{s})$ is the maximum utility of $v$. Clearly, for the simple case of a unit-demand buyer,
an optimum assignment is the item  $u \in U$ for which $p_{u,v} w_{u,v} - \alpha_{u}$ is
maximized. 

Returning the setting of arbitrary edge probabilities, even the case of a unit-demand buyer is a non-trivial optimization problem in the stochastic probing framework. Observe that we may view the edge probabilities $(p_{e})_{e \in \partial(v)}$ as modelling the setting when there is uncertainty in whether or not the purchase proposals will succeed; that is, $\sta(u,v)=1$, provided the seller agrees to sell item $u$ to buyer $v$. In this interpretation, \eqref{eqn:demand_oracle} is the expected utility of the unit-demand buyer $v$ which purchases the first item $u \in U$ such that $\sta(u,v)=1$, at which point $v$ gains utility $w_{u,v} - \alpha_u$. This easily reduces to the problem solved by $\textsc{DP-OPT}$,
and we include the details below:
\begin{proposition}\label{prop:demand_to_membership_reduction}
	If $\scr{C}_v$ is downward-closed, then for any selection of values $(\alpha_u)_{u \in U}$,
	\eqref{eqn:demand_oracle} can be solved efficiently by $\textsc{DP-OPT}$, assuming access to a membership query
	oracle for $\scr{C}_v$.
\end{proposition}
\begin{proof}
	Compute $\til{w}_{e}:= w_{e} - \alpha_u$ for each $e=(u,v) \in \partial(v)$,
	and define $P:=\{e \in \partial(v) : \til{w}_e \ge 0\}$.
	First observe that if $P = \emptyset$, then \eqref{eqn:existence_of_string}
	is maximized by the empty-string $\lambda$. Thus, for now on assume that $P \neq \emptyset$. Since $\scr{C}_v$ is downward-closed, it suffices to consider those $\bm{e} \in \scr{C}_v$ whose
	edges all lie in $P$. As such, for notational convenience, let us hereby assume
	that $\partial(v)=P$. Observe then that solving \eqref{eqn:demand_oracle} corresponds to executing \textsc{DP-OPT}
	on the stochastic graph $G[U \cup \{v\}]$, with edge weights replaced
	by $(\til{w}_{e})_{e \in \partial(v)}$.
\end{proof}

As a corollary of \Cref{thm:LP_solvability}, it immediately follows that \Cref{alg:ROM_edge_weights} is poly-time
in $|G|$. Let us now consider when we have a known i.d. instance $(H_{\typ}, (\scr{D}_i)_{i=1}^{n})$,
and $|\scr{D}_{i}|$ denotes the amount of space needed to encode the distribution $\scr{D}_i$.
Following the standard in the literature, we consider an algorithm to be efficient if it can be
implemented in time $\poly(|H_{\typ}|, (|\scr{D}_i|)_{i=1}^{n})$.

We first observe that \ref{LP:new_id} can be solved in time $\poly(|H_{\typ}|, (|\scr{D}_i|)_{i=1}^{n})$
using the same strategy as in \Cref{thm:LP_solvability}, as the same maximization problem \eqref{eqn:demand_oracle} is needed
to separate the dual of \ref{LP:new_id}. On the other hand, \Cref{thm:known_id_reduction} is a polynomial-time reduction.
Thus,  since the OCRS and RCRS used in \Cref{cor:known_id_aom} and \Cref{cor:known_id_rom}
are polynomial time, our $1/2$ and $1-1/e$ competitive ratios are attained by efficient online probing algorithms.

\section{A Tight Adaptivity Gap} \label{sec:adaptivity_gap}

\label{sec:non-adaptive_negative}
Let $G=(U,V,E)$ be a stochastic graph with \textbf{substring-closed} probing constraints $(\scr{C}_v)_{v \in V}$.
Here $\scr{C}_v$ is substring-closed if any substring of $\bm{e} \in \scr{C}_v$ is also in $\scr{C}_v$. This is a less restrictive definition than imposing $\scr{C}_v$ must be downward-closed, and is the minimal assumption one needs to ensure that the offline stochastic matching is well-defined.

Given an offline probing algorithm, we say it is \textbf{non-adaptive} provided its probes are a randomized function of $G$.
Similar to the definition of the offline adaptive benchmark,
we define the \textbf{non-adaptive benchmark} as the optimum performance of a non-adaptive offline probing algorithm on $G$. That is, $\OPT_{\text{n-adap}}(G):= \sup_{\scr{B}} \mb{E}[ w(\scr{B}(G))]$,
where the supremum is over all non-adaptive offline probing algorithms.
We define the \textbf{adaptivity gap} of the \textbf{stochastic matching problem with
	one-sided probing constraints},
as the ratio
\begin{equation} \label{eqn:adaptivity_gap}
	\inf_{G} \frac{\OPT_{\text{n-adap}}(G)}{\OPT(G)},
\end{equation}
where the infimum is over all (bipartite) stochastic graphs $G=(U,V,E)$
with substring-closed probing constraints $(\scr{C}_v)_{v \in V}$. Observe
that \eqref{eqn:adaptivity_gap} is upper bounded by $1$ by definition. We
now state a better upper bound (i.e., negative/impossibility result) on \eqref{eqn:adaptivity_gap}.

\begin{theorem}\label{thm:adaptivity_gap_negative_restatement}
	The adaptivity gap of the stochastic matching problem with one-sided probing constraints 
	is upper bounded by $1-1/e$.
\end{theorem}

Theorem \ref{thm:adaptivity_gap_negative_restatement} follows by considering
a sequence of stochastic graphs. In particular, given $n \ge 1$, consider functions $p=p(n)$ and $s = s(n)$ which satisfy the following:
\begin{enumerate}
	\item $p \ll 1/ \sqrt{n}$ and $s \rightarrow \infty$ as $n \rightarrow \infty$.
	\item $s \le  p n$ and $s = (1 - o(1)) pn$.
\end{enumerate}
Consider now an unweighted stochastic graph $G_{n}=(U,V,E)$ with unit patience
values, and which satisfies $|U| = s$ and $|V| = n$. Moreover, assume that $p_{u,v} = p$ for all $u \in U$ and $v \in V$.
Observe that $G_n$ corresponds to the bipartite Erdős–Rényi random graph $\mb{G}(s,n,p)$. 
\begin{lemma}\label{lem:committal_benchmark_hardness}
	
	The offline adaptive benchmark returns a matching of size asymptotically equal to $s$
	when executing on $G_n$; that is, $\OPT(G_n) = (1 + o(1)) s$.
	
\end{lemma}

We omit the proof of Lemma \ref{lem:committal_benchmark_hardness}, as it is routine analysis
of the Erdős–Rényi random graph $\mb{G}(s,n,p)$. Instead, we focus on proving the following
lemma, which together with Lemma \ref{lem:committal_benchmark_hardness}
implies the upper bound of Theorem \ref{thm:adaptivity_gap_negative_restatement}:

\begin{lemma} \label{lem:non_adaptive_hardness}
	The non-adaptive benchmark returns in expectation a matching of size at most
	$(1+o(1)) \left(1 - \frac{1}{e} \right) s$ when executing on
	$G_n$. That is,
	\[
	\nOPT(G) \le (1 + o(1)) \left(1 - \frac{1}{e} \right) s.
	\] 
\end{lemma}
\begin{proof}
	Let $\scr{A}$ be a non-adaptive probing algorithm, which we may assume is deterministic
	w.l.o.g. As the probes of $\scr{A}$ are
	determined independently of the random variables $(\sta(e))_{e \in E}$,
	we can define $x_{e} \in \{0,1\}$ for each $e \in E$ to indicate whether or not
	$\scr{A}$ probes the edge $e$.
	
	Now, if $\scr{A}(G)$ is the matching returned by $\scr{A}$,
	then using the independence of the edge states $(\sta(e))_{e \in E}$, we get that
	\begin{align}
		\mb{P}[\text{$u$ matched by $\scr{A}(G)$}] & \le \mb{P}\left[ \cup_{\substack{v \in V: \\ x_{u,v} =1}} \sta(u,v)=1 \right] \\
		&= 1 - \prod_{v \in V} (1 - p x_{u,v})
	\end{align}
	and so $\mb{E}[|\scr{A}(G))|] \le s - \sum_{u \in U} \prod_{v \in V} (1 - p x_{u,v})$. As such, if we can show that  
	\[
	\sum_{u \in U} \prod_{v \in V} (1 - p x_{u,v}) \ge (1 - o(1))\frac{s}{e},
	\]
	then this will imply that $\mb{E}[|\scr{A}(G)|] \le (1 + o(1))  \left(1 - \frac{1}{e} \right) s$.
	
	To see this, first observe that since $p(n) \rightarrow 0$ as $n \rightarrow \infty$, we know
	that $1 - p x_{u,v} = (1 + o(1)) \exp(-p x_{u,v})$
	for each $v \in V$. In fact, since $p x_{u,v} \le p$ for all $v \in V$,
	the asymptotics are uniform across $V$. More precisely, there exists $C > 0$,
	such that for $n$ sufficiently large,
	$
	1 - p x_{u,v} \ge (1 - C p^2) \exp(-p x_{u,v})
	$
	for all $v \in V$. As a result, 
	\begin{align*}
		\prod_{v \in V} (1 - p x_{u,v}) &\ge (1 - Cp^{2})^{n} \exp\left(-\sum_{v \in V} p x_{u,v}\right)	\\
		&= (1 + o(1))\exp\left(-\sum_{v \in V} p x_{u,v}\right),
	\end{align*}
	where the second line follows since $p \ll 1/\sqrt{n}$ by assumption. On the other hand, Jensen's inequality
	ensures that
	\[
	\sum_{u \in U}  \frac{\exp\left(-\sum_{v \in V} p x_{u,v}\right)}{s} \ge \exp\left(- \frac{\sum_{u \in U, v \in V} p x_{u,v}}{n}\right).
	\]
	However, $\sum_{u \in U} x_{u,v} \le 1$ for each $v \in V$.
	Thus, $\sum_{u \in U, v \in V} p x_{u,v} \le p n$, and so
	\[
	\exp\left(- \frac{\sum_{u \in U, v \in V} p x_{u,v}}{s}\right) \ge \exp\left(- \frac{p n}{s}\right) \ge \frac{1}{e},
	\]
	where the last line follows since $p n \le s$. It follows that
	$
	\sum_{u \in U} \prod_{v \in V} (1 - p x_{u,v}) \ge (1 +o(1)) \frac{s}{e},
	$
	and so $\mb{E}[|\scr{A}(G)|] \le (1 + o(1))  \left(1 - \frac{1}{e} \right) s.$
	As the asymptotics hold uniformly across each deterministic non-adaptive algorithm $\scr{A}$, 
	this completes the proof.	
\end{proof}
When considered in the known stochastic graph setting,
the $1-1/e$-competitive algorithm of \Cref{cor:known_id_rom} is non-adaptive.
Moreover, it applies more generally to stochastic
graphs with substring-closed probing constraints. The stronger downward closed condition is only needed to solve \ref{LP:new}
efficiently. 
Thus, Corollary \ref{cor:known_id_rom} and Theorem \ref{thm:adaptivity_gap_negative_restatement}
exactly characterize the adaptivity gap of the offline stochastic matching problem with one-sided probing constraints:

\begin{corollary}\label{cor:adaptivity_gap}
	The adaptivity gap of the offline stochastic matching problem with one-sided probing constraints is $1-1/e$.
\end{corollary}

\section{Conclusion and Open Problems}

There are some basic questions that are unresolved. Perhaps the most basic question 
which is also unresolved in the classical setting without probing is 
to bridge the gap between the positive $1-1/e$ competitive ratio and in-approximations in the context of known i.d. random order arrivals. In terms of the single item prophet secretary problem (without probing), Correa et al. \cite {CorreaSZ19} obtain a $0.669$ competitive ratio following Azar et al. \cite{AzarCK18} who were the first to surpass the  $1-1/e$ ``barrier''.     
Correa et al. \cite{CorreaSZ19} also establish a $0.732$ in-approximation for the i.d. setting, and Huang et al. \cite{huang2022} recently established a $0.703$ in-approximation for i.i.d. arrivals in the multi-item case. 
Can we surpass $1-1/e$ in the probing setting for i.d. input arrivals or for the special case of i.i.d. input arrivals? As previously
mentioned, Yan \cite{Yan2022} recently proved that $0.645 > 1-1/e$ is attainable for known i.i.d. arrivals when probing is not required. 
Is there a provable difference between  stochastic bipartite matching (with probing constraints) and the 
classical online settings? Can we obtain the same competitive results against an optimal offline {\it  non-committal}  benchmark which respects the probing constraints  but doesn't operate in the probe-commit model? The $0.51$ in-approximation result of Fata et al. \cite{Fata2019MultiStageAM} suggests that $0.51$ may be the optimal competitive ratio against this stronger benchmark.

One interesting extension of the probing model is to 
allow non-Bernoulli edge random variables 
to describe edge uncertainty. Even for a single online
vertex in the unconstrained setting, this problem is interesting as it corresponds to computing an optimal
policy for the \textbf{free-order prophets problem},
which was recently studied by Segev and Singla in \cite{Segev2020}.

\subsection*{Acknowledgements} We would like to thank Denis Pankratov, Rajan Udwani, and David Wajc for their constructive comments on early versions of this paper.

\bibliographystyle{plain}
\bibliography{ref}

\appendix

\section{LP Relations} \label{sec:LP_relations}

Suppose that we are given an arbitrary stochastic
graph $G=(U,V,E)$. In this section, we first prove the equivalence between 
the relaxed stochastic matching problem and \ref{LP:new}.
We then state \ref{LP:standard_definition_general}, the standard
LP in the stochastic matching literature, as introduced by Bansal et al. \cite{BansalGLMNR12},
as well as \ref{LP:full_patience}, the LP introduced by Gamlath et al. \cite{Gamlath2019}.
We then show that \ref{LP:full_patience} and \ref{LP:new} have the same optimum value when $G$ has unbounded patience. 
\begin{theorem}\label{thm:LP_relaxation_benchmark_equivalence}
	$\OPT(G) = \LPOPT(G)$
\end{theorem}
\begin{proof}
	Clearly, Theorem \ref{thm:new_LP_relaxation} accounts for one side of the inequality,
	so it suffices to show that $\LPOPT(G) \le \rOPT(G)$. Suppose we are presented a feasible solution $(x_{v}(\bm{e}))_{v \in V, \bm{e} \in \scr{C}_v}$ to \ref{LP:new}. Consider then the following algorithm:
	\begin{enumerate}
		\item $\scr{M} \leftarrow \emptyset$.
		\item For each $v \in V$, set $e \leftarrow \mathtt{VertexProbe}(v, \partial(v), (x_{v}(\bm{e}))_{\bm{e} \in \scr{C}_v})$.
		If $e \neq \emptyset$, then add $e$ to $\scr{M}$.
		\item Return $\scr{M}$.
	\end{enumerate}
	Using Lemma \ref{lem:fixed_vertex_probe}, it is clear that
	$\mb{E}[ w(\scr{M})] = \sum_{v \in V} \sum_{\bm{e} \in \scr{C}_v} \val(\bm{e}) \cdot x_{v}(\bm{e}).$
	Moreover, each vertex $u \in U$ is matched by $\scr{M}$ at most once in expectation, as a consequence of
	constraint \eqref{eqn:relaxation_efficiency_matching} of \ref{LP:new}, and so the algorithm satisfies
	the required properties of a relaxed probing algorithm.
	The proof is therefore complete.
\end{proof}
Consider \ref{LP:standard_definition_general}, which is defined
only when $G$ has patience values $(\ell_v)_{v \in V}$. Here
each $e \in E$ has a variable $x_e$ corresponding to the probability
that the offline adaptive benchmark probes $e$. 
\begin{align}\label{LP:standard_definition_general}
	\tag{LP-std}
	&\text{maximize} & \sum_{e \in E} w_{e} \cdot p_{e} \cdot x_{e} \\
	&\text{subject to} & \sum_{e \partial(u)} p_{e} \cdot x_{e} & \leq 1 && \forall u \in U \\
	& &\sum_{e \in \partial(v)} p_{e} \cdot x_{e} & \leq 1 && \forall v \in V  \\
	& &\sum_{e \in \partial(v)} x_{e} & \leq \ell_v && \forall v \in V  \\
	& &0 \leq x_{e} &\leq 1 && \forall e \in E.
\end{align}
Observe that \ref{LP:new} and \ref{LP:standard_definition_general} are the same
LP in the case of unit patience:
\begin{align} \label{LP:standard_benchmark}
	\tag{LP-std-unit}
	&\text{maximize}&  \sum_{v \in V} \sum_{e \in \partial(v)} w_e \cdot p_{e} \cdot x_{e}\\
	&\text{subject to}& \sum_{e \in \partial(u)} p_{e} \cdot x_{e} &\leq 1 && \forall u \in U\\
	&&\sum_{e \in \partial(v)} x_{e} &\leq 1 && \forall v \in V\\
	&&  x_{e}  &\ge 0 && \forall e \in E
\end{align}

Gamlath et al. modified \ref{LP:standard_definition_general} in for the case of unbounded patience by adding in exponentially
many extra constraints. Specifically, for each $v \in V$ and $S \subseteq \partial(v)$, they ensure that
\begin{equation}\label{eqn:svensson_constraints}
	\sum_{e \in S} p_{e} \cdot x_{e} \le 1 - \prod_{e \in S}(1 - p_{e}),
\end{equation}
In the same variable interpretation as \ref{LP:standard_definition_general},
the left-hand side of \eqref{eqn:svensson_constraints} corresponds
to the probability the adaptive benchmark matches an edge of $S \subseteq \partial(v)$, and the right-hand side corresponds to the probability an edge of $S$ is active\footnote{The LP considered by Gamlath et al. in \cite{Gamlath2019}
	also places the analogous constraints of \eqref{eqn:svensson_constraints} on the vertices of $U$. That being said, these additional constraints are not used anywhere in the work of Gamlath et al., so we omit them.}. 
\begin{align}\label{LP:full_patience}
	\tag{LP-QC}
	&\text{maximize} & \sum_{e \in E} w_{e} \cdot p_{e} \cdot x_{e} \\
	&\text{subject to} & \sum_{e \in S} p_{e} \cdot x_{e} &\le 1 - \prod_{e \in S}(1 - p_{e}) && \forall v \in V, S \subseteq \partial(v) \notag \\
	& &\sum_{e \in \partial(u)} p_{e} \cdot x_{e} & \leq  1&& \forall u \in U  \label{eqn:relxation_matched_full_patience}\\
	& &x_{e} &\ge 0 && \forall e \in E. \notag
\end{align}
Let us denote $\qLPOPT(G)$ as the optimum value of \ref{LP:full_patience}.
\begin{proposition} \label{prop:full_patience_equivalence}
	If $G$ is unconstrained, then $\qLPOPT(G) = \LPOPT(G)$.
\end{proposition}
In order to prove Proposition \ref{prop:full_patience_equivalence}, we make use
of a result of Gamlath et al.  We mention that an almost identical result
is also proven by Costello et al. \cite{costello2012matching} using different techniques.
\begin{theorem}[\cite{Gamlath2019}] \label{thm:costello_svensson_guarantee}
	Suppose that $G=(U,V,E)$ is an unconstrained stochastic graph, and 
	$(x_e)_{e \in E}$ is a solution to \ref{LP:full_patience}. 
	For each $v \in V$, there exists an online probing algorithm $\scr{B}_{v}$
	whose input is $(v,\partial(v), (x_e)_{e \in \partial(v)})$, and which satisfies
	$\mb{P}[\text{$\scr{B}_{v}$ matches $v$ to $e$}] = p_{e} x_e$ for each $e \in \partial(v)$.
\end{theorem}

\begin{proof}[Proof of Proposition \ref{prop:full_patience_equivalence}]
	
	Observe that by Theorem \ref{thm:LP_relaxation_benchmark_equivalence},  in order to prove the claim it suffices
	to show that $\qLPOPT(G) = \rOPT(G)$. Clearly, $\rOPT(G) \le \qLPOPT(G)$,
	as can be seen by defining $x_{e}$ as the probability that the relaxed benchmark probes
	the edge $e \in E$. Thus, we focus on showing that $\qLPOPT(G) \le \rOPT(G)$.
	
	Suppose that $(x_e)_{e \in E}$ is an optimum solution to $\qLPOPT(G)$. 
	We design the following algorithm, which we denote by $\scr{B}$:
	\begin{enumerate}
		\item $\scr{M} \leftarrow \emptyset$.
		\item For each $v \in V$, execute  $\scr{B}_{v}$ on $(v,\partial(v), (x_e)_{e \in \partial(v)})$,
		where $\scr{B}_v$ is the online probing algorithm of Theorem \ref{thm:costello_svensson_guarantee}.
		If $\scr{B}_v$ matches $v$, then let $e'$ be this edge, and add $e'$ to $\scr{M}$
		\item Return $\scr{M}$.
	\end{enumerate}
	Using Theorem \ref{thm:costello_svensson_guarantee}, it is clear that
	$\mb{E}[ w(\scr{M})] = \sum_{e \in E} w_e p_e x_e.$
	Moreover, each vertex $u \in U$ is matched by $\scr{M}$ at most once in expectation, as a consequence of
	constraint \eqref{eqn:relxation_matched_full_patience}. As a result, $\scr{B}$ is a relaxed probing algorithm.
	Thus, $\qLPOPT(G) =  \sum_{e \in E} w_e p_e x_e \le \rOPT(G)$, and so the proof is complete.
\end{proof}

\section{Deferred Proofs} \label{sec:deferred_proofs}

\begin{proof}[Proof of Lemma \ref{lem:edge_value_lower_bound}]
Set $\alpha:=1/e$ for clarity, and take $t \ge \lceil \alpha n \rceil$. 
Define $e_{t}:=(u_t,v_{t})$, where $u_{t}$ is the vertex of $U$ which $v_{t}$ proposes to
(which is the empty set $\emptyset$, if no such proposal occurs). 
For each $u \in U$, define $C(u,v_t)$ to be the event in which $v_t$ proposes to $u$. 
Let us now condition on the random subset $V_{t}$, as well as the 
random vertex $v_{t}$. 
In this case, 
\[
    \mb{E}[ w(e_{t})  \, | \, V_{t}, v_t] = \sum_{u \in U} w_{u,v_t} \, \mb{P}[ C(u,v_t)  \mid V_{t}, v_{t}].
\]
Observe however that once we condition on $V_{t}$ and $v_{t}$, Algorithm \ref{alg:ROM_edge_weights} corresponds to executing $\mathtt{VertexProbe}$ on the instance $(v_t, \partial(v_t), (x_{v_t}(\bm{e}))_{\bm{e} \in \scr{C}_v})$,
where we recall that $(x_{v}(\bm{e}))_{\bm{e} \in \scr{C}_v, v \in v_t}$ is an optimum solution
to \ref{LP:new} for $G_t=G[U \cup V_t]$. Thus,
$\mb{P}[ C(u,v_t)  \mid V_{t}, v_{t}]  =  p_{u,v_t} \til{x}_{u,v_t}$ by \Cref{lem:fixed_vertex_probe},
where
\[
\til{x}_{u,v_t}:= \sum_{\substack{\bm{e}' \in \scr{C}_{v_t}: \\ e \in \bm{e}'}} q(\bm{e}_{ < e}') \cdot x_{v_t}( \bm{e}'),
\]
and so $\mb{E}[ w(e_{t})  \, | \, V_{t}, v_t] = \sum_{u \in U} w_{u,v_t}  p_{u,v_t}  \til{x}_{u,v_t}$.
On the other hand, if we condition on \textit{solely} $V_{t}$, then $v_{t}$ remains distributed uniformly
at random amongst the vertices of $V_{t}$. Moreover, once we condition on $V_t$, the graph $G_t$ is determined, and thus so are
the values $(x_{v}(\bm{e}))_{v \in V_{t}, \bm{e} \in \scr{C}_v}$. These observations together imply that
\begin{equation}\label{eqn:conditional_expectation_value}
\mb{E}[ w_{u,v_t} \, p_{u,v_t} \,  \til{x}_{u,v_t}  \, | \, V_{t}] = \frac{\sum_{v \in V_t} w_{u,v} \, p_{u,v} \, \til{x}_{u,v}}{t}
\end{equation}
for each $u \in U$ and $\lceil \alpha n \rceil \le t \le n$. If we now take expectation over $v_{t}$, then using the law of iterated
expectations,
\begin{align*}
    \mb{E}[ w(e_t)  \, | \, V_{t}] &= \mb{E}[ \,  \mb{E}[ w(e_t) \, | \, V_t, v_t] \,   \, | \, V_{t} ]  \\
                              &= \mb{E}\left[ \sum_{u \in U} w_{u,v_t} \, p_{u,v_t} \, \til{x}_{u,v_t}   \, | \, V_{t} \right] \\
                              &= \sum_{u \in U} \mb{E}[ w_{u,v_t} \, p_{u,v_t} \, \til{x}_{u,v_t}  \, | \, V_{t}] \\
                              &= \sum_{u \in U} \sum_{v \in V_t} \frac{w_{u,v} p_{u,v} \, \til{x}_{u,v}}{t},
\end{align*}
where the final equation follows from \eqref{eqn:conditional_expectation_value}.
Now,
$
    \LPOPT(G_t)=\sum_{v \in V_t} \sum_{u \in U} w_{u,v_t} \, p_{u,v_t} \, \til{x}_{u,v_t},
$
as $(x_{v}(\bm{e}))_{v \in V_{t}, \bm{e} \in \scr{C}_v}$ is an optimum solution to
\ref{LP:new} for $G_t$. As a result,
$
    \mb{E}[ w(e_t)  \, | \, V_{t}]  = \LPOPT(G_t)/t,
$
and so $\mb{E}[ w(e_t)]  = \mb{E}[\LPOPT(G_t)]/t$,
after taking taking expectation over $V_{t}$. On the other hand, Lemma \ref{lem:random_induced_subgraph} implies that
\[
	\frac{\mb{E}[\LPOPT(G_t)]}{t} \ge  \frac{\LPOPT(G)}{n}.
\]
Thus, $\mb{E}[ w(e_t)] \ge \LPOPT(G)/n$,
provided $\lceil \alpha n \rceil \le t \le n$, thereby completing the proof.
\end{proof}

\begin{proof}[Proof of Lemma \ref{lem:availability_lower_bound}]
Let us assume that $\lceil \alpha n  \rceil \le t \le n$ is fixed, and $(x_{v}^{(t)}(\bm{e}))_{v \in V, \bm{e} \in \scr{C}_v}$
is the optimum solution of \ref{LP:new} for $G_t$, as used by Algorithm \ref{alg:ROM_edge_weights}.
For each $u \in U$ and $v \in v$, define the \textbf{edge variable} $\til{x}^{(t)}_{u,v}$,
where
\[
	\til{x}^{(t)}_{u,v} := \sum_{\substack{\bm{e}' \in \scr{C}_{v_t}: \\ e \in \bm{e}'}} q(\bm{e}_{ < e}') \cdot x_{v_t}^{(t)}( \bm{e}')
\]
We wish to prove that for each $u \in U$, 
\begin{equation} \label{eqn:conditional_available}
\mb{P}[u \in R_{t}  \, | \, V_{t},v_t] \ge \lfloor \alpha n \rfloor /(t-1).
\end{equation}
As such, we must condition on $(V_t,v_t)$ throughout the remainder of the proof.
To simplify the argument, we abuse notation slightly and remove $(V_t,v_t)$
from the subsequent probability computations, though it is understood to implicitly
appear.

Given arriving node $v_{j}$ for $j=1, \ldots ,n$, denote $C(u,v_j)$
as the event in which $v_j$ proposes to $u \in U$. As $R_{t}$ denotes the
unmatched nodes after the vertices $v_{1}, \ldots , v_{t-1}$ are processed
by Algorithm \ref{alg:ROM_edge_weights}, observe that
$u \in R_{t}$ if and only if $\neg C(u,v_j)$ occurs for each $j=1, \ldots , t-1$,
and so $\mb{P}[ u \in R_{t} ] = \mb{P}[\cap_{j=1}^{t-1} \neg C(u,v_j)]$.
We therefore focus on lower bounding $\mb{P}[\cap_{j=1}^{t-1} \neg C(u,v_j) ]$ in order to prove
the lemma.

First observe that for $j=1, \ldots , \lfloor \alpha n \rfloor$, the algorithm passes on probing
$\partial(v_j)$ by definition, and so \eqref{eqn:conditional_available} holds
if $t = \lceil \alpha n \rceil$. We may thereby
assume that $t \ge \lceil \alpha n \rceil +1$ and focus on lower bounding $\mb{P}[\cap_{j= \lceil \alpha n \rceil}^{t-1} \neg C(u,v_j)]$. Observe that this event depends only on the probes of the vertices of $V_{t-1} \setminus V_{\lfloor \alpha n \rfloor}$.
We denote $\bar{t}:= t - 1 - \lfloor \alpha n \rfloor = t - \lceil \alpha n \rceil$ as the number of vertices within this set.

Let us first consider the vertex $v_{t-1}$,
and the edge variable $\til{x}^{(t-1)}_{u,v}$ for each $v \in V_{t-1}$.
Observe that after applying Lemma \ref{lem:fixed_vertex_probe},
\begin{align*}
	\mb{P}[ C(u,v_{t-1})] &= \sum_{ v \in V_{t-1}} \mb{P}[ C(u,v_{t-1})  \, | \, v_{t-1} = v] \cdot \mb{P}[v_{t-1}=v] \\
											  &= \frac{1}{t-1} \sum_{v \in V_{t-1}} \til{x}_{u,v}^{(t-1)}  p_{u,v},
\end{align*}
as once we condition on the set $V_{t}$
and the vertex $v_t$, we know that $v_{t-1}$ is uniformly distributed amongst $V_{t-1}$.
On the other hand, the values $(\til{x}_{u,v}^{(t-1)})_{u \in U, v \in V_{t-1}}$ are derived
from a solution to \ref{LP:new} for $G_{t-1}$, and so by constraint \eqref{eqn:relaxation_efficiency_matching},
$
	\sum_{v \in V_{t-1}} \til{x}_{u,v}^{(t-1)}  p_{u,v} \le 1.
$
We therefore get that $\mb{P}[ C(u,v_{t-1}) ]  \le \frac{1}{t-1}$.
Similarly, if we fix $1 \le k \le \bar{t}$, then we can generalize the above
argument by conditioning
on the identities of all the vertices preceding $v_{t-k}$, as well as the probes
they make; that is, $(u_{t-1},v_{t-1}), \ldots ,(u_{t-(k-1)},v_{t-(k-1)})$ (in addition to $V_t$ and $v_t$ as always). 

In order to simplify the resulting indices, let us reorder the vertices of $V_{t-1} \setminus V_{\lfloor \alpha n \rfloor }$.
Specifically, define $\bar{v}_{k}:=v_{t-k}, \bar{u}_k := u_{t-k}$ and $\bar{e}_k:=e_{t-k}$
for $k=1,\ldots ,\bar{t}$. With this notation, we denote
$\scr{H}_{k}$ as encoding the information available based on the vertices $\bar{v}_1, \ldots , \bar{v}_{k}$ and the proposals they
(potentially) made. Formally, $\scr{H}_k$ is the sigma-algebra
	generated from $V_t,v_t$ and $\bar{e}_{1}, \ldots ,\bar{e}_{k}$. By convention,
we define $\scr{H}_{0}$ as the sigma-algebra generated from $V_t$ and $v_t$. 

An analogous computation to the above case then implies that
\[
	\mb{P}[ C(u,\bar{v}_{k})  \, | \, \scr{H}_{k-1}] = \sum_{ v \in V_{t-k}} \til{x}_{u,v}^{(t-k)}  p_{u,v}  \mb{P}[\bar{v}_k=v] \le \frac{1}{t-k},
\]
for each $k=1,\ldots , \bar{t}$, where
$\til{x}_{u,v}^{(t-k)}$ is the edge variable for $v \in V_{t-k}$.

Observe now that in each step, we condition on strictly more information;
that is, $\scr{H}_{k-1} \subseteq \scr{H}_{k}$
for each $k=2, \ldots , \bar{t}$. On the other hand, observe that
if we condition on $\scr{H}_{k-1}$ for $1 \le k \le \bar{t}-1$,
then the event $C(u,\bar{v}_{j})$ can be determined from $\scr{H}_{k-1}$
for each $1 \le j \le k-1$.

Using these observations, for $1 \le k \le \bar{t}$, the following recursion holds:
\begin{align*}
\mb{P}[ \cap_{j=1}^{k} \neg C(u,\bar{v}_j)] 
  &= \mb{E}\left[ \, \mb{E}\left[ \prod_{j=1}^{k} \bm{1}_{ [\neg C(u,\bar{v}_j)]}  \, | \, \scr{H}_{k-1}\right] \right]	\\
&= \mb{E}\left[ \,  \prod_{j=1}^{k-1} \bm{1}_{ [\neg C(u,\bar{v}_j)]} \, \mb{P}[\neg C(u, \bar{v}_k)  \, | \, \scr{H}_{k-1}]\right] \\
&\ge \left(1 - \frac{1}{t-k}\right) \mb{P}[ \cap_{j=1}^{k-1} \neg C(u,\bar{v}_j)]
\end{align*}
It follows that if $k= \bar{t}= t - \lceil \alpha n \rceil$, then applying the above recursion implies that
\[
	\mb{P}[\cap_{j= \lceil \alpha n \rceil}^{t-1} \neg C(u,v_j)] \ge \prod_{k=1}^{t- \lceil \alpha n \rceil} \left(1 - \frac{1}{t-k}\right).
\]
Thus, after cancelling the pairwise products,
$
	 \mb{P}[ u \in R_{t}] = \mb{P}[\cap_{j= \lceil \alpha n \rceil}^{t-1}\neg C(u,v_j)] \ge  \frac{\lfloor \alpha n \rfloor}{t-1},
$
and so \eqref{eqn:conditional_available} holds for all $t \ge \lceil \alpha n \rceil$, thereby completing the argument.
\end{proof}

\begin{proof}[Proof of  \Cref{lem:adv_feasibility}]
We must show that for each fixed $u_0 \in U$ and $v_0 \in V$,
\begin{equation}\label{eqn:dual_feasibility_adv}
	\mb{E}[ p_{u_0,v_0} \cdot \alpha_{u_0} +   w_{u_0} \cdot p_{u_0,v_0} \, \sum_{\substack{R \subseteq U: \\ u_0 \in R}} \phi_{v_0,R}] \ge w_{u_0} \cdot p_{u_0,v_0}.
\end{equation}
Let $R_{v_0}$ be the unmatched vertices of $U$ when $v_0$ arrives (this is a random subset of $U$).
Moreover, define $\til{E}:= E \setminus \partial(v_0)$. We claim the following:
\[
\sum_{\substack{R \subseteq U: \\ u_{0} \in R}} \mb{E}[ \phi_{v_{0},R} \, | \, (\sta(e))_{e \in \til{E}}] = \bm{1}_{[u_{0} \in R_{v_0}]}.
\]
To see this, observe that if we take a \textit{fixed} subset $R \subseteq U$, then the charging assignment to $\phi_{v_0,R}$
ensures that
\[
	\phi_{v_0,R} = w( \scr{M}(v_0)) \cdot \frac{1}{\OPT(v_0,R)} \cdot \bm{1}_{[ R_{v_0} = R]},
\]
where $w( \scr{M}(v_0))$ corresponds to the weight of the vertex matched to $v_0$ (which is zero if $v_0$ remains unmatched after the execution on $G$). In order to make use of this relation,
let us first condition on $(\sta(e))_{e \in \til{E}}$. Observe
that once we condition on this information, we can determine $R_{v_0}$. As such,
\[
	\mb{E}[ \phi_{v_0, R} \, | \, (\sta(e))_{e \in \til{E}}] = \frac{1}{\OPT(v_0,R)} \, \mb{E}[w( \scr{M}(v_0)) \, | \, (\sta(e))_{e \in \til{E}}] \cdot  \bm{1}_{[ R_{v_0} = R]}.
\]
On the other hand, the only randomness which remains in 
the conditional expectation involving $w(\scr{M}(v_0))$ is 
over $(\sta(e))_{e \in \partial(v_0)}$. However, 
since Algorithm \ref{alg:dynamical_program} behaves optimally
on $G[ \{v_{0}\} \cup R_{v_{0}}]$,
we get that
\begin{equation}\label{eqn:DP_optimality_adv}
	\mb{E}[w( \scr{M}(v_0)) \, | \,  (\sta(e))_{e \in \til{E}}] = \OPT(v_0,R_{v_0}),
\end{equation}
and so for the \textit{fixed} subset $R \subseteq U$,
$
	\mb{E}[w( \scr{M}(v_0)) \, | \, (\sta(e))_{e \in \til{E}}] \cdot \bm{1}_{[ R_{v_0} = R]} = \OPT(v_0,R) \cdot\bm{1}_{[ R_{v_0} = R]}
$
after multiplying each side of \eqref{eqn:DP_optimality_adv} by the indicator random variable $\bm{1}_{[R_{v_{0}}=R]}$.
Thus, 
$
	\mb{E}[\phi_{v_0,R}\, | \, (\sta(e))_{e \in \til{E}}] =  \bm{1}_{[ R_{v_0} = R]},
$
after cancellation. We therefore get that
$
	\sum_{\substack{R \subseteq U: \\ u_{0} \in R}} \mb{E}[ \phi_{v_{0},R} \, | \, (\sta(e))_{e \in \til{E}}] 
	= \sum_{\substack{R \subseteq U: \\ u_{0} \in R}} \bm{1}_{[ R_{v_0} = R]} = \bm{1}_{[u_{0} \in R_{v_0}]},
$
as claimed. On the other hand, if we focus on the vertex $u_{0}$, then observe
that if $u_0 \notin R_{v_0}$, then $\alpha_{u_0}$ must have been
charged $w_{u_0}$. In other words, $\alpha_{u_0} \ge w_{u_0} \cdot \bm{1}_{[u_{0} \notin R_{v_0}]}$.
As a result,
\[
	\mb{E}[ p_{u_0,v_0}  \alpha_{u_0}+   w_{u_0}  p_{u_0,v_0}\sum_{\substack{R \subseteq U: \\ u_0 \in R}} \phi_{v,R} \, | \, (\sta(e))_{e \in \til{E}}] \ge w_{u_0} p_{u_0,v_0}\cdot \bm{1}_{[u_{0} \notin R_{v_0}]}  + w_{u_0} p_{u_0,v_0} \cdot \bm{1}_{[u_{0} \in R_{v_0}]},
\]
and so \eqref{eqn:dual_feasibility_adv} follows after taking expectations.
The solution $((\alpha^*_u)_{u \in U}, (\phi^*_{v,R})_{v \in V, R \subseteq U})$ is therefore feasible.
\end{proof}

\end{document}